\documentclass[a4paper,reqno,11pt]{amsart}

\usepackage[hmargin=2.5cm,vmargin=2.5cm]{geometry}

\usepackage{amssymb,amsthm,mathtools,cancel,mathrsfs,xcolor,mathdots,graphicx,enumitem,tikz}

\usetikzlibrary{patterns, decorations.markings,arrows, calc}

\usepackage[english]{babel}

\usepackage[colorlinks=true, pdfstartview=FitV, urlcolor=blue, citecolor=red, linkcolor=blue, unicode]{hyperref}
\usepackage{stmaryrd}
\SetSymbolFont{stmry}{bold}{U}{stmry}{m}{n} 

\usepackage[T1]{fontenc}
\usepackage[landscape=true]{pdflscape}

\def\Z{\mathbb{Z}}
\def\N{\mathbb{N}}
\def\C{\mathbb{C}}

\def\wt{\widetilde}
\def\wh{\widehat}
\def\d{\mathrm{d}}
\def\e{\mathrm{e}}
\let\idotless\i
\renewcommand{\i}{\mathrm{i}}

\def\Partitions{\mathscr{P}}

\def\Vir{\mathsf{Vir}}

\def\Fermionic{\mathscr{F}}
\def\Bosonic{\mathscr{B}}

\def\la{\lambda}
\def\ep{\varepsilon}
\def\Wr#1{\mathrm{Wr}_{#1}\,}
\def\Hb{\mathscr H^{(\beta)}}
\def\Hbp{\mathscr H_{\mathcal{P}}^{(\beta)}}

\def\bs{\boldsymbol}

\def\res{\mathop{\mathrm {res}}\limits_}

\renewcommand{\log}{\ln}

\def\be{\begin{equation}}
\def\ee{\end{equation}}

\newtheorem{theorem}{Theorem}[section]
\newtheorem{conjecture}[theorem]{Conjecture}
\newtheorem{definition}[theorem]{Definition}
\newtheorem{lemma}[theorem]{Lemma}
\newtheorem{proposition}[theorem]{Proposition} 
\newtheorem{corollary}[theorem]{Corollary}

\theoremstyle{remark}
\newtheorem{remark}[theorem]{Remark}
\newtheorem{example}[theorem]{Example}
\newtheorem{question}[theorem]{Question}

\allowdisplaybreaks

\begin{document}

\numberwithin{equation}{section}

\title[ODE/IM Correspondence at the Free-Fermion Point]{
ODE/IM Correspondence at the Free-Fermion Point.\\
Laguerre Wronskians, Shifted Symmetric Functions, and Quantum KdV
}

\author{Davide Masoero}
\address{Grupo de Física Matemática, Instituto Superior Técnico, Universidade de Lisboa, Av. Rovisco Pais, 1049-001, Lisboa, Portugal}
\address
{Centro de Estudos Florestais, Instituto Superior de Agronomia, Universidade de Lisboa, Tapada da Ajuda, 1349-017, Lisboa, Portugal}
\email{davidem@isa.ulisboa.pt}

\author{Giulio Ruzza}
\address{CEMS.UL, Departamento de Matem\'atica, Faculdade de Ci\^encias da Universidade de Lisboa, Campo Grande Edif\'{\idotless}cio C6, 1749-016, Lisboa, Portugal}
\email{gruzza@ciencias.ulisboa.pt}

\date{}

\begin{abstract}
We consider the ODE/IM correspondence for the value $c=-2$ of the Virasoro central charge (free-fermion point) and the associated quantum KdV model --- the quantization of the second hamiltonian structure of the classical periodic KdV model.
We prove that the ODE/IM correspondence is complete (in the sense of V.~Bazhanov, S.~Lukyanov, and A.~Zamolodchikov), namely that any solution of the Bethe equations coincides with the spectrum of a rational extension of the (quantum) harmonic oscillator.

To this end, on the ODE side we consider Crum--Darboux transformations of the harmonic oscillator and the associated Laguerre Wronskians, which are remarkable special functions parametrized by pairs of partitions which we study in depth.

As a further result, on the IM side, we diagonalize explicitly the first three hamiltonian operators of quantum KdV (in the free field representation): the eigenstates are Schur functions and the eigenvalues are shifted symmetric functions on partitions.
We give two applications of this result: i) we prove that the eigenvalues are given by the evaluation of the Newton symmetric polynomials at the poles of the associated monster potentials, as further conjectured by V.~Bazhanov, S.~Lukyanov, and A.~Zamolodchikov; ii) we show that these hamiltonian operators also belong to the algebra of hamiltonian operators obtained by quantizing the first hamiltonian structure of the classical periodic dispersionless KdV model.
\end{abstract}

\subjclass[2020]{
82B23; 
81R10; 
05A17; 
34M46. 
}
\keywords{ODE/IM Correspondence; Bethe Equations; Virasoro Algebra; Quantum KdV; Partitions; Shifted Symmetric Functions; Laguerre Polynomials}

\maketitle

\setcounter{tocdepth}{1}

{
\hypersetup{linkcolor=black}
\tableofcontents
}

\section{Introduction and Results}

The ODE/IM correspondence is a remarkable and largely conjectural duality between certain linear ordinary differential operators (ODE side) and certain integrable quantum field theory solvable by Bethe equations (IM side).
It originated from the discovery of P.~Dorey and R.~Tateo~\cite{dorey99} that the spectral determinants of the anharmonic oscillators
\begin{equation}
   \mathscr{L}\,=\, -   \frac{\d^2}{\d x^2}+x^{2\alpha} , \qquad \alpha>0
\end{equation}
yield the ground-state solution of the Bethe equations for the quantum KdV model, built on the Virasoro Verma module with central charge $c=1-\frac{6\alpha^2}{\alpha+1}$ and highest weight $h=-\frac{\alpha^2}{4(\alpha+1)}$.
Here and in what follows, by quantum KdV model we mean the model obtained by quantizing the \textit{second} (or ``Lenard--Magri'') hamiltonian structure of classical KdV~\cite{BLZ1}, unless otherwise explicitly stated.

Since then, the ODE/IM correspondence has been extended to many more ODE/IM pairs, and it nowadays constitutes a cornerstone of integrable quantum field theory; see e.g. \cite{FF11,marava15,mara18,fioravanti23,gaiotto22,lukyanov10,KL}.

Our paper stems from the work~\cite{BLZHigher} of V.~Bazhanov, S.~Lukyanov, and A.~Zamolodchikov, who extended the Dorey--Tateo conjecture (building also on \cite{doreytateo98,bazhanov01}) to include every value of the momentum and all states of the quantum model.
They showed that the spectral determinants of the operator
\begin{equation}\label{eq:monsterL}
\mathscr{L}\,=\,-\frac{\d^2}{\d x^2}+ x^{2\alpha}+\frac{\beta^2-\frac14}{x^2}-2 \frac{\d^2}{\d x^2} \ln\mathcal{P}\big(x^{2\alpha+2}\big)  , 
\qquad 
\beta \in \C,
\end{equation}
with $\mathcal{P}$ a monic polynomial, are solutions to the Bethe equations for quantum KdV with $c=1-\frac{6\alpha^2}{\alpha+1}$ and $h=\frac{\beta^2-\alpha^2}{4(\alpha+1)}$, provided that the general solution of the corresponding stationary Schr\"odinger equation $\mathscr {L} \psi=E \psi$ has trivial monodromy about all poles of $\frac{\d^2}{\d x^2} \ln\mathcal{P}\big(x^{2\alpha+2}\big)$ for every value of~$E \in \C$.
The potentials of such Schr\"odinger operators are called \textit{monster potentials}.

In the same work~\cite{BLZHigher}, the authors conjectured (also providing several analytical and numerical checks) that the ODE/IM correspondence is \textit{complete}, namely, that every solution of the Bethe equations for the quantum KdV model can be constructed in this way.
The conjecture was verified by R.~Conti and the first author for $\alpha>1$ and $\beta$ large \cite{coma20,coma21} (see also \cite{mamura})\footnote{Their proof is conditional on a technical lemma which was only later partially addressed in~\cite{felder26}.}.

In the present paper, we prove this conjecture for~$\alpha=1$ and every value of~$\beta \in \C$.
The corresponding value of the Virasoro central charge is $c=-2$ and so this setting corresponds to the ODE/IM correspondence for the quantum KdV model at the free-fermion point~\cite{AdlerBoos,doreytateo98,KotusovLukyanovShabetnik}.

More specifically, our main results can be summarized as follows.
\begin{enumerate}[leftmargin=*]
    \item We prove that the ODE/IM correspondence at the free-fermion point is complete: every solution of the QQ system coincides with the spectral determinant of the operator \eqref{eq:monsterL} with $\alpha=1$ for some $\mathcal P$ obtained via a Crum--Darboux transformation; see Theorem \ref{thm:completeness}.
    \item We prove that the number of monster potentials (obtained via Crum--Darboux transformations) with $\deg \mathcal{P}=n$ is less than or equal to~$p(n)$, with $p(n)$ the number of partitions of $n$, and that equality holds if and only if the Verma module with $c=-2$ and $h=\frac{\beta^2-1}{8}$, is irreducible in degree $n$; see Theorem~\ref{thm:irreducible?}.
    \item We diagonalize explicitly the first three quantum KdV hamiltonians at $c=-2$ and we show that their eigenvalues (which we express explicitly as \textit{shifted symmetric} functions on partitions) essentially coincide with the evaluation of the Newton symmetric polynomials at the poles of the monster potentials, as conjectured in~\cite{BLZHigher}; see Corollary~\ref{cor:sumzi} and Remark~\ref{rem:sumzi}. 
    \item As a by-product, we show that first three quantum KdV hamiltonians at $c=-2$ belong to the algebra of commuting hamiltonians of the quantum KdV model obtained by quantizing the \textit{first} (or ``Gardner--Zakharov--Faddeev'') hamiltonian structure of the \textit{dispersionless} classical KdV~\cite{BuryakRossi,Dubrovin,Eliashberg,Pogrebkov,Rossi}.
    This seems to be the first result in the literature comparing the two quantum deformations of classical KdV.
\end{enumerate}

The main tools that we use in our analysis are \textit{Laguerre Wronskians}, namely, Wronskians of functions of the form
\be\label{eq:psiintro}
\psi_\pm^{(\beta)}(n,x)\,=\,\e^{-\frac 12 x^2}\,x^{\frac 12\pm\beta}\,\wh{L}_{n}^{(\pm\beta)}\bigl(x^2\bigr),
\ee
where $\wh{L}_{n}^{(\pm\beta)}\bigl(x^2\bigr)$ is the $n$th monic generalized Laguerre polynomial. In this generality, Laguerre Wronskians have already been considered in the literature; for example, see~\cite{BonneuxKuijlaars,Duran,GUKM,OS,Quesne}.

The introduction of these Wronskians is natural because $\psi_\pm^{(\beta)}(n,x)$ are eigenfunctions of the operator
\begin{equation}\label{eq:Hbintro}
\mathscr{H}^{(\beta)}\,=\,
-  \frac{\d^2}{\d x^2}+ x^{2}+\frac{\beta^2-\frac14}{x^2}
\end{equation}
and their logarithmic derivative is a meromorphic function with a finite number of poles.
Hence, they can be used to generate monster potentials at $\alpha=1$ via Crum--Darboux transformations.
Indeed, as it is well known, iterated application of Darboux transformations on~\eqref{eq:Hbintro} generates an operator of the form
\begin{equation}
-\frac{\d^2}{\d x^2}+ x^{2}+\frac{\beta^2-\frac14}{x^2} -2 \frac{\d^2}{\d x^2} \ln W(x)
\end{equation}
where $W(x)$ is the Wronskian of the eigenfuctions used and so, in our case, it is a Laguerre Wronskian.

Laguerre Wronskians are naturally parametrized by a pair of partitions $(\la,\mu)$ and they have a factorization akin to \eqref{eq:psiintro}: an exponential times a power of $x$ times an even polynomial of $x$ which does not vanish at $0$ for generic $\beta$, see Theorem~\ref{thm:propertiesWronskianLaguerre}.
We denote by $\Phi_{\la,\mu}^{(\beta)}(y)$ with $y=x^2$ the latter polynomial.
We obtain several new results on these special functions: 
\begin{enumerate}[leftmargin=*]
    \item We compute the action of the transformation $y \mapsto -y$ (i.e. $x \mapsto \i x$), that is, we show that $\Phi_{\la,\mu}^{(\beta)}(-y)=(-1)^{|\lambda|+|\mu|} \Phi_{\mu',\lambda'}^{(\beta)}(y)$, where $\lambda',\mu'$ are the conjugate partitions of $\lambda,\mu$; see Theorem \ref{thm:symmetry}.
    Moreover, assuming that $\Phi_{\la,\mu}^{(\beta)}(0) \neq 0$ (which is true for generic $\beta$), we have $\Phi_{\la,\mu}^{(\beta)}(-y)= \Phi_{\lambda,\mu}^{(\beta)}(y)$ if and only if $\mu$ is the conjugate partition of~$\lambda$; see Theorem \ref{thm:lamuextension}(VI). 
    \item  As we already mentioned, $\Phi_{\la,\mu}^{(\beta)}(0) \neq 0$ for generic values of $\beta$.
    We characterize explicitly, in terms of standard combinatorics of partitions (2-cores and 2-quotients), the set of $\beta$ such that $\Phi_{\la,\mu}^{(\beta)}(0)=0$; see Theorem~\ref{cor:coalescence} and Proposition~\ref{prop:Bn=Cn}.
    This is a key result to understand how the number of monster potentials (obtained by Crum--Darboux transformations), which is generically equal to the number of partitions of~$n$, drops to a smaller number for such critical values of~$\beta$. 
    \item We address and solve the problem: are the polynomials $\Phi_{\wt\la,\wt\mu}^{(\wt \beta)}$ and $ \Phi_{\la,\mu}^{(\beta)}$ distinct if $(\beta,\la,\mu) \neq (\wt \beta,\wt\la,\wt\mu)$? We show, for example, that if $\beta \notin \Z$, the only other polynomial identical to $\Phi_{\la,\mu}^{( \beta)}(y)$ is the one obtained via the symmetry $(\beta,\la,\mu)\to (-\beta,\mu,\la)$; see Theorem \ref{thm:lamuextension}(V).
    Specializing in the case $\beta=\pm\frac12$, this proves a conjecture in~\cite{BDS} on Wronskians of Hermite polynomials
    \item We compute explicitly the coefficients of the first and second sub-leading coefficients of the polynomial $\Phi^{(\beta)}_{\la,\mu}(y)$; see Theorem~\ref{thm:propertiesWronskianLaguerre}.
    Specializing to the symmetric case in which $\mu$ is the conjugate partition of $\lambda$, we express these  coefficients in terms of shifted symmetric functions on partitions; see Proposition~\ref{prop:sumzi}.
\end{enumerate}

\subsection*{Methods and plan of the paper}

Sections~\ref{sec:laguerre},~\ref{sec:wrsym}, and~\ref{sec:coalescence} are devoted to the study of Laguerre Wronskians, deferring the proofs to Appendix~\ref{app:propertiesLaguerre}. 
In Section~\ref{sec:laguerre} we study the general structure of Laguerre Wronskians and their parametrization by pairs of partitions. 
Although Laguerre Wronskians have already appeared extensively in the literature (see, e.g.,~\cite{BonneuxKuijlaars,Duran,GUKM,OS,Quesne}), we felt that a comprehensive and self-contained discussion of their main properties was necessary.
Section ~\ref{sec:wrsym} is devoted to the symmetric case in which $\mu$ is the conjugate partition of $\lambda$.
In Section ~\ref{sec:coalescence}, we address the phenomenon of coalescence of roots at $0$, that is, we explicitly describe the values of $\beta$ such that $\Phi_{\la,\mu}^{(\beta)}(0)=0$.

Next, in Section~\ref{sec:spectra} we study rational extensions of the harmonic oscillator and their spectral problems. We compute explicitly the spectral determinants (and the spectra) for every rational extension obtained via Crum--Darboux transformations (i.e. via Laguerre Wronskians). These results allow us to classify the monster potentials (obtained via Crum--Darboux transformations) and to understand under which conditions $\Phi_{\wt\la,\wt\mu}^{(\wt \beta)}$ and $ \Phi_{\la,\mu}^{(\beta)}$ are distinct if $(\beta,\la,\mu) \neq (\wt \beta,\wt\la,\wt\mu)$.
Building on this, in Section \ref{sec:ODEIM} we prove that the ODE/IM correspondence is complete, and we compare the number of monster potentials of degree $n$ with the dimension of the degree $n$ part of the irreducible highest weight Virasoro module with $c=-2$ and $h=\frac{\beta^2-1}{8}$.

Finally, in Section~\ref{sec:qKdV} we explicitly diagonalize the action in the free field representation of the first three hamiltonian operators of the quantum KdV model with $c=-2$.
The basis of eigenstates is given by Schur functions and eigenvalues are given explicitly in terms of shifted symmetric functions on partitions. The latter are remarkable functions on partitions playing an important role in asymptotic representation theory and enumerative geometry (see Section~\ref{sec:wrsym} for more details and references), and also appearing naturally in the spectral problem of the quantization of KdV with respect to the \textit{first} hamiltonian structure~\cite{Dubrovin,vIR1,vIR2,RY}.
Thus, their appearance here is a manifestation of a relation between the two (in principle, distinct~\cite{KupershmidtMathieu}) quantizations of the KdV model associated with its two hamiltonian structures; see Remark~\ref{rem:comparisonhamiltonianstructures}.
Moreover, the expression in terms of shifted symmetric functions allows us to make a direct comparison with the formulas derived in Sections~\ref{sec:laguerre} and~\ref{sec:wrsym} for Laguerre Wronskians and so to prove that eigenvalues can be expressed in terms of the evaluation of the Newton symmetric polynomials at the poles of the associated monster potentials, as first conjectured in~\cite{BLZHigher}; see Corollary~\ref{cor:sumzi} and Remark~\ref{rem:sumzi}.

\subsection*{Open problems and future directions}
For the reader's convenience, we collect here a number of problems scattered throughout the text that our work leaves open.

\begin{itemize}[leftmargin=*]
    \item One knows from A.~Oblomkov \cite{Oblomkov} (after a conjecture by A.~Veselov) that when $\beta=\pm\frac12$ every rational extension of the harmonic oscillator is obtained via Crum--Darboux transformations. However, when $\beta \neq \pm\frac12$ this is false, since there exist rational extensions that are not symmetric under the transformation $x \mapsto -x$.
    Oblomkov's theorem suggests the following question: are all rational extensions symmetric under the transformation $x \mapsto -x$ obtained via Crum--Darboux transformations? See Conjecture~\ref{conj:oblo}.
    \item  We prove that the number of monster potentials of degree $n$ and momentum $\beta$ is equal to the dimension of the degree $n$ part of the irreducible highest weight Virasoro module with $c=-2$ and $h=\frac{\beta^2-1}{8}$, when $\beta \notin \Z$.
    In Section~\ref{sec:newconjecture} we conjecture that when~$\beta\in\mathbb Z$ the number of monster potentials of degree~$n$ and momentum~$\beta$ is equal to the number of quantum KdV common eigenvectors in the Virasoro Verma module which are not null vectors (see~Conjecture~\ref{conjecture:new} for the precise statement). It is worth observing that the quantum KdV hamiltonians need not be semisimple on the Virasoro Verma module when $\beta\in\mathbb Z$.
    \item Does the Newton polynomial of degree $k$ evaluated at the poles of monster potentials belong to the algebra of shifted symmetric functions on partitions? See Question \ref{quest:shifted}. Moreover, since shifted symmetric functions on partitions have already been shown to appear naturally in the spectral problem of the quantization of KdV with respect to the \textit{first} Poisson structure~\cite{vIR2}, one may speculate whether shifted symmetric functions play an analogous role here also beyond the $c=-2$ case.
    \item We have explicitly diagonalized the first three quantum KdV hamiltonians at $c=-2$, see Theorem \ref{thm:diagonalonSchur}. The results of~\cite{DiFrancescoMathieuSenechal} strongly suggest that these computations can be extended to the full family of quantum KdV hamiltonians at $c=-2$. Also in view of the previous point, we defer this investigation to future work.
    \item The ODE/IM correspondence has been formulated as a precise mathematical conjecture for the generalized $\mathfrak{g}$-quantum KdV model, with $\mathfrak{g}$ an affine Lie algebra \cite{FH16,marava17}. In particular, the analog of the monster potentials were defined in~\cite{mara18} in the case $\mathfrak{g}$ is the untwisted affinization of a simply-laced Lie algebra.
    Is it possible to generalize our construction to these models? What would be the analog of the free-fermion point and of the Crum--Darboux transformations?
\end{itemize}

\subsection*{Notation}
We collect here some notation used throughout this paper.

\begin{itemize}[leftmargin=*]
\item We denote by $\mathbb N=\lbrace 0,1,2,\ldots\rbrace$ the set of natural numbers (including $0$) and by $\mathbb N^*=\mathbb N\setminus\lbrace 0\rbrace$ the set of positive natural numbers.
For any $m\in\mathbb N^*$, we denote $[m]=\lbrace 1,2,\dots,m\rbrace$.
    
\item 
A \textit{partition} is an infinite sequence $\lambda=(\lambda_1,\lambda_2,\dots)$ of integers $\lambda_i$ such that $\lambda_i\geq \lambda_{i+1}$ for all $i\geq 1$ and $\lambda_i=0$ for some $i$.
We write $\Partitions$ for the set of all partitions.

\noindent The \textit{size} of $\lambda\in\Partitions$ is $|\lambda|=\sum_{i\geq 1}\lambda_i$ and the \textit{length} of $\lambda$ is $\ell(\lambda)=\min\lbrace i\geq 0: \lambda_{i+1}=0\rbrace$. If $|\lambda|=n$ we say that $\lambda$ is a partition of $n$.
We denote the number of partitions of $n$ by~$p(n)$.

\noindent The \textit{conjugate} partition of $\lambda\in\Partitions$ is $\lambda'=(\lambda'_1,\lambda'_2,\dots)$ where $\lambda'_i \,=\, \#\lbrace j\in\Z_{\geq 1}:\ \lambda_j\geq i\rbrace$. We say that $\lambda$ is \textit{self-conjugate} if $\lambda=\lambda'$.

\item $[y^k]f(y)$ is the coefficient of $y^k$ in the Taylor/Laurent series of $f(y)$ at $y=0$.

\item $(z)_m=\textstyle\prod_{j=0}^{m-1}(z+j)$ is the rising Pochhammer symbol.

\item $|\bs{v}|=\sum_{1=1}^k|v_i|$, for any $\bs v=(v_1,\ldots,v_k)$.

\item  $(\bs v,\bs w)=(v_1,\dots,v_k,w_1,\dots,w_h)$ and $a\bs v+b=(a v_1+b,\dots,a v_k+b)$ for vectors $\bs v=(v_1,\dots,v_k)$, $\bs w=(w_1,\dots,w_h)$ and scalars $a,b$.
    
\item $\Delta(\bs{v})=\prod_{1\leq i<j\leq k}(v_j-v_i)$, for any $\bs v=(v_1,\ldots,v_k)$, is the Vandermonde determinant.

\item $\Wr x \bigl(f_1(x),\dots,f_k(x)\bigr)=\det\bigl(\partial_x^{a-1}f_b(x)\bigr)_{a,b=1}^k$ is the Wronskian determinant.
\end{itemize}

\subsection*{Acknowledgments}

D.M. acknowledges financial support from the grant UID/PRR/00208/2025 (https://doi.org/10.54499/UID/PRR/00208/2025) funded by  Fundação para a Ciência e a Tecnologia (FCT)  and União Europeia - Estrutura de Missão Recuperar Portugal (UE - EMRP), and from the grant UID/00208/2025 funded by Fundação para a Ciência e a Tecnologia (FCT). 

G.R. acknowledges the support of the Center for Mathematical Studies of the Faculty of Sciences of the University of Lisbon (FCT – Portuguese national funding, UID/04561/2025).

We are grateful to Giordano Cotti, Gabriele Degano, Jan-Willem van Ittersum, Andrea Raimondo, Roberto Tateo, Don Zagier for useful conversations.

\subsection*{Data availability statement.}
Data sharing not applicable to this article as no datasets were generated or analysed during the current study.

\subsection*{Statement on the use of artificial intelligence}
AI tools were used only for proofreading, literature searches, and occasional assistance in drafting computational code. The AI-assisted portions of the code were independently checked by the authors.

\section{Laguerre Wronskians} \label{sec:laguerre}

\subsection{Laguerre polynomials}

Let ${L}_n^{(\beta)}(y)$ be the standard \textit{generalized Laguerre polynomials} (the unique family of orthogonal polynomials with respect to the measure $y^\beta\e^{-y}\d y$ on $(0,+\infty)$ with leading coefficient $(-1)^n/n!$) and let $\wh{L}_n^{(\beta)}(y)=(-1)^n\, n! \,L_n^{(\beta)}(y)$ be their monic normalization.
Their coefficients are explicitly given by
\be
\label{eq:LaguerreExplicit}
c_n^{[i]}(\beta)\,=\,[y^i]\,\wh{L}_n^{(\beta)}(y)\,=\,(-1)^{n+i}\binom ni (\beta+i+1)_{n-i}\qquad (0\leq i\leq n).
\ee

Introduce, for $n\in\N$ and $\beta\in\C$,
\be
\psi_\pm^{(\beta)}(n,x)\,=\,\e^{-\frac 12 x^2}\,x^{\frac 12\pm\beta}\,\wh{L}_{n}^{(\pm\beta)}\bigl(x^2\bigr).
\ee
\begin{remark}\label{remark:LaguerrePolynomials}
\begin{enumerate}[leftmargin=*]
    \item The functions $\psi_\pm^{(\beta)}(n,x)$ are eigenfunctions of the operator~$\mathscr H^{(\beta)}$ defined in~\eqref{eq:Hbintro},
    namely,
    \be \label{eq:eigenharm}
    \mathscr H^{(\beta)} \,\psi^{(\beta)}_\pm(n,x) \,=\, E_{\pm}^{(\beta)}(n)\,\psi_\pm^{(\beta)}(n,x),
    \quad\text{where}\quad
    E_\pm^{(\beta)}(n) \, = \,2(2n+1\pm\beta),
    \ee
    for all $\in\N$, $\beta\in\C$.
    The operator $\mathscr{H}^{(\beta)}$ is the $\beta$-harmonic oscillator, namely, the radial part of the Schr\"odinger operator for the isotropic harmonic oscillator with momentum $\beta+\frac 12$.
    
    \item When $\beta=\pm \frac 12$, the functions $\psi_\pm^{(\beta)}(n,x)$ reduce to Hermite functions: for all $n\in\N$,
\be
\label{eq:Hermite}
\psi_\pm^{(\pm\frac 12)}(n,x)\,=\,\e^{-\frac 12 x^2}\,\wh{H}_{2n+1}(x),\quad
\psi_\pm^{(\mp\frac 12)}(n,x)\,=\,\e^{-\frac 12 x^2}\,\wh{H}_{2n}(x),
\ee
where $\wh{H}_n(x)$ is the $n$th monic Hermite polynomial. See also Section~\ref{sec:HermiteReduction}.

    \item If $E_+^{(\beta)}(m)=E_-^{(\beta)}(n)$ for some $m,n\in\N$ (equivalently, if $\beta=n-m$ for some $m,n\in\N$) then $\psi_+^{(\beta)}(m,x)=\psi_-^{(\beta)}(n,x)$.
This follows from the identity
\be
\label{eq:identityLaguerrecoalescence}
\frac{(-y)^{m}}{m!}L^{(m-n)}_{n}(y)=\frac{(-y)^{n}}{n!}L^{(n-m)}_{m}(y),\qquad m,n\in\Z.
\ee

\end{enumerate}
\end{remark}

\subsection{Laguerre Wronskians}

For $r\in\N$, let 
\be
\mathscr V_r \,=\, \bigl\lbrace\bs m=(m_1,\dots,m_r)\in\N^r : m_1>m_2>\dots >m_r\bigr\rbrace.
\ee
Given $r,s\in\N$, $\bs m\in\mathscr V_{r}$ and $\bs n\in\mathscr V_{s}$ we define
\be \label{eq:Psimnax}
\Psi_{\bs m,\bs n}^{(\beta)}(x)=
\Wr x \bigl(\psi_+^{(\beta)}(m_1,x),\dots,\psi_+^{(\beta)}(m_r,x),\psi_-^{(\beta)}(n_1,x),\dots,\psi_-^{(\beta)}(n_s,x)\bigr).
\ee
The structure of the functions $\Psi_{\bs m,\bs n}^{(\beta)}(x)$ is clarified in the next proposition.
To state it, for any $\bs m\in\mathscr V_{r}$ and $\bs n\in\mathscr V_{s}$ we define
\be\label{eq:kappamna}
\kappa_{\bs m,\bs n}^{(\beta)}=\Delta(2\bs m+\beta,2\bs n-\beta)=2^{\binom{r+s}2}\Delta(\bs m)\Delta(\bs n)\prod_{i=1}^{r}\prod_{j=1}^{s}(n_j-m_i-\beta).
\ee

\begin{theorem}
\label{thm:propertiesWronskianLaguerre} 
For all $r,s\in\N$, $\bs m\in\mathscr V_{r}$, and $\bs n\in\mathscr V_{s}$, we have
\be
\label{eq:PsiPhi}
\Psi^{(\beta)}_{\bs m,\bs n}(x)\,=\,\kappa_{\bs m,\bs n}^{(\beta)}\,\e^{-\frac12(r+s)x^2}\,x^{\frac 12(r-s)^2+\beta(r-s)}\,\widetilde{\Phi}_{\bs m,\bs n}^{(\beta)}\bigl(x^2\bigr)
\ee
where $\widetilde{\Phi}_{\bs m,\bs n}^{(\beta)}(y)\in\mathbb Z[\beta,y]$ is a monic polynomial in $y$ of degree
\be
\label{eq:degree}
d=|\bs m|+|\bs n|-\binom r2-\binom s2.
\ee
More precisely, we have
\be
\label{eq:phik}
\widetilde{\Phi}_{\bs m,\bs n}^{(\beta)}(y)=\sum_{k=0}^{d}\widetilde{\phi}^{[k]}_{\bs m,\bs n}(\beta)\,y^{k}
\ee
where each $\widetilde{\phi}^{[k]}_{\bs m,\bs n}(\beta)\in\Z[\beta]$ has degree $\leq d-k$ in~$\beta$.
Moreover, the constant term is
\be
\label{eq:phi0}
\widetilde{\phi}^{[0]}_{\bs m,\bs n}(\beta)=(-1)^d\prod_{i=1}^r(\beta+i)_{m_i+1-i}\prod_{j=1}^s(-\beta+j)_{n_j+1-j}\prod_{i=1}^{r}\prod_{j=1}^{s} \frac{j-i-\beta}{n_j-m_i-\beta}
\ee
and the three leading coefficients are $\widetilde{\phi}^{[d]}_{\bs m,\bs n}(\beta)=1$,
\be
\label{eq:subleading1}
\widetilde{\phi}^{[d-1]}_{\bs m,\bs n}(\beta)\,=\,-\rho_{2}^{+}+(N-1)\rho_{1}^{+}-\binom{N}{3}-\beta\bigl(\rho_{1}^{-}-R+S \bigr)\,,
\ee
and
\be
\label{eq:subleading2}
\begin{aligned}
\widetilde{\phi}^{[d-2]}_{\bs m,\bs n}(\beta)\,&=\,\frac{(\rho_2^+)^2}{2}-\rho_3^+-(N-1) \rho_2^+\rho_1^+
+\frac{(N-1)(N^2-2 N+9)}{6}\rho_2^++\frac{N^2-2N+2}{2} (\rho_1^+)^2
\\
&\quad
-\frac{(N-1)^2(N^2-2N+6)}{6}\rho_1^++\frac{N(N-1)^2(N-2) (N^2-2N+9)}{72}
\\
&\quad
+\beta\Biggl(\rho_2^+\rho_1^--\frac{3}{2}\rho_2^-
-(N-1)\bigl(\rho_1^+\rho_1^-+\frac{r-s}{2}\rho_2^+\bigr)
+\frac{(N^2-2N+2)(r-s)}{2}\rho_1^+
\\
&\qquad
+\frac{(N-1)(N^2-2N+9)}{6}\rho_1^-
-\frac {(N^2-2N+6)(N-1)^2(r-s)}{12}
\Biggr)
\\
&\quad+\beta^2\Biggl(\frac{\bigl(\rho_1^-\bigr)^2}{2}-\frac{\rho_1^+}{2}-\big(R-S\bigr)\rho_1^-
+\frac {R(R+1)+S(S+1)}2-RS\Biggr)\,,
\end{aligned}
\ee
where
\be
N\,=\,r+s\,,\quad \rho_{k}^{\pm}\,=\,\sum_{i=1}^rm_i^k\,\pm\,\sum_{j=1}^sn_j^k\,,\quad R\,=\,\binom r2,\quad S\,=\,\binom s2\,.
\ee
\end{theorem}

For clarity of exposition, the proof is deferred to Appendix~\ref{app:propertiesLaguerre}, see Section~\ref{sec:proofstructure}.

\begin{remark}
\label{remark:constantterm}
While not immediately evident from~\eqref{eq:phi0}, the coefficient $\widetilde{\phi}^{[0]}_{\bs m,\bs n}(\beta)$, like all coefficients $\widetilde{\phi}^{[k]}_{\bs m,\bs n}(\beta)$, belongs to $\Z[\beta]$.
A denominator-free formula for $\widetilde{\phi}^{[0]}_{\bs m,\bs n}(\beta)$ requires more combinatorial work with partitions and will be derived in Section~\ref{sec:denfreep0}.
Nevertheless, \eqref{eq:phi0} already implies that all the roots of $\widetilde{\phi}^{[0]}_{\bs m,\bs n}(\beta)$ in $\beta$ are integers; these integers will be fully characterized by the aforementioned denominator-free formula.
\end{remark}

\begin{remark}\label{rem:4vs2}
One can concoct more general Wronskian determinants combining eigenfunctions of~\eqref{eq:Hbintro} of four kinds, rather than just the two kinds $\psi_+^{(\beta)}(x)$ and $\psi_-^{(\beta)}(x)$.
We refer to the exhaustive discussion by N.~Bonneux and A.~Kuijlaars~\cite{BonneuxKuijlaars} about this point.
In particular, in Theorem~4.2 of \textit{op. cit.} it is shown that the reduction to two kinds does not cause any loss of generality, as a general Wronskian involving all four kinds of eigenfunctions can always be reduced to a Wronskian involving only two.
\end{remark}

\subsection{Parametrization by partitions}

Given $r\in\N$ we introduce the vector
\be
{\bs\delta}_r\,=\,(r-1\,,\,r-2\,,\,\dots\,,\,1\,,\,0)\,\in\,\mathscr V_r,
\ee
agreeing that $\bs\delta_0$ is the empty vector.
Given a partition $\lambda\in\Partitions$ and an integer $r\geq \ell(\lambda)$, we define
\be
\lambda+{\bs\delta}_r \,=\,(\lambda_1+r-1\,,\,\lambda_2+r-2\,,\,\dots\,,\,\lambda_{r-1}+1\,,\,\lambda_r) \,\in\,\mathscr V_r\,.
\ee
Every $\bs m\in\mathscr V_r$ can be uniquely written as $\lambda+{\bs\delta}_r$ for some $\lambda\in\Partitions$ with $\ell(\lambda)\leq r$.

\begin{theorem}
\label{thm:equivalence}
For all $\lambda,\mu\in\Partitions$, the polynomial $\widetilde{\Phi}_{\lambda+{\bs\delta}_r,\mu+{\bs\delta}_s}^{(\beta-r+s)}
(y)$ depends only on $\beta,\lambda,\mu$ and not on $r,s$, as long as $r\geq \ell(\lambda)$ and $s\geq\ell(\mu)$.
\end{theorem}

For clarity of exposition, the proof is deferred to Appendix~\ref{app:propertiesLaguerre}, see Section~\ref{sec:proofequivalence}.

\begin{definition}
\label{def:PhiLaMu}
    For all $\lambda,\mu\in\Partitions$ we define
    \be
    \label{eq:defPlambdamu}
    \Phi_{\lambda,\mu}^{(\beta)}(y)\,=\,\widetilde{\Phi}_{\lambda+{\bs\delta}_r,\mu+{\bs\delta}_s}^{(\beta-r+s)}(y)
    \ee
    where $r,s$ are any integers satisfying $r\geq \ell(\lambda)$ and $s\geq\ell(\mu)$.
\end{definition}

We note that $\Phi_{\lambda,\mu}^{(\beta)}(y)$ is a polynomial of degree $|\lambda|+|\mu|$ in $y$, see~\eqref{eq:degree}.

\begin{remark}
\label{remark:equivalentextensions}
We observe that, essentially by definition, $\wt\Phi_{\bs m,\bs n}^{(\beta)}(y)=\wt\Phi_{\bs n,\bs m}^{(-\beta)}(y)$ for all $\bs m\in\mathscr V_r$ and $\bs n\in\mathscr V_s$. It follows that $\Phi_{\lambda,\mu}^{(\beta)}(y)=\Phi_{\mu,\lambda}^{(-\beta)}(y)$ for all $\lambda,\mu\in\Partitions$.
\end{remark}

\section{Laguerre Wronskians: Symmetric Case}\label{sec:wrsym}

\begin{theorem}
    \label{thm:symmetry}
For all $\lambda,\mu\in\Partitions$ we have
\be
\Phi_{\lambda',\mu'}^{(\beta)}(y)\,=\,(-1)^{|\lambda|+|\mu|}\,\Phi_{\mu,\lambda}^{(\beta)}(-y).
\ee
In particular, $\Phi^{(\beta)}_{\lambda,\lambda'}(y)$ is an even polynomial in $y$.
\end{theorem}

For clarity of exposition, the proof is deferred to Appendix~\ref{app:propertiesLaguerre}, see Section~\ref{sec:proofsymmetry}.

\begin{remark}
In Theorem~\ref{thm:lamuextension}(V), we will prove the converse of the last assertion in Theorem~\ref{thm:symmetry}, namely, we will prove that (for general values of $\beta$) if $\Phi_{\lambda,\mu}^{(\beta)}(y)$ is an even polynomial in $y$ then $\mu=\lambda'$.
This generalizes and settles~\cite[Conjecture~1]{BDS}.
\end{remark}

\begin{remark}
Similar symmetry properties have been considered in~\cite{BonneuxKuijlaars,CurberaDuran} (with entirely different proofs).
\end{remark}

\begin{definition}
\label{def:Pla}
    For any $\lambda\in\Partitions$, let $\mathcal P^{(\beta)}_\lambda(z)$ be the polynomial of degree $|\lambda|$ in $z$ defined by
    \be
        \Phi_{\lambda,\lambda'}^{(\beta)}(y)\,=\,\mathcal P^{(\beta)}_\lambda(y^2)\,.
    \ee
    For any $\lambda\in\Partitions$, let $z_1^{(\beta)}(\lambda),\dots,z_{|\lambda|}^{(\beta)}(\lambda)$ be the roots of $\mathcal P^{(\beta)}_\lambda(z)$:
    \be
        \mathcal P^{(\beta)}_\lambda(z)\,=\,\prod_{i=1}^{|\lambda|}\Bigl(z-z_i^{(\beta)}(\lambda)\Bigr)\,.
    \ee
\end{definition}

In order to conveniently state the next result, we need to recall the definition of shifted symmetric functions on partitions, a remarkable class of functions $\Partitions\to\mathbb Q$ appearing in asymptotic representation theory \cite{KerovOlshanski} and in enumerative geometry (for instance, in the Hurwitz/Gromov--Witten correspondence of A.~Okounkov and R.~Pandharipande and in the computation of volumes and Siegel--Veech constants of the moduli space of flat surfaces \cite{ChenMollerSauvagetZagier,EskinOkounkov,HahnIttersumLeid,OkounkovPandharipande}).
They are also relevant in view of a remarkable relation to quasimodular forms discovered by S.~Bloch and A.~Okounkov~\cite{BO,Zagier} and have already been shown to appear in the spectral problem of the quantum KdV model obtained by quantization of the \textit{first} hamiltonian structure~\cite{Dubrovin,vIR1,vIR2,RY}.

\begin{definition}[Shifted symmetric functions]\label{def:SS}
The functions $\mathsf{p}_k:\Partitions\to\mathbb Q$, for $k\geq 0$, are defined by $\mathsf{p}_0(\lambda)=1$ and
\be
\label{eq:defP}
\mathsf{p}_k(\lambda)\,=\,\sum_{i\geq 1}\Bigl(\bigl(\lambda_i-i+\tfrac 12\bigr)^{k-1}-\bigl(-i+\tfrac 12\bigr)^{k-1}\Bigr)\qquad (k\geq 1).
\ee
More generally, the $\mathbb{Q}$-algebra  $\Lambda^*=\mathbb{Q}[\mathsf{p}_2,\mathsf{p}_3,\ldots]$ of shifted symmetric functions on partitions is the algebra of functions $\Partitions\to\mathbb Q$ which can be expressed as polynomials in the $\mathsf{p}_k$.
\end{definition}

\begin{proposition}\label{prop:sumzi}
The sum of the roots of $\mathcal P_{\lambda}^{(\beta)}(z)$ is (for any given $\beta$) a shifted symmetric function of $\lambda$.
More precisely, for any $\lambda\in\Partitions$ we have
\be
\label{eq:sumzi}
\sum_{i=1}^{|\lambda|}z_i^{(\beta)}(\lambda) \,=\, 
2 \mathsf{p}_4(\lambda)-2 \mathsf{p}_2(\lambda)^2+3\beta \mathsf{p}_3(\lambda)+\beta^2 \mathsf{p}_2(\lambda)
+\frac{1}{2} \mathsf{p}_2(\lambda).
\ee
\end{proposition}
\begin{proof}
    Since $\sum_{i=1}^{|\lambda|}z_i^{(\beta)}(\lambda) =-\widetilde{\phi}^{[d-2]}_{\lambda+\bs\delta_r,\lambda'+\bs\delta_s}(\beta-r+s)$ for any $r\geq \ell(\lambda)$ and $s\geq\ell(\lambda')=\lambda_1$, the proposition follows by specializing identity~\eqref{eq:subleading2} to $\bs m=\lambda+\bs\delta_r$ and $\bs n=\lambda'+\bs\delta_s$.
    This is done by straightforward algebra based on the following identities valid in this case:
    \be
    \label{eq:substsymmetric}
    \begin{aligned}
    \rho_1^-&=\binom r2-\binom s2,\qquad\qquad
    \rho_1^+=2\mathsf{p}_2(\lambda)+\binom r2+\binom s2,\\
    \rho_2^-&=2\mathsf p_3(\lambda)+2(r-s)\mathsf p_2(\lambda)+\frac{r(r-1)(2r-1)}{6}-\frac{s(s-1)(2s-1)}{6},\\
    \rho_2^+&=2(r+s-1)\mathsf p_2(\lambda)+\frac{r(r-1)(2r-1)}{6}+\frac{s(s-1)(2s-1)}{6},\\
    \rho_3^+&=2\mathsf p_4(\lambda)+3(r-s)\mathsf{p}_3(\lambda)+\frac{3}{4}\bigl((2r-1)^2+(2s-1)^2\bigr)\mathsf{p}_2(\lambda)+\frac {r^2(r-1)^2}{4}+\frac {s^2(s-1)^2}{4}.
    \end{aligned}
    \ee
    The latter identities follow from the following ones, valid for any partition $\lambda$ and any $r\geq\ell(\lambda)$,
    \be
    \begin{aligned}
    \sum_{i=1}^r(\lambda_i+r-i)&=\mathsf{p}_2(\lambda)+\binom r2,\\
    \sum_{i=1}^r(\lambda_i+r-i)^2&=\mathsf{p}_3(\lambda)+(2r-1)\mathsf{p}_2(\lambda)+\frac{r(r-1)(2r-1)}{6},\\
    \sum_{i=1}^r(\lambda_i+r-i)^3&=\mathsf{p}_4(\lambda)+\frac 32(2r-1)\mathsf{p}_3(\lambda) +\frac 34(2r-1)^2 \mathsf{p}_2(\lambda)+\frac {r^2(r-1)^2}{4},
    \end{aligned}
    \ee
    and from the symmetry relation $\mathsf p_k(\lambda') \,=\, (-1)^k\,\mathsf p_k(\lambda)$.
    Substitution of~\eqref{eq:substsymmetric} into the formula~\eqref{eq:subleading2} for $-\widetilde{\phi}^{[d-2]}_{\lambda+\bs\delta_r,\lambda'+\bs\delta_s}(\beta-r+s)$ completes the proof.
\end{proof}

We conjecture that (omitting the dependence on $\lambda$)
\be
\label{eq:sumzisquared}
\begin{aligned}
\sum_{i=1}^{|\lambda|}\bigl(z_i^{(\beta)}\bigr)^2 \,=\,& 
\biggl(14 \mathsf{p}_6-40 \mathsf{p}_4 \mathsf{p}_2+\frac{64}{3} \mathsf{p}_2^3+\frac{65}{3} \mathsf{p}_4-20 \mathsf{p}_2^2+\frac{27}{8} \mathsf{p}_2\biggr)+\beta  \biggl(35 \mathsf{p}_5-60 \mathsf{p}_3 \mathsf{p}_2+\frac{65}{2} \mathsf{p}_3\biggr)\\
&+\beta ^2 \biggl(30 \mathsf{p}_4-22 \mathsf{p}_2^2+\frac{25}{2} \mathsf{p}_2\biggr)+10 \beta ^3 \mathsf{p}_3+\beta^4 \mathsf{p}_2\,.
\end{aligned}
\ee
This conjecture is based on numerical experiments and is verified for all $\lambda\in\Partitions$ with $|\lambda|\leq 15$.
Although a proof using the methods used in Proposition~\ref{prop:sumzi} seems, in principle, possible, we did not carry out the necessary computations.
Formula~\eqref{eq:sumzi} and the conjectural formula~\eqref{eq:sumzisquared}, together with the obvious identity $\sum_{i=1}^{|\lambda|}\bigl(z^{(\beta)}_i(\lambda)\bigr)^0=|\lambda|=\mathsf{p}_2(\lambda)$, suggest the following question, to which we hope to return in future work.
Recall the $\mathbb Q$-algebra $\Lambda^*$ of shifted symmetric functions on partitions from Definition~\ref{def:SS} and introduce a grading on $\Lambda^*[\beta]$ by assigning $\deg \beta=1$ and $\deg \mathsf{p}_k=k$.
For any $k\in\mathbb{N}$, let $\Lambda^*[\beta]_{\leq k}$ be the subspace of elements of degree at most~$k$.

\begin{question}\label{quest:shifted}
Is it true for all $k\in\N$ that $\,\sum_{i=1}^{|\lambda|}\bigl(z_i^{(\beta)}(\lambda)\bigr)^k\ \in \ \Lambda^*[\beta]_{\leq 2k+2}$ ?
\end{question}

\section{Coalescence of Zeros of Laguerre Wronskians at the Origin}\label{sec:coalescence}

\subsection{Some combinatorics of partitions}

We now review a number of standard notions in the theory of partitions which will be useful in our analysis.
For more details we refer the reader to the first chapter of~\cite{Macdonald}.

\subsubsection{Diagrams, border strips, and hooks}
For $\Lambda=(\Lambda_1,\Lambda_2,\dots)\in\Partitions$, the \textit{diagram} of $\Lambda$
\be
D(\Lambda)\,=\,\lbrace (i,j)\in\mathbb Z^2:\ i\geq 1,\ 1\leq j\leq \Lambda_i\rbrace
\ee
and it is customary to refer to elements of $D(\lambda)$ as \textit{cells}.
The \textit{border diagram} is the subset
\be
B(\Lambda)\,=\,\lbrace (i,j)\in D(\Lambda):\ (i+1,j+1)\not\in D(\Lambda)\rbrace\,.
\ee
A \textit{hook} of $\Lambda$ is a set of the form
\be
H_{ij}(\Lambda) \,=\, \lbrace (i ,j' )\in D(\Lambda):\ j'\geq j\rbrace\cup\lbrace (i',j)\in D(\Lambda):\ i'\geq i\rbrace\,.
\ee
for some $(i,j)\in D(\Lambda)$. A \textit{border strip} of $\Lambda$ is a set of the form
\be
B_{ij}(\Lambda) = \lbrace (i' ,j' )\in B(\Lambda):\ i' \geq i,\ j'\geq j\rbrace
\ee
for some $(i,j)\in D(\Lambda)$.
The cardinalities of the hook $H_{ij}(\Lambda)$ and of the border strip $B_{ij}(\Lambda)$ are equal and coincide with the \textit{hook length}
\be
\label{eq:defhooklength}
h_{ij}(\Lambda)\,=\,\Lambda_i-i+\Lambda_j'-j+1.
\ee
For any $(i,j)\in D(\Lambda)$, the set $D(\Lambda)\setminus B_{ij}(\Lambda)$ is the diagram $D(\wt\Lambda)$ of another partition $\wt\Lambda$, in which case we say that $\wt\Lambda$ is obtained from $\Lambda$ by \textit{removing a border strip}.
In such a case and if $L\geq\ell(\Lambda)$, the vector $\wt{\bs N}=\wt\Lambda+{\bs\delta}_L$ is obtained from $\bs N=\Lambda+{\bs\delta}_L$ by subtracting $h_{ij}(\Lambda)$ from $N_i$ and re-arranging the components in decreasing order.
Conversely, if the vector $\wt{\bs N}=\wt\Lambda+{\bs\delta}_L$ is obtained from $\bs N=\Lambda+{\bs\delta}_L$ by subtracting $q$ from some part $N_i$ and re-arranging the components in decreasing order, then $\wt\Lambda$ is obtained from $\Lambda$ by removing a border strip $B_{ij}(\Lambda)$ of size $q=h_{ij}(\Lambda)$.

\begin{example}
The diagram of the partition $\Lambda=(4,3,1,1)$ can be depicted as
\[
\begin{tikzpicture}[scale=0.4]
\def\cellsize{1}

\foreach \r/\c in {0/2,0/3,1/1,1/2} {
    \fill[color=gray!30] (\c, -\r) rectangle ++(\cellsize,-\cellsize);
}

\foreach \r/\c in {0/1,1/1,0/2,0/3} {
    \fill[pattern=
    vertical lines, pattern color=red!80] (\c, -\r) rectangle ++(\cellsize,-\cellsize);
}

\foreach \r/\c in {0/0,0/1,0/2,0/3} {
    \draw (\c, -\r) rectangle ++(\cellsize,-\cellsize);
}

\foreach \r/\c in {1/0,1/1,1/2} {
    \draw (\c, -\r) rectangle ++(\cellsize,-\cellsize);
}

\foreach \r/\c in {2/0} {
    \draw (\c, -\r) rectangle ++(\cellsize,-\cellsize);
}

\foreach \r/\c in {3/0} {
    \draw (\c, -\r) rectangle ++(\cellsize,-\cellsize);
}

\node at (0.5, -0.5) {7};
\node at (1.5, -0.5) {4};
\node at (2.5, -0.5) {3};
\node at (3.5, -0.5) {1};

\node at (0.5, -1.5) {5};
\node at (1.5, -1.5) {2};
\node at (2.5, -1.5) {1};

\node at (0.5, -2.5) {2};

\node at (0.5, -3.5) {1};

\end{tikzpicture}
\]
where we also indicated the hook length of each cell as well as the border strip (in gray) and hook (in red) of the cell $(1,2)$ (the unique cell with hook length~4).
\end{example}

\subsubsection{2-cores and 2-quotients}
Let $\Lambda\in \Partitions$, $L\geq \ell(\Lambda)$, and $\bs N=\Lambda+{\bs\delta}_L$.
Let $\bs N$ be a permutation of the vector $(2\bs m,2\bs n+1)$ for some $\bs m\in\mathscr V_r$ and $\bs n\in\mathscr V_s$.
If $L$ is even, the ordered pair of partitions $(\lambda,\mu)$ such that $\bs m=\lambda+{\bs\delta}_r$ and $\bs n=\mu+{\bs\delta}_s$ is called the \textit{2-quotient} of $\Lambda$.
It only depends on $\Lambda$, as long as $L$ is even.
(If $L$ is odd the same procedure would yield the ordered pair of partitions $(\mu,\lambda)$.)

Let $\wt\Lambda$ be obtained from $\Lambda$ by removing a border strip of size~$2$.
Let $\sigma$ be the permutation such that $\bs N=\sigma(2\bs n+1,2\bs m)$.
Then, the vector $\wt{\bs N}=\wt\Lambda+{\bs\delta}_L$ is obtained by subtracting $2$ from one of the components of $\bs N$ and, if and only if the border strip is \textit{vertical}, swapping two components of $\bs N$.
In other words, subtracting $1$ from one of the components of either $\bs m$ or $\bs n$, we obtain vectors $\wt{\bs m},\wt{\bs n}$ such that $\wt{\bs N} = \sigma_1\sigma(2\wt{\bs n}+1,2\wt{\bs m})$, where $\sigma_1$ is either the identity (if the border strip removed is horizontal) or a transposition (if the border strip removed is vertical).

If we iteratively remove border strips of size~$2$ from $\Lambda$ until no further removal is possible, we reach a partition $\overline\Lambda$, termed the \textit{2-core} of $\Lambda$.
It is independent of the specific sequence of border strips removed and it is characterized by the fact that $\overline\Lambda+{\bs\delta}_L$ is a permutation of the vector $(2{\bs\delta}_s+1,2{\bs\delta}_r)$.
From this, it is easy to see that $\overline\Lambda$ is the (``staircase'') partition $(c,c-1,\dots,1,0,0,\dots)$, with $c=r-s-1$ if $r>s$ or $c=s-r$ if $s\geq r$.
In both cases,
\be
\label{eq:weigthcore}
|\overline\Lambda|=\binom {c+1}2 =\binom {r-s}2=\binom{s-r+1}2.
\ee
Moreover, if $v_\Lambda$ is the number of vertical border strips of size 2 to be removed from $\Lambda$ in order to get $\overline\Lambda$, we have
\be
\label{eq:signp0}
(-1)^{v_\Lambda}=\varepsilon\,\mathrm{sign}\left(\begin{smallmatrix}
        2{\bs\delta}_s+1,2{\bs\delta}_r \\
        \overline\Lambda+{\bs\delta}_{L}
\end{smallmatrix}\right)=\begin{cases}
    \varepsilon(-1)^{\binom{s+1}2+sr}&(r>s)\\
    \varepsilon(-1)^{\binom r2}&(s\geq r)\\
\end{cases}
\ee
where $\varepsilon$ is the sign of $\sigma$.
Hence, the parity of $v_\Lambda$ depends only on $\Lambda$ and not on $r,s$ nor on the specific sequence of border strips removed.

\begin{example}
The 2-quotient and 2-core of the partition $\Lambda=(4,3,1,1)$ are $\bigl(\lambda=(1),\mu=(1,1)\bigr)$ and $(2,1)$, and $v_\Lambda$ is odd.
\end{example}

\begin{remark}[Re-constructing a partition from its 2-quotient and 2-core]\label{remark:constructionLambda}
Given an integer $c\geq 0$ and a pair of partitions $\lambda,\mu$ there is a unique partition $\Lambda$ such that the 2-quotient of $\Lambda$ is $(\lambda,\mu)$ and the 2-core of $\Lambda$ is $(c,c-1,\dots,1,0,0,\dots)$.
Indeed, we can always find $r,s\in\N$ such that $r\geq\ell(\lambda)$, $s\geq \ell(\mu)$, $L=r+s$ is even, and $c=r-s-1$ (if $c$ is odd) or $c=s-r$ (if $c$ is even). (In particular, when $c=0$, we have to take $r=s$.) Then, the partition sought is $\Lambda$ such that $\Lambda+{\bs\delta}_L$ is a permutation of $\bigl(2(\lambda+{\bs\delta}_r),2(\mu+{\bs\delta}_s)+1\bigr)$.
\end{remark}

\subsection{Denominator-free expression for the constant term}\label{sec:denfreep0}
In this section we derive an explicit formula for the constant term $\phi^{[0]}_{\lambda,\mu}(\beta)=\Phi^{(\beta)}_{\lambda,\mu}(0)$ of the polynomial $\Phi_{\lambda,\mu}^{(\beta)}(y)$ which, opposite to formula~\eqref{eq:phi0}, makes it manifest that it is a polynomial of $\beta$ (with integer roots which we also characterize explicitly).
We start by recalling the following well-known formula for products of functions of hook lengths (we direct the reader to the first chapter of~\cite{Macdonald} for more details).

\begin{lemma}
\label{lemma:hookproductgeneral}
Let $\Lambda$ be a partition, $L\geq \ell(\Lambda)$, and $\bs N=\Lambda+{\bs\delta}_L$.
If $f_1,\dots,f_L$ are functions on the integers taking values in any commutative ring, we have
\be
\prod_{(i,j)\in D(\Lambda)}f_i\bigl(h_{ij}(\Lambda)\bigr)=
\frac{\prod_{a=1}^L\prod_{k=1}^{N_a}f_a(k)}{\prod_{1\leq a<b\leq L}f_a(N_a-N_b)}.
\ee
\end{lemma}

When $f_i(s)=s$ for all $i$, we get the well-known formula for the product of hook lengths.
We need a different application of this lemma. Define, for any $\Lambda\in\Partitions$,
\be
D^{\mathrm{odd}}(\Lambda)=\lbrace(i,j)\in D(\Lambda):\ h_{ij}(\Lambda)\ \mathrm{is\ odd}\rbrace\,.
\ee

\begin{proposition}
\label{prop:oddhookproducttemp}
Let $\Lambda$ be a partition, $L\geq \ell(\Lambda)$, and $\bs N=\Lambda+{\bs\delta}_L$.
Let $\bs m\in\mathscr V_r$ and $\bs n\in\mathscr V_s$ (with $r+s=L$) be such that $(2\bs n+1,2\bs m)$ is a permutation of $\bs N$ and let $\varepsilon$ be the sign of this permutation.
Then,
\be
\prod_{(i,j)\in D^{\mathrm{odd}}(\Lambda)}\biggl(\frac {h_{ij}(\Lambda)}2+(-1)^{N_i}\biggl(\frac 12+\beta\biggr)\biggr)=\varepsilon\,\frac{\prod_{i=1}^{r}(\beta+1)_{m_i}\,\prod_{j=1}^{s}(-\beta)_{n_j+1}}{\prod_{i=1}^r\prod_{j=1}^s(n_j-m_i-\beta)}.
\ee
Under the same hypotheses and assuming $d$ is as in~\eqref{eq:degree}, if $s\geq r$ we have
\be
\widetilde{\phi}^{[0]}_{\bs m,\bs n}(\beta)=\varepsilon\,(-1)^{d+\frac 12r(r-1)}\,\frac{\prod_{(i,j)\in D^{\mathrm{odd}}(\Lambda)} \left(\frac {h_{ij}(\Lambda)}2+(-1)^{\Lambda_i-i+r+s}\left(\frac 12+\beta\right)\right)}{\prod_{j=1}^{s-r}(-\beta)_j}
\ee
and if $r>s$ we have
\be
\widetilde{\phi}^{[0]}_{\bs m,\bs n}(\beta)=\varepsilon\,(-1)^{d+\frac 12s(s+1)+rs}\,\frac{\prod_{(i,j)\in D^{\mathrm{odd}}(\Lambda)} \left(\frac {h_{ij}(\Lambda)}2+(-1)^{\Lambda_i-i+r+s}\left(\frac 12+\beta\right)\right)}{\prod_{i=1}^{r-s-1}(\beta+1)_i}\,.
\ee
\end{proposition}
\begin{proof}
Apply Lemma~\ref{lemma:hookproductgeneral} with $f_i(k)=1$ if $k$ is even and $f_i(k)=\frac k2+(-1)^{N_i}\bigl(\frac 12+\beta\bigr)$ if $k$ is odd.
For the last claim, use~\eqref{eq:phi0} and the identities
\begin{align}
\label{eq:prodsgeqr}
\prod_{i=1}^r\prod_{j=1}^s(j-i-\beta)\,&=\,(-1)^{\frac 12r(r-1)}\,\frac{\prod_{i=1}^r(\beta+1)_{i-1}\,\prod_{j=1}^s(-\beta)_j}{\prod_{j=1}^{s-r}(-\beta)_{j}},&&s\geq r,
\\
\label{eq:prodrges}
\prod_{i=1}^r\prod_{j=1}^s(j-i-\beta)\,&=\,(-1)^{\frac 12s(s+1)+rs}\,\frac{\prod_{i=1}^r(\beta+1)_{i-1}\,\prod_{j=1}^s(-\beta)_j}{\prod_{i=1}^{r-s-1}(\beta+1)_{i}},&&r>s.
\end{align}
\end{proof}

These formulas can be expressed more neatly in terms of partitions, employing the notions of 2-quotient and 2-core of a partition we reviewed above.

\begin{proposition}
\label{prop:p0denfreeCORE}
Let the partition $\Lambda$ have 2-core~$\overline\Lambda=(c,c-1,\dots,1,0,\dots)$ and 2-quotient $(\lambda,\mu)$.
We have
\be
\label{eq:p0denfreeCORE}
\phi^{[0]}_{\lambda,\mu}(\beta)=(-1)^{|\lambda|+|\mu|+v_\Lambda}\,\frac{\prod_{(i,j)\in D^{\mathrm{odd}}(\Lambda)}\left(\frac {h_{ij}(\Lambda)}2+(-1)^{\Lambda_i-i}\left(\beta+(-1)^c\bigl(\frac 12+c\bigr)\right)\right)}{\prod_{(i,j)\in D(\overline\Lambda)}\left(\frac {h_{ij}(\overline\Lambda)}2+(-1)^{\overline\Lambda_i-i}\left(\beta+(-1)^c\bigl(\frac 12+c\bigr)\right)\right)}
\ee
where $v_\Lambda$ is the number of vertical border strips of size $2$ to be removed from $\Lambda$ to obtain $\overline\Lambda$ (whose parity is well-defined).
\end{proposition}

The right-hand side of~\eqref{eq:p0denfreeCORE} is independent of $c$, because the left-hand side is independent of~$c$.
By taking $c=0$ we obtain a denominator-free formula.

\begin{corollary}
\label{corollary:p0denfree}
Let $\lambda,\mu$ be two partitions. Then
\be
\phi^{[0]}_{\lambda,\mu}(\beta)=(-1)^{|\lambda|+|\mu|+v_\Lambda}\,\prod_{(i,j)\in D^{\mathrm{odd}}(\Lambda)}\left(\frac {h_{ij}(\Lambda)}2+(-1)^{\Lambda_i-i}\left(\frac 12+\beta\right)\right)
\ee
where $\Lambda$ is the unique partition with 2-quotient $(\lambda,\mu)$ and empty 2-core and $v_\Lambda$ is the number of vertical border strips of size $2$ to be removed from $\Lambda$ to obtain the empty partition (whose parity is well-defined).
\end{corollary}
\begin{proof}[Proof of Proposition~\ref{prop:p0denfreeCORE}]
Assume $c$ is even and let $r\geq\ell(\lambda)$, $s\geq \ell(\mu)$, with $r+s$ even and $s-r=c\geq 0$.
Then, if  $\bs m=\lambda+{\bs\delta}_r$, $\bs n=\mu+{\bs\delta}_s$ and $\wt\beta=\beta-r+s=\beta+c$, we have (by Proposition~\ref{prop:oddhookproducttemp})
\be
\phi_{\lambda,\mu}^{[0]}(\beta)
=\widetilde{\phi}^{[0]}_{\bs m,\bs n}(\wt\beta)
=\varepsilon(-1)^{d+\binom r2}\frac{\prod_{(i,j)\in D^{\mathrm{odd}}(\Lambda)} \left(\frac {h_{ij}(\Lambda)}2+(-1)^{\Lambda_i-i}\left(\frac 12+\wt\beta\right)\right)}{\prod_{j=1}^{c}(-\wt\beta)_j}
\ee
and then it suffices to use $d=|\lambda|+|\mu|$, cf.~\eqref{eq:degree}, $(-1)^{v_\Lambda}=\varepsilon(-1)^{\binom r2}$, cf.~\eqref{eq:signp0}, and the identity
\be
\prod_{j=1}^{c}(-\wt\beta)_j = \prod_{(i,j)\in D(\overline\Lambda)}\left(\frac {h_{ij}(\overline\Lambda)}2+(-1)^{\overline\Lambda_i-i}\left(\wt\beta+\frac 12\right)\right)
\ee
with $\overline\Lambda=(c,c-1,\dots,1,0,\dots)$ and $c$ even.

Assume $c$ is odd and let $r\geq\ell(\lambda)$, $s\geq \ell(\mu)$, with $r+s$ even and $r-s-1=c\geq 1$.
Then, if  $\bs m=\lambda+{\bs\delta}_r$, $\bs n=\mu+{\bs\delta}_s$ and $\wt\beta=\beta-r+s=\beta-c-1$, we have (by Proposition~\ref{prop:oddhookproducttemp})
\be
\phi_{\lambda,\mu}^{[0]}(\beta)
=\widetilde{\phi}^{[0]}_{\bs m,\bs n}(\wt\beta)=\varepsilon(-1)^{d+\binom{s+1}2+rs}\frac{\prod_{(i,j)\in D^{\mathrm{odd}}(\Lambda)} \left(\frac {h_{ij}(\Lambda)}2+(-1)^{\Lambda_i-i}\left(\frac 12+\wt\beta\right)\right)}{\prod_{i=1}^{c}(\wt\beta+1)_i}
\ee
and then it suffices to use $d=|\lambda|+|\mu|$, cf.~\eqref{eq:degree}, $(-1)^{v_\Lambda}=\varepsilon(-1)^{\binom{s+1}2+rs}$, cf.~\eqref{eq:signp0}, and the identity
\be
\prod_{i=1}^{c}(\wt\beta+1)_i = \prod_{(i,j)\in D(\overline\Lambda)}\left(\frac {h_{ij}(\overline\Lambda)}2+(-1)^{\overline\Lambda_i-i}\left(\wt\beta+\frac 12\right)\right)
\ee
with $\overline\Lambda=(c,c-1,\dots,1,0,\dots)$ and $c$ odd.
\end{proof}

\begin{example}
\label{ex:constant}
The unique partition $\Lambda$ with 2-quotient $(\lambda=(3,1),\mu=(2))$ and empty 2-core is $\Lambda=(5,5,1,1)$.
(This can be easily computed according to the recipe given in Remark~\ref{remark:constructionLambda}, taking $r=s=2$, such that $\bs m=\lambda+\bs\delta_2=(4,1)$ and $\bs n=\mu+\bs\delta_2=(3,0)$, yielding $2\bs m=(8,2)$ and $\bs n=(7,1)$ and so $\bs N=(8,7,2,1)=(5,5,1,1)+\bs\delta_4$.)
The diagram of $\Lambda$ and its odd hook lengths are
\[
\begin{tikzpicture}[scale=0.4, baseline=(current bounding box.north)]
\def\cellsize{1}

\foreach \r/\c/\val in {1/2/5,1/4/3} {
    \draw (\c, -\r) rectangle ++(\cellsize,-\cellsize);
    \node at (\c+0.5,-\r-0.5) {\val};
}
\foreach \r/\c in {1/1,1/3,1/5} {
    \draw (\c, -\r) rectangle ++(\cellsize,-\cellsize);
}

\foreach \r/\c/\val in {2/1/7,2/3/3,2/5/1} {
    \draw (\c, -\r) rectangle ++(\cellsize,-\cellsize);
    \node at (\c+0.5,-\r-0.5) {\val};
}
\foreach \r/\c in {2/2,2/4} {
    \draw (\c, -\r) rectangle ++(\cellsize,-\cellsize);
}

\draw (1,-3) rectangle ++(\cellsize,-\cellsize);

\draw (1,-4) rectangle ++(\cellsize,-\cellsize);
\node at (1.5,-4.5) {1};

\node[right=0.5cm] at (6.5,-1.5) {$+$};
\node[right=0.5cm] at (6.5,-2.5) {$-$};
\node[right=0.5cm] at (6.5,-3.5) {$+$};
\node[right=0.5cm] at (6.5,-4.5) {$-$};

\end{tikzpicture}
\]
where on the right we also indicate the sign $(-1)^{\Lambda_i-i}$ for each row.
According to Corollary~\ref{corollary:p0denfree}, noting that in this case $(-1)^{v_\Lambda}=1$,
we get
\be
\begin{aligned}
\phi^{[0]}_{\lambda,\mu}(\beta)&=(\tfrac 52+\tfrac 12+\beta)(\tfrac 32+\tfrac 12+\beta)
(\tfrac 72-\tfrac 12-\beta)(\tfrac 32-\tfrac 12-\beta)(\tfrac 12-\tfrac 12-\beta)
(\tfrac 12-\tfrac 12-\beta)
\\
&=(\beta+3)(\beta+2)\beta^2(\beta-1)(\beta-3).
\end{aligned}
\ee
\end{example}

\subsection{Hermite reduction}\label{sec:HermiteReduction}
Let $r,s\in\N$.
By~\eqref{eq:Hermite} and Theorem~\ref{thm:propertiesWronskianLaguerre} (with $\beta=-1/2$) we obtain that for all $\bs m\in\mathscr V_r$ and $\bs n\in\mathscr V_s$ Laguerre Wronskians reduce to Hermite Wronskians and we have (see also~\cite{BDS,FelderHemeryVeselov2012})
\be
   \frac{\Wr x\bigl(\wh H_{2m_1}(x),\dots,\wh H_{2m_r}(x),\wh H_{2n_1+1}(x),\dots,\wh H_{2n_s+1}(x)\bigr)}
{\Delta(2\bs m,2\bs n+1)}
=x^{\binom{r-s}2}\widetilde{\Phi}_{\bs m,\bs n}^{(-\frac 12)}(x^2)
\ee
where $\wh H_n(x)=2^{-n} H_n(x)$ are the monic Hermite polynomials.
Moreover, by~Proposition~\ref{prop:oddhookproducttemp} we obtain
\be
\widetilde{\Phi}_{\bs m,\bs n}^{(-\frac 12)}(0)=\frac{(-1)^{|\lambda|+|\mu|+v_\Lambda}}{2^{|\lambda|+|\mu|}}\frac{\prod_{(i,j)\in D^{\mathrm{odd}}(\Lambda)}h_{ij}(\Lambda)}{\prod_{(i,j)\in D(\overline\Lambda)}h_{ij}(\overline\Lambda)}
\ee
where $\Lambda$ is the unique partition such that $\Lambda+{\bs\delta}_{r+s}$ is a permutation of $(2\bs m,2 \bs n+1)$, $\overline\Lambda$ is the 2-core of $\Lambda$ and $v_\Lambda$ the number of vertical border strips to be removed from $\Lambda$ to obtain $\overline\Lambda$ (whose parity is well-defined).

Recalling that $|\overline\Lambda|=\binom{r-s}2$, cf.~\eqref{eq:weigthcore}, these formulas recover a result of N.~Bonneux, C.~Dunning, and M.~Stevens, namely~\cite[Theorem~1]{BDS}.
We note that our proof is entirely different (as it is not based on any recurrence but on the direct analysis of the relevant determinant).

Moreover, Theorem~\ref{thm:symmetry} recovers~\cite[Corollary~2]{BDS}.
We note that the converse of the last assertion in Theorem~\ref{thm:symmetry} is a generalization (from Hermite to Laguerre) of \cite[Conjecture~1]{BDS}.
We establish this converse assertion in Theorem~\ref{thm:lamuextension}(V), thus proving the conjecture of N.~Bonneux, C.~Dunning, and M.~Stevens in a more general setting.

\subsection{Coalescence of zeros of Laguerre Wronskians at the origin}

By Corollary~\ref{corollary:p0denfree}, the roots of $\phi^{[0]}_{\lambda,\mu}(\beta)$ are certain integers parametrized by cells in $D^{\mathrm{odd}}(\Lambda)$, where $\Lambda$ has 2-quotient $(\lambda,\mu)$ and trivial 2-core.
It turns out that when~$\beta$ equals one of these roots, the polynomial $\Phi_{\lambda,\mu}^{(\beta)}(y)$ is, up multiplication by a positive power of $y$, a polynomial of $y$ in the same family, namely, the polynomial $\Phi_{\wt\lambda,\wt\mu}^{(\beta)}(y)$ corresponding to the 2-quotient $(\wt\lambda,\wt\mu)$ of the partition obtained from $\Lambda$ by removing the border strip corresponding to the cell in $D^{\mathrm{odd}}(\Lambda)$ parametrizing the root~$\beta$.

In the interest of clarity, let us formalize this situation with the following definition.

\begin{definition}
For triples $(\beta,\lambda,\mu),(\wt\beta,\wt\lambda,\wt\mu)\in\C\times\Partitions\times\Partitions$, we write
\be
(\beta,\lambda,\mu) \xrightarrow{(i,j)}(\wt\beta,\wt \lambda,\wt\mu)
\ee
to indicate that:
\begin{itemize}[leftmargin=*]
    \item $(i,j)\in D^{\mathrm{odd}}(\Lambda)$ where $\Lambda$ is the unique partition with trivial 2-core and 2-quotient $(\lambda,\mu)$;
    \item $\beta=-\frac 12-(-1)^{\Lambda_i-i}\frac {h_{ij}(\Lambda)}2$ (and so $\phi_{\lambda,\mu}^{[0]}(\beta)=0$ by Corollary~\ref{corollary:p0denfree});
    \item $(\wt\lambda,\wt\mu)$ is the 2-quotient of the partition obtained from $\Lambda$ by removing the border strip $B_{ij}(\Lambda)$;
    \item $\wt\beta=\beta\pm 2$, with $+$ if $\beta\geq 0$ and $-$ if $\beta<0$.
\end{itemize}
\end{definition}

\begin{proposition}
\label{prop:coalescence}
If $(\beta,\lambda,\mu),(\wt\beta,\wt\lambda,\wt\mu)\in\C\times\Partitions\times\Partitions$ satisfy $(\beta,\lambda,\mu) \xrightarrow{(i,j)}(\wt\beta,\wt \lambda,\wt\mu)$, we have
\be
\Phi_{\lambda,\mu}^{(\beta)}(y)\,=\,y^{|\beta|+1}\,\Phi_{\wt\lambda,\wt\mu}^{(\wt\beta)}(y)\,.
\ee
\end{proposition}
\begin{proof}
By~\eqref{eq:identityLaguerrecoalescence}, for any $n,b\in\Z$ such that $n\geq 0$ and $n\geq b$ we have
\be
\psi_+^{(b)}(n-b,x)\,=\,\psi_-^{(b)}(n,x).
\ee
Hence, for all $b\in\Z$, $\bs m\in\mathscr V_r$, $\bs n\in\mathscr V_s$, and $1\leq k\leq s$ such that $n_k\geq b$ and $n_k-b\notin\bs m$, we have
\be
\Psi_{\bs m,\bs n}^{(b)}(x)\,=\,\pm\Psi_{\bs m\sqcup \lbrace n_k-b\rbrace,\bs n\setminus \lbrace n_k\rbrace}^{(b)}(x),
\ee
up to a sign which is inconsequential in our analysis, and so, according to~\eqref{eq:PsiPhi},
\be
\label{eq:coalescenceproof1}
\widetilde{\Phi}_{\bs m,\bs n}^{(b)}(y)\,=\,y^{b+r-s+1}\,\widetilde{\Phi}_{\bs m\sqcup \lbrace n_k-b\rbrace,\bs n\setminus \lbrace n_k\rbrace}^{(b)}(y).
\ee
Since $\widetilde{\Phi}_{\bs m,\bs n}^{(b)}(y)=\widetilde{\Phi}_{\bs n,\bs m}^{(-b)}(y)$, for all $b\in\Z$, $\bs m\in\mathscr V_r$, $\bs n\in\mathscr V_s$ and any $1\leq k\leq r$ such that $m_k\geq -b$ and $m_k+b\notin\bs n$, we have
\be
\label{eq:coalescenceproof2}
\widetilde{\Phi}_{\bs m,\bs n}^{(b)}(y)\,=\,y^{-b+s-r+1}\,\widetilde{\Phi}_{\bs m\setminus \lbrace m_k\rbrace,\bs n\sqcup \lbrace m_k+b\rbrace}^{(b)}(y).
\ee
Next, let $(\beta,\lambda,\mu),(\wt\beta,\wt\lambda,\wt\mu)\in\C\times\Partitions\times\Partitions$ satisfy $(\beta,\lambda,\mu) \xrightarrow{(i,j)}(\wt\beta,\wt \lambda,\wt\mu)$, as in the statement, and let $\Lambda$ be the unique partition with trivial 2-core and 2-quotient $(\lambda,\mu)$, such that $(i,j)\in D^{\mathrm{odd}}(\Lambda)$ and let $\wt\Lambda$ be the partition obtained from $\Lambda$ by removing the border strip $B_{ij}(\Lambda)$.
We introduce, as usual, $r\geq\ell(\lambda)$, $s\geq \ell(\mu)$, such that $L=r+s$ is even, $\bs N=\Lambda+\delta_L$, $\bs m=\lambda+\delta_r$, and $\bs n=\mu+\delta_s$.
As we recalled, the fact that $\wt\Lambda$ is obtained from $\Lambda$ by removing the boder strip $B_{ij}(\Lambda)$ is equivalent to the fact that $\wt{\bs N}=\wt\Lambda+\delta_L$ is obtained from $\bs N$ by subtracting $h_{ij}(\Lambda)$ from $N_i$ and re-arranging in decreasing order.
Letting $\wt{\bs m}=\wt\lambda+\delta_r$ and $\wt{\bs m}=\wt\mu+\delta_s$, the fact that $(\wt\lambda,\wt\mu)$ is the 2-quotient of $\wt\Lambda$ is equivalent to the fact that $\wt{\bs N}$ is a permutation of $(2\wt{\bs m},2\wt{\bs n}+1)$, hence we have two cases, depending on whether $N_i$ is even or odd.
If $N_i$ is odd (which happens if and only if $\Lambda_i-i$ is odd, because $L$ is even by assumption, and so if and only if $\beta\geq 0$), then $N_i=2n_k+1$ for some $k$ and 
\be
\wt{\bs m}=\bs m\sqcup\biggl\lbrace \frac{N_i-h_{ij}(\Lambda)}2 \biggr\rbrace=\bs m\sqcup\lbrace n_k-\beta\rbrace,\quad
\wt{\bs n}=\bs n\setminus\biggl\lbrace \frac{N_i-1}2\biggr\rbrace=\bs n\setminus\lbrace n_k\rbrace
\ee
and so we can use~\eqref{eq:defPlambdamu} and~\eqref{eq:coalescenceproof1} (with $b=\beta-r+s$) to prove that
\be
\Phi_{\lambda,\mu}^{(\beta)}(y)\,=\,\widetilde{\Phi}_{\bs m,\bs n}^{(\beta-r+s)}(y)
\,=\,y^{\beta+1}\,\widetilde{\Phi}_{\wt{\bs m},\wt{\bs n}}^{(\beta-r+s)}(y)
\,=\,y^{\beta+1}\,\Phi_{\wt\lambda,\wt\mu}^{(\beta+2)}(y).
\ee
Instead, if $N_i$ is even (which happens if and only if $\Lambda_i-i$ is even, and so if and only if $\beta<0$), then $N_i=2m_k$ for some $k$ and
\be
\wt{\bs m}=\bs m\setminus\biggl\lbrace \frac{N_i}2 \biggr\rbrace=\bs m\setminus\lbrace m_k\rbrace,\qquad
\wt{\bs n}=\bs n\sqcup\biggl\lbrace \frac{N_i}2-\frac{h_{ij}(\Lambda)+1}2\biggr\rbrace=\bs n\sqcup\lbrace m_k+\beta\rbrace
\ee
and so we can use~\eqref{eq:defPlambdamu} and~\eqref{eq:coalescenceproof2} (with $b=\beta-r+s$) to prove that
\be
\Phi_{\lambda,\mu}^{(\beta)}(y)\,=\,\widetilde{\Phi}_{\bs m,\bs n}^{(\beta-r+s)}(y)
\,=\,y^{-\beta+1}\,\widetilde{\Phi}_{\wt{\bs m},\wt{\bs n}}^{(\beta-r+s)}(y)
\,=\,y^{-\beta+1}\,\Phi_{\wt\lambda,\wt\mu}^{(\beta-2)}(y).
\ee
The proof is complete.
\end{proof}

\begin{example}
Take $\lambda=(3,1)$, $\mu=(2)$, and $\Lambda=(5,5,1,1)$ as in Example~\ref{ex:constant}, in which we illustrated that the root $\beta=1$ corresponds to the cell $(2,3)$ (with hook length $3$) whose border strip is highlighted in the diagram of $\Lambda$ as follows:
    \[\begin{tikzpicture}[scale=0.4]
\def\cellsize{1}

\foreach \c in {2,3,4} {
    \fill[color=gray!30] (\c, -1) rectangle ++(\cellsize,-\cellsize);
}

\foreach \c in {0,1,2,3,4} {
    \draw (\c, 0) rectangle ++(\cellsize,-\cellsize);
}

\foreach \c in {0,1,2,3,4} {
    \draw (\c, -1) rectangle ++(\cellsize,-\cellsize);
}

\draw (0, -2) rectangle ++(\cellsize,-\cellsize);

\draw (0, -3) rectangle ++(\cellsize,-\cellsize);

\node at (1.5, -0.5) {5};
\node at (3.5, -0.5) {3};

\node at (0.5, -1.5) {7};
\node at (2.5, -1.5) {3};
\node at (4.5, -1.5) {1};
\node at (0.5, -3.5) {1};

\end{tikzpicture}
\]
With the notations of Proposition~\ref{prop:coalescence}, we have $\wt\Lambda=(5,2,1,1)$, with 2-quotient $\wt\lambda=(2,1,1)$ and $\wt \mu=\emptyset$, such that $\bigl(1,(3,1),(2)\bigr)\,\xrightarrow{(2,3)}\,\bigl(3,(2,1,1),\emptyset\bigr)$.
Indeed, we have
\be
\Phi_{(3,1),(2)}^{(1)}(y)=y^6-8 y^5+20 y^4-32 y^3+24 y^2
,\quad \Phi_{(2,1,1),\emptyset}^{(3)}(y)=y^4-8 y^3+20 y^2-32 y+24.
\ee
\end{example}

\begin{definition}
\label{def:reduced}
Consider the set $\C\times\Partitions\times\Partitions$ of triples $\mathcal T=(\beta,\lambda,\mu)$, where $\beta\in\C$, $\lambda,\mu\in\Partitions$. Introduce an equivalence relation $\sim$ on $\C\times\Partitions\times\Partitions$ as follows: we have $\mathcal T\sim\wt{ \mathcal T}$ if and only if $\mathcal T=\wt{\mathcal T}$ or there are $\mathcal T_0,\mathcal T_1,\dots,\mathcal T_k\in\C\times\Partitions\times\Partitions$ with either $\mathcal T_0=\mathcal T$ and $\mathcal T_k=\wt{\mathcal T}$ or $\mathcal T_0=\wt{\mathcal T}$ and $\mathcal T_k=\mathcal T$ such that $\mathcal T_0\xrightarrow{(i_1,j_1)}\mathcal T_1$, $\mathcal T_1\xrightarrow{(i_2,j_2)}\mathcal T_2$, $\cdots$, $\mathcal T_{k-1}\xrightarrow{(i_k,j_k)}\mathcal T_k$.

We say that a triple $(\beta,\lambda,\mu)\in\C\times\Partitions\times\Partitions$ is reduced if and only if~$\phi^{[0]}_{\lambda,\mu}(\beta)\not=0$.
\end{definition}

Proposition~\ref{prop:coalescence} implies the following corollary.

\begin{corollary}\label{cor:coalescence}
    \begin{enumerate}[leftmargin=*]
        \item There is a unique reduced triple $(\wt\beta,\wt\lambda,\wt\mu)\in\C\times\Partitions\times\Partitions$ in each equivalence class.
    For any other $(\beta,\lambda,\mu)$ in the same equivalence class of a reduced triple $(\wt\beta,\wt\lambda,\wt\mu)$, we have $|\wt\beta|>|\beta|$, $|\wt\lambda|+|\wt\mu|<|\lambda|+|\mu|$, and $\Phi_{\lambda,\mu}^{(\beta)}(y)\,=\,y^{|\lambda|+|\mu|-|\wt\lambda|-|\wt\mu|}\,\Phi_{\wt\lambda,\wt\mu}^{(\wt\beta)}(y)$.
    \item Any triple $(\beta,\lambda,\mu)$ with $\beta\not\in\mathbb Z$ is reduced and its equivalence class consists of that triple alone.
    \item If $(\beta,\la,\mu) \sim (\wt \beta, \wt \la, \wt \mu) $ then the following rational functions coincide:
    \begin{equation}\label{eq:operatorinvariance}
    \frac{\beta^2-\frac14}{x^2}- 2 \frac{\d^2}{\d x^2 }\ln \Phi^{(\beta)}_{\la,\mu}(x^2)\,=\,  \frac{\wt \beta^2-\frac14}{x^2}- 2 \frac{\d^2}{\d x^2 }\ln \Phi^{(\wt\beta)}_{\wt\la,\wt\mu}(x^2) \,.
    \end{equation}
    \end{enumerate}

\end{corollary}

\subsection{Integers for which coalescence of zeros at the origin occurs (in the symmetric case) for partitions of assigned size}
For later convenience, we give an explicit description for all $n\in\mathbb N$ of the set
\be\label{eq:setBn}
B_n\,=\,\bigl\lbrace\, \beta\in\Z\,:\, \phi^{[0]}_{\lambda,\lambda'}(\beta)=0 \mbox{ for some partition }\lambda\mbox{ with }|\lambda|=n\,\bigr\rbrace\,.
\ee

\begin{proposition}\label{prop:Bn=Cn}
    For all $n\in\N$ we have $B_n=C_n$ where
    \be\label{eq:setCn}
    C_n\,=\,\bigl\lbrace \beta\in 2\Z+1\,:\, |\beta|\leq 2n-1\rbrace
    \cup\lbrace \beta\in 2\Z\,:\, |\beta|\leq n-2\,\bigr\rbrace.
    \ee
\end{proposition}
\begin{proof}
    Let us introduce, for all $\lambda\in\Partitions$,
    \be
    \begin{aligned}
    B_\lambda &\,=\,\biggl\lbrace
    \beta_{ij}(\Lambda)\ :\ (i,j)\in D^{\mathrm{odd}}(\Lambda),\text{ where $\Lambda$ is the unique partition} \\
    &\qquad\qquad\qquad\qquad\qquad\qquad\qquad\qquad\text{ with empty 2-core and 2-quotient $(\lambda,\lambda')$}\biggr\rbrace\\
    \end{aligned}
    \ee
    where
    \be
    \beta_{ij}(\Lambda)\,=\,\frac{-(-1)^{\Lambda_i-i}h_{ij}(\Lambda)-1}{2}.
    \ee
    Then, according to Corollary~\ref{corollary:p0denfree}, we have $B_n = \bigcup_{\lambda\in\Partitions,\,|\lambda|=n}B_\lambda$.
    The inclusion $C_n\subseteq B_n$ follows from the following two facts. First (see Lemma~\ref{lemma:Cn1}), for all $n\in\Z_{\geq 1}$ we have
    \be
    \label{eq:crossreflemmaCn1}
    \begin{aligned}
    B_{(n)} \,&=\,\lbrace \beta\in \mathbb{Z}\,:\,\beta=-2n+1 \text{ or }-n+2\leq \beta\leq 1\rbrace,\\
    B_{(1^n)}\,&=\,\lbrace \beta\in \mathbb{Z}\,:\,\beta=2n-1 \text{ or }-1\leq \beta\leq n-2\rbrace.
    \end{aligned}
    \ee
    This already shows that $B_n=C_n$ when $n=1,2$.
    Moreover (see Lemma~\ref{lemma:Cn2}):
    \begin{itemize}[leftmargin=*]
    \item for the partition $\Lambda$ whose 2-core is trivial and whose 2-quotient is $(\lambda,\lambda')$ with $\lambda=(2^k,1^{n-2k})$ ($n\geq 2k$, $n\geq 3$, $k\geq 0$) we have $\Lambda_{1}=2(n-k)$ and so $h_{11}(\Lambda)=2\Lambda_1-1=4(n-k)-1$ such that $\beta_{11}(\Lambda)=2(n-k)-1\in B_{n}$ for all $n\geq 2k$, $n\geq 3$, and $k\geq 0$;
    \item for the partition $\Lambda$ whose 2-core is trivial and whose 2-quotient is $(\lambda,\lambda')$ with $\lambda=(n-k,k)$ ($n\geq 2k$, $n\geq 3$, $k\geq 0$) we have $\Lambda_{1}=2(n-k)-1$ and so $h_{11}(\Lambda)=2\Lambda_1-1=4(n-k)-3$ such that $\beta_{11}(\Lambda)=-2(n-k)+1\in B_{n}$ for all $n\geq 2k$, $n\geq 3$, and $k\geq 0$.
    \end{itemize}
    This proves that $C_n\subseteq B_n$.
    To see that $B_n\subseteq C_n$, namely, that $\beta_{ij}(\Lambda)\in C_n$ whenever $\Lambda$ is a self-conjugate partition of $4n$ with empty 2-core, we consider two cases.
    \begin{itemize}[leftmargin=*]
    \item $i=j$. Since $h_{ii}=2\Lambda_i-2i+1$ and since $\frac{-(-1)^{\xi}(2\xi-1)-1}{2}$ is an odd integer for all $\xi\in\mathbb{Z}$, we obtain that $\beta_{ii}(\Lambda)$ is always an odd integer. Since $h_{ii}(\Lambda)\leq 4n-1$ we easily conclude that $\beta_{ii}(\Lambda)$ can only be an odd integer between $-2n+1$ and $2n-1$ and so $\beta_{ii}(\Lambda)\in C_n$.
    \item If $i\not=j$ we can as well assume $i<j$. Indeed, $h_{ij}(\Lambda)=\Lambda_i-i+\Lambda_j-j+1=h_{ji}(\Lambda)$, cf.~\eqref{eq:defhooklength}, and so when $h_{ij}(\Lambda)$ is odd the two integers $\Lambda_i-i$ and $\Lambda_j-j$ have the same parity thus $\beta_{ij}(\Lambda)=\beta_{ji}(\Lambda)$.
    Next, since the diagram of $\Lambda$ must contain both hooks $H_{ij}(\Lambda)$ and $H_{ji}(\Lambda)$ (which are either disjoint or meet at the cell $(j,j)$) as well as the square $\lbrace 1,\dots,j-1\rbrace\times\lbrace 1,\dots,j-1\rbrace$, we must have $2h_{ij}(\Lambda)-1+(j-1)^2\leq 4n$, i.e.
    \be
    \label{eq:boundhij}
    h_{ij}(\Lambda)\leq 2n-\frac{(j-1)^2-1}{2}.
    \ee
    In particular, since $j>i\geq 1$, we must have $h_{ij}(\Lambda)\leq 2n-1$. 
    Moreover, if $h_{ij}(\Lambda)\leq 2n-5$ then $\bigl|\beta_{ij}(\Lambda)\bigr|\leq n-2$ and so $\beta_{ij}(\Lambda)\in C_n$. Hence, we have to consider the cases $h_{ij}(\Lambda)=2n-1$ and $h_{ij}(\Lambda)=2n-3$ only.
    \begin{itemize}[leftmargin=*]
        \item $h_{i,j}(\Lambda)=2n-1$. According to~\eqref{eq:boundhij} we must have $(i,j)=(1,2)$. If $H_{12}(\Lambda)$ and $H_{21}(\Lambda)$ are disjoint then $\Lambda$ is a single self-conjugate hook, and this is impossible because $|\Lambda|=4n$ is even.
        Hence $H_{12}(\Lambda)\cap H_{21}(\Lambda)=\lbrace (2,2)\rbrace$ and so the diagram of $\Lambda$ must contain that of a self-conjugate partition of the form $(x,y,2,\dots,2,1,\dots,1)$ of size $4n-2$.
        It is impossible to attach exactly two cells to the latter diagram in order to obtain that of a self-conjugate partition $\Lambda$, hence this case is impossible.

        \item $h_{ij}(\Lambda)=2n-3$. According to~\eqref{eq:boundhij} we must have $j=2$ or $j=3$.
        \begin{itemize}[leftmargin=*]
            \item When $(i,j)=(1,2)$ the hooks $H_{12}(\Lambda)$ and $H_{21}(\Lambda)$ cannot be disjoint, otherwise $\Lambda$ would be a single hook.
            Therefore, $H_{12}(\Lambda)\cap H_{21}(\Lambda)=\lbrace (2,2)\rbrace$ and so the diagram of $\Lambda$ must contain that of a partition of the form $(x,y,2,\dots,2,1,\dots,1)$ (with $x+y=2n-1$) and we must attach $6$ more cells to the latter diagram to get the diagram of $\Lambda$. Since the only self-conjugate partition of $6$ is $(3,2,1)$, we necessarily have $\Lambda=(x,y,5,4,3,2,\dots,2,1,\dots,1)$, with $y-5$ entries equal to $2$ and $x-y$ equal to $1$. Since $x+y=2n-1$, $x$ and $y$ have opposite parity. If $x$ is odd and $y$ is even, the core of $(x,y,5,4,3,2,\dots,2,1,\dots,1)$ is $(7,6,5,4,3,2,1)$.
            Hence $x=\Lambda_1$ must be even and so $\beta_{12}(\Lambda) = \frac{(2n-3)-1}{2}=n-2\in C_n$.
            \item When $(i,j)=(1,3)$, we have $H_{13}(\Lambda)\cap H_{31}(\Lambda)=\lbrace(3,3)\rbrace$ (otherwise, if $H_{13}(\Lambda)\cap H_{31}(\Lambda)=\emptyset$, the diagram of $\Lambda$ would be the disjoint union of the square $\lbrace 1,2\rbrace\times\lbrace 1,2\rbrace$ and of the two hooks, which only make up a partition of $2(2n-3)+4=4n-2$). Therefore, the diagram of $\Lambda$ contains a square of side $2$ and the two hooks, totaling $4+2(2n-3)-1=4n-3$ cells. Let $2w+1$ be the hook-length of the cell $(3,3)$. Then, $\Lambda$ must also contain $2w$ more cells (in the second column and second row), whence $w=0,1$. In the first case, $\Lambda$ contains a self-conjugate partition of the form $(x,3,3,1,\dots,1)$ (the union of the two hooks and of the square of size $2$) of size $4n-3$ and is it impossible to add $3$ more squares to it to get a self-conjugate partition of $4n$ without altering the hook length of $(3,3)$. In the second case, we necessarily have $\Lambda=(x,4,4,4,1,\dots,1)$, with $x-4$ entries equal to $1$ and $x=2n-4$. Hence, also in this case $\Lambda_1=x$ is even and so $\beta_{13}(\Lambda) = \frac{(2n-3)-1}{2}=n-2\in C_n$.
            \item Finally, when $(i,j)=(2,3)$, if $H_{23}(\Lambda)\cap H_{32}(\Lambda)=\emptyset$ then $|\Lambda|=2^2+4(2n-3)=8n-8$, which is $\leq 4n$ only for $n=1,2$ and we already checked these cases, cf.~\eqref{eq:crossreflemmaCn1}.
            Then, we must have $H_{23}(\Lambda)\cap H_{32}(\Lambda)=\lbrace(3,3)\rbrace$. If the hook associated with the cell $(2,3)$ has the form $(h+1,1\dots,1)$, with $v$ entries equal to $1$ after the entry $h+1$, then $h+v+1=2n-3$, $h\geq v-1$, and so $4h+2v\geq 3h+3v-1=6n-10$, which implies $|\Lambda|\geq 9+4h+2(v-1)\geq 6n-3$ and $6n-3>4n$ for $n\geq 2$ and so there is nothing to check in this case.
        \end{itemize}
    \end{itemize}
\end{itemize}
\end{proof}

\begin{lemma}
\label{lemma:Cn1}
Identity~\eqref{eq:crossreflemmaCn1} holds true for all $n\geq 1$.
\end{lemma}
\begin{proof}
Using the construction of Remark~\ref{remark:constructionLambda}, it is easy to check that the unique partition $\Lambda$ with 2-quotient $\lambda=(n)$ and $\mu=(1^n)$ is $\Lambda=(2,2)$ when $n=1$ and $\Lambda\,=\,(2n-1,3,2,1^{2n-4})$ when $n\geq 2$. 
Its diagram and odd hook lengths are as follows:
\begin{equation*}
\begin{tikzpicture}[scale=0.7]
\def\w{1.7}
\def\h{0.5}
\draw (0*\w,0) rectangle ++(\w,\h) node[pos=.5] {$4n-3$};
\draw (1*\w,0) rectangle ++(\w,\h);
\draw (2*\w,0) rectangle ++(\w,\h);
\draw (3*\w,0) rectangle ++(\w,\h);
\draw (4*\w,0) rectangle ++(\w,\h) node[pos=.5] {$2n-3$};
\node at (5.75*\w,.5*\h) {$\cdots$};
\draw[dashed] (5*\w,\h) -- (6.5*\w,\h);
\draw[dashed] (5*\w,0) -- (6.5*\w,0);

\draw (6.5*\w,0) rectangle ++(\w,\h) node[pos=.5] {3};
\draw (7.5*\w,0) rectangle ++(\w,\h);
\draw (8.5*\w,0) rectangle ++(\w,\h) node[pos=.5] {1};

\draw (0*\w,-1*\h) rectangle ++(\w,\h);
\draw (1*\w,-1*\h) rectangle ++(\w,\h) node[pos=.5] {3};
\draw (2*\w,-1*\h) rectangle ++(\w,\h) node[pos=.5] {1};

\draw (0*\w,-2*\h) rectangle ++(\w,\h);
\draw (1*\w,-2*\h) rectangle ++(\w,\h) node[pos=.5] {1};

\draw (0*\w,-3*\h) rectangle ++(\w,\h);

\draw (0*\w,-4*\h) rectangle ++(\w,\h) node[pos=.5] {$2n-3$};

\node at (.5*\w,-4.5*\h) {$\vdots$};
\draw[dashed] (0,-4*\h) -- (0,-5.5*\h);
\draw[dashed] (\w,-4*\h) -- (\w,-5.5*\h);

\draw (0*\w,-6.5*\h) rectangle ++(\w,\h) node[pos=.5] {3};

\draw (0*\w,-7.5*\h) rectangle ++(\w,\h);

\draw (0*\w,-8.5*\h) rectangle ++(\w,\h) node[pos=.5] {1};
\end{tikzpicture}
\end{equation*}
    Similarly, the unique partition $\Lambda$ with 2-quotient $\lambda=(1^n)$ and $\mu=(n)$ is $\Lambda\,=\,(2n,2,1^{2n-2})$ when $n\geq 1$.
    Its diagram and odd hook lengths are as follows:
\begin{equation*}
\begin{tikzpicture}[scale=0.7]
\def\w{1.7}
\def\h{0.5}
\draw (0*\w,0) rectangle ++(\w,\h) node[pos=.5] {$4n-1$};
\draw (1*\w,0) rectangle ++(\w,\h);
\draw (2*\w,0) rectangle ++(\w,\h);
\draw (3*\w,0) rectangle ++(\w,\h) node[pos=.5] {$2n-3$};
\node at (4.75*\w,.5*\h) {$\cdots$};
\draw[dashed] (4*\w,\h) -- (5.5*\w,\h);
\draw[dashed] (4*\w,0) -- (5.5*\w,0);
\draw (5.5*\w,0) rectangle ++(\w,\h) node[pos=.5] {3};
\draw (6.5*\w,0) rectangle ++(\w,\h);
\draw (7.5*\w,0) rectangle ++(\w,\h) node[pos=.5] {1};

\draw (0*\w,-1*\h) rectangle ++(\w,\h);
\draw (1*\w,-1*\h) rectangle ++(\w,\h) node[pos=.5] {1};

\draw (0*\w,-2*\h) rectangle ++(\w,\h);

\draw (0*\w,-3*\h) rectangle ++(\w,\h) node[pos=.5] {$2n-3$};

\node at (.5*\w,-3.5*\h) {$\vdots$};
\draw[dashed] (0,-3*\h) -- (0,-4.5*\h);
\draw[dashed] (\w,-3*\h) -- (\w,-4.5*\h);

\draw (0*\w,-5.5*\h) rectangle ++(\w,\h) node[pos=.5] {3};

\draw (0*\w,-6.5*\h) rectangle ++(\w,\h);

\draw (0*\w,-7.5*\h) rectangle ++(\w,\h) node[pos=.5] {1};
\end{tikzpicture}
\end{equation*}
    It is easy to check that~\eqref{eq:crossreflemmaCn1} holds true in both cases.
\end{proof}
\begin{lemma}
    \label{lemma:Cn2}
    For the partition $\Lambda$ whose 2-core is trivial and whose 2-quotient is $(\lambda,\lambda')$ with $\lambda=(2^k,1^{n-2k})$ ($n\geq 2k$, $n\geq 3$) we have $\Lambda_{1}=2(n-k)$.
    For the partition $\Lambda$ whose 2-core is trivial and whose 2-quotient is $(\lambda,\lambda')$ with $\lambda=(n-k,k)$ ($n\geq 2k$, $n\geq 3$) we have $\Lambda_{1}=2(n-k)-1$.
\end{lemma}
\begin{proof}
For the partitions $\lambda=(2^k,1^{n-2k})$ and $\mu=\lambda'=(n-k,k)$, when $n\geq 2k$ and $n\geq 3$, we have
\be\begin{aligned}
\bs m=\lambda+\delta_{n} \,&=\, (n+1,n,n-1,\dots, n-2k+1,n-2k-1,n-2k-2,\dots,2,1),\\
\bs n=\mu+\delta_n\,&=\,(2n-k-1,n+k-2,n-3,n-2,\dots,2,1,0).
\end{aligned}
\ee
According to the construction in Remark~\ref{remark:constructionLambda},  $\Lambda=\bs N-\delta_{2n}$ where $\bs N=\bigl(2(2n-k-1)+1,\cdots\bigr)$, hence $\Lambda_1=2(2n-k-1)+1-(2n-1) = 2(n-k)$.

Similarly, for the partitions $\lambda=(n-k,k)$ and $\mu=\lambda'=(2^k,1^{n-2k})$, when $n\geq 2k$ and $n\geq 3$, we have
\be\begin{aligned}
\bs m=\lambda+\delta_{n} \,&=\,(2n-k-1,n+k-2,n-3,n-2,\dots,2,1,0),\\
\bs n=\mu+\delta_n\,&=\,(n+1,n,n-1,\dots, n-2k+1,n-2k-1,n-2k-2,\dots,2,1).
\end{aligned}
\ee
According to the construction in Remark~\ref{remark:constructionLambda},  $\Lambda=\bs N-\delta_{2n}$ where $\bs N=\bigl(2(2n-k-1),\cdots\bigr)$, hence $\Lambda_1=2(2n-k-1)-(2n-1) = 2(n-k)-1$.
\end{proof}

\section{Rational Extensions of the Harmonic Oscillator and their Spectral Theory}\label{sec:spectra}

If $\beta \in \C$ and $\mathcal P(x)$ is a monic polynomial with coefficients in $\C$, we define the linear ordinary differential operator
\be
\label{eq:Lop}
\mathscr H_{\mathcal P}^{(\beta)}=-\frac{\d^2}{\d x^2}+x^2+\frac{\beta^2-\frac 14}{x^2}-2 \frac{\d^2}{\d x^2} \ln \mathcal P(x).
\ee
In the particular case $\mathcal P=1$, we write $\mathscr H^{(\beta)}_1=\mathscr H^{(\beta)}$, in agreement with~\eqref{eq:Hbintro}.

\begin{definition}\label{def:rational}
Let $\mathcal P(x)$ be a monic polynomial of degree $M>0$.
We say that $\mathscr H_{\mathcal P}^{(\beta)}$ is a degree $M$ rational extension of the harmonic oscillator with momentum $\beta$ if and only if $\mathcal P(0) \neq 0$ and, for every $E \in \C$, the general solution of the Schr\"odinger equation
\begin{equation}\label{eq:extendedharmonic}
  \mathscr H_{\mathcal P}^{(\beta)} \, \psi(x)\,=\,E \, \psi(x)
\end{equation}
has trivial monodromy about all roots of $\mathcal P$.
\end{definition}
\begin{remark}
   We notice that $ \mathscr H^{(\beta)}_{\mathcal P}$ is a rational extension if and only if every local solution of the differential equation $\mathscr H_{\mathcal P}^{(\beta)}  \psi(x)=E \psi(x)$ extends to a meromorphic function on the universal cover of $\mathbb C^*$. From now on, the term ``solution'' will always point  to such a global function.
\end{remark}

\begin{definition}
We say that the monic polynomial~$\mathcal P$ is a Bazhanov--Lukyanov--Zamolodchikov polynomial 
(BLZ polynomial, for short) of level~$n$ and momentum~$\beta$ if 
$\mathscr H_{\mathcal P}^{(\beta)}$ is a rational extension of the harmonic oscillator of degree $4n$ and $\mathcal P(\i x)=\mathcal P(x)$.
\end{definition}

\begin{remark}
Assuming that $\mathcal P$ has distinct roots\footnote{If roots coincide, the system 
\eqref{eq:calogeromoser} reads $-x_k+ \frac{\beta^2-\frac14}{x_k^3} +   \sum_{x_j \neq x_k} \frac{2}{(x_k-x_j)^3}  =0 $ and it is a necessary but not sufficient condition for $\mathcal{P}$ to be a rational extension.}, it is well-known (see, e.g.,~\cite{coma20}) that $\Hbp$ is a rational extension of the harmonic oscillator if and only if the roots  $x_1,\dots,x_M$ of $\mathcal P$ satisfy the system of rational equations
\begin{equation}\label{eq:calogeromoser}
-x_k+ \frac{\beta^2-\frac14}{x_k^3} +   \sum_{j \neq k} \frac{2}{(x_k-x_j)^3}  \,=\,0, \qquad k=1, \dots, M,
\end{equation}
which is the equilibrium condition for the Calogero--Moser potential
\begin{equation*} 
V(x)\,=\,x^2 +\frac{\beta^2-\frac14}{x^2}+\sum_{k=1}^M\frac{2}{(x-x_k)^2}.
\end{equation*}
If we further assume that $\mathcal P$ is a BLZ polynomials, then $\mathcal P(x)=\prod_{i=1}^n (x^4-z_i)$ (with pairwise distinct $z_i$) and~\eqref{eq:calogeromoser} reduces to the system of rational equations
\begin{equation}\label{eq:blzsystem}
\sum_{j \neq k} \frac{z_k(z_k^2+12 z_k z_j+ 3 z_j^2 )}{(z_k-z_j)^3}-\frac{z_k}{8}+\frac{\beta^2-1}{8}\,=\,0, \qquad k=1\dots n,
\end{equation}
which is known as BLZ system~\cite{BLZHigher,coma20}.
\end{remark}

\subsection{Spectra}
For every rational extension $\Hbp$, we define three spectral problems and three corresponding spectral determinants. To this aim, we define a basis of distinguished solutions in a neighborhood of~$0$ and a distinguished subdominant solution at $+\infty$ of the stationary Schr\"odinger equation \eqref{eq:extendedharmonic}.

Near $0$, we look for solutions $\chi_{\pm}^{(\beta)}(x,E;\mathcal P)$ which have a Frobenius expansion without logarithmic term,
\begin{equation}\label{eq:chipm}
    \chi_{\pm}^{(\beta)}(x,E;\mathcal P)\,=\,x^{\pm \beta+\frac12}\, \biggl( 1 + \sum_{k \geq 1} c^{\pm}_k x^k \biggr)\,,\qquad  \mbox{for } |x|  \mbox{ small enough}.
\end{equation}
We have the following standard result.
\begin{lemma}\label{lem:frobenius}
Let $\Hbp$ be a rational extension of the harmonic oscillator.
\begin{enumerate}[leftmargin=*]
    \item The equation $\Hbp \psi= E \psi$ admits a unique solution $ \chi_{\pm}^{(\beta)}(x,E;\mathcal P)$ with the expansion \eqref{eq:chipm} for all values of $E$, provided $\mp 2\beta$ is not a positive integer.
    \item If $\mathcal P$ is an even polynomial, the equation $\Hbp \psi= E \psi$ admits a unique solution $ \chi_{\pm}^{(\beta)}(x,E;\mathcal P)$ with the expansion \eqref{eq:chipm} for all values of $E$,  provided $\mp\beta$ is not a positive integer. 
    \item Assume that $\mathcal{P}=\mathcal{P}(x;\beta)$ is an even polynomial in $x$ which depends analytically on the parameter $\beta$ as $\beta$ varies on an open $D \subseteq \C$, and that  $\mathcal{P}(0;\beta)\neq 0$ for all $\beta \in D$.
    The solution $ \chi_{\pm}^{(\beta)}(x,E;\mathcal P)$ is an entire function of the parameter $E$ and a meromorphic function of the parameter $\beta$, whose only singularities are simple poles located at $D \cap \big(\mp\N^*\big)$.
\end{enumerate}
\end{lemma}
\begin{proof}
   See, e.g., \cite[Chapter 5]{tricomi}.
\end{proof}

We now turn our attention at $+\infty$. We are interested in the solution decreasing exponentially fast at $+\infty$. We have the following classical result.
\begin{lemma}\label{lem:sibuya}
Let $\Hbp$ be a rational extension of the harmonic oscillator.
For every $(E,\beta) \in \C \times \C$, there exists a unique solution $\varphi_0^{(\beta)}(x,E;\mathcal P)$ such that 
\begin{equation}\label{eq:psizero}
    \varphi_0^{(\beta)}(x,E;\mathcal P)\,=\, x^{\frac{E-1}{2}} \,\e^{-\frac{x^2}{2}}\,\bigl(1+ o(1) \bigr)\, , \qquad
    x \to + \infty\,.
\end{equation}
Moreover,
\begin{equation}\label{eq:psizeroder}
   \frac{\d}{\d x} \varphi_0^{(\beta)}(x,E;\mathcal P)\,=\, -x^{\frac{E+1}{2}} \,\e^{-\frac{x^2}{2}}\,\bigl(1+ o(1) \bigr)\, , \qquad
    x \to + \infty\,,
\end{equation}
and the asymptotics (\ref{eq:psizero},\ref{eq:psizeroder}) are valid more generally as $x\to\infty$ in the sector $|\mathrm{arg}\,x|\leq \frac{3\pi}{4}-\ep$ for all $\ep>0$.

Assume that $\mathcal{P}=\mathcal{P}(x;\beta)$ is a polynomial in $x$ which depends analytically on the parameter $\beta$ as $\beta$ varies on an open $D \subseteq \C$. Then $\varphi_0^{(\beta)}(x,E;\mathcal P)$ is an entire function of the parameter $E$ and an analytic function of the parameter $\beta\in D$. 
\end{lemma}
\begin{proof}
See \cite[Proposition 4.6]{mamura}.
\end{proof}
The ray $\arg x=0$ is not the only direction of steepest descent for a soultion, but the rays $\arg x=  \frac{\pi k}{2}$ with $k=\pm1,2$ are too.
In fact, as we have defined the solution subdominant at $+\infty $, we can define the solutions subdominant at $\infty$ in these further directions. We denote such solutions by $\varphi_k$, and we fix the normalization
    \begin{equation}\label{eq:psidue}
    \varphi_k^{(\beta)}(x,E;\mathcal P)\,=\, |x|^{\frac{E-1}{2}} \,\e^{-\frac{x^2}{2}}\,\bigl(1+ o(1) \bigr)\, , \;
    |x| \to  \infty\ \mbox{ with } \arg x= \frac{\pi k}{2},\ k=-1,1,2
\end{equation}
For our study, we only need the solution $\varphi^{(\beta)}_2$. This can be written very simply in terms of $\varphi^{(\beta)}_0$ if $\mathcal{P}$ is an even polynomial:
\begin{align}\label{eq:phi2asphi0}
  &  \varphi_2^{(\beta)}(x,E;\mathcal{P})=\varphi_0^{(\beta)}(\e^{-\i \pi} x, E;\mathcal{P})  .
\end{align}

We can now define the spectral determinants $Q^{\pm}$ for the central connection problems and the spectral determinant for a lateral connection problem, better known as the Stokes multiplier $T$.
\begin{definition}
Let $\Hbp$ be a rational extension of the harmonic oscillator.
Assume that $\beta$ is such that $\chi_{\pm}(x,E,\beta)$ exists for all $E\in \C$. Then
the $\pm$ spectrum $\mathscr E^{(\beta)}_{\pm,\mathcal P}$ is
\be
\mathscr E^{(\beta)}_{\pm,\mathcal P}\,=\,\bigl\lbrace
E\in\C\,:\,
 \chi_{\pm}^{(\beta)}(x,E;\mathcal P)\mbox{ is proportional to }\varphi_0^{(\beta)}(x,E;\mathcal P)
\bigr\rbrace\,.
\ee
It coincides with the zero locus of the spectral determinant $Q^{\pm}_{\mathcal P}(\cdot; \beta): \C \to \C$ defined by
\begin{equation}
 E\, \mapsto \,  Q^{\pm}_{\mathcal P}(E;\beta)\,=\, \frac12 \,\Wr x\bigl(\chi_{\pm}^{(\beta)}(x,E;\mathcal P)\,,\, \varphi_0^{(\beta)}(x,E;\mathcal P)\bigr),
\end{equation}
where $\mathrm{Wr}_x$ is the Wronskian with respect to the variable $x$.

The Stokes multiplier $T_{\mathcal{P}}: \C^2 \to \C$ is defined\footnote{Most authors denote by $T_{\mathcal{P}}(E,\beta)$ the function $T_{\mathcal{P}}(-\i \, E,\beta)$; we follow here a normalization better suited to our study.} as 
\begin{align}
 & (E,\beta) \,\mapsto\, T_{\mathcal{P}}(E,\beta) = \frac14 \Wr x\bigl(\varphi_0(x,E;\mathcal{P})\,,\,\varphi_2(x,E;\mathcal{{P}})\bigr)
\end{align}
and it vanishes if and only if $\varphi_0$ is proportional to $\varphi_2$.
\end{definition}
The spectral determinants $Q^{\pm}$ are widely considered to be the most important gadget of quantum integrability, \cite{baxter85,FH16,mukhin04,marava17}.
On the other hand, the pair $(Q^+,Q^-)$ is only defined when $\beta \notin \Z$, hence the Stokes multiplier $T$ plays a major role in the study of the case $\beta \in \Z$. 

In the next lemma we collect a few analytic properties of the spectral determinants which follow rather directly from Lemma \ref{lem:frobenius} and \ref{lem:sibuya}.
\begin{lemma}\label{lem:spectralanalityc}
Let $\mathcal{P}=\mathcal P(x;\beta)$ be a monic polynomial of $x$ which depends analytically on the parameter $\beta$ on some open $D \subseteq\C$, and such that every solution of $\mathscr{H}^{(\beta)}_{\mathcal{P}}\psi=E\psi$ has the trivial monodromy property for every $\beta \in D$ and every $E \in \C$.
\begin{enumerate}[leftmargin=*]
\item Assume furthermore that $\mathcal P(0;\beta) \neq 0$ for all $\beta \in D$. The spectral determinant $Q^{\pm}_{\mathcal{P}}(E,\beta)$ is an entire function of $E$ and a meromorphic function of $\beta$, with only simple poles; the set of poles being $D \cap \big(- \N^* \big)$.
\item The Stokes multiplier
$T_{\mathcal{P}}(E,\beta)$ is analytic on $\C \times D$. In particular, it is an entire function in $E$ and $\beta$ if $D=\C$.
\item If $\beta \notin \Z$,
\begin{equation}\label{eq:coeffQ}
    \varphi_0^{(\beta)}(x,E;\mathcal P)\,=\, \frac{1}{\beta} \left( Q^{-}_{\mathcal P}(E;\beta) \,\chi^{(\beta)}_{+}(x,E;P)\,-\,
    Q^{+}_{\mathcal P}(E;\beta) \,\chi^{(\beta)}_{-}(x,E;P)\right).
\end{equation}
\item If $\mathcal{P}$ is an even polynomial, then
 \begin{align}\label{eq:T=QQ}
   &  T_{\mathcal{P}}(E,\beta)\,=\, Q^+_{\mathcal P}(E;\beta)\, Q^-_{\mathcal P}( E;\beta) \,\frac{\sin(\pi \beta)}{\beta} \,.
 \end{align}
\end{enumerate}
\end{lemma}
\begin{proof}
\textit{(1)} It follows from Lemma \ref{lem:frobenius} and Lemma \ref{lem:sibuya}. \newline
\textit{(2)} It follows from Lemma \ref{lem:sibuya}. \newline
\textit{(3)} Using the expansion \eqref{eq:chipm} we obtain
\begin{equation*}
\Wr x\left(\chi_-^{(\beta)}(x;E,\mathcal{P})\,,\,\chi_-^{(\beta)}(x;E,\mathcal{P})\right)=2 \beta.
\end{equation*}
Hence $\lbrace \chi_-^{(\beta)}(x;E,\mathcal{P}),\chi_+^{(\beta)}(x;E,\mathcal{P})\rbrace$ is a basis of the space of solutions of the differential equation $\mathscr{H}^{(\beta)}_{\mathcal{P}}\psi=E \psi$. Therefore, exist a unique pair $\left(A(E,\beta),B(E,\beta)\right)$ such that
\begin{equation*}
\varphi_0^{(\beta)}(x;E,\mathcal{P})\,=\,A(E,\beta)\,\chi_+^{(\beta)}(x;E,\mathcal{P})\,+\,B(E,\beta)\,\chi_-^{(\beta)}(x;E,\mathcal{P}).
\end{equation*}
Taking the Wronskian of $\varphi_0$ with $\chi_+$ and $\chi_-$, and using the above identities one obtains the thesis. \newline
\textit{(4)} It follows from (\ref{eq:phi2asphi0}) and \eqref{eq:coeffQ}.
\end{proof}

The fundamental object of the ODE/IM correspondence at the free-fermion point are the QQ relations satisfied by the spectral determinants
$Q^{\pm}$ which we now introduce following \cite{doreytateo98}.
\begin{lemma}\label{lem:QQsystemspectral}
If $\mathcal{P}$ is a BLZ polynomial and $\beta$ is not an integer, then
\be\label{eq:QQrelationscorrected}
 \e^{ \i  \frac{\beta \pi}{2}}\,Q^+_{\mathcal P}(\i E;\beta) \,Q^-_{\mathcal P}(-\i E;\beta)\,-\, \e^{- \i  \frac{\beta \pi}{2}} \,Q^+_{\mathcal P}(- \i E;\beta) \,Q^-_{\mathcal P}(\i E;\beta)\,=\, \i\, \e^{\frac{E \pi}4}, \quad \forall E \in \C.
\ee
\end{lemma}

These are the QQ relations at the free-fermion point.

\begin{proof}
Let us define the following functions
\begin{align*}
\varphi_{\pm \frac12}^{(\beta)}(x,E;\mathcal P)\,&=\,\varphi^{(\beta)}(\e^{\mp \i \frac{\pi}{4}}x,\e^{\pm \i \frac{\pi}{2}} E;\mathcal P) ,
\qquad
\chi^{(\beta)}_{-,\pm \frac{1}{2}}(x,E;\mathcal P)\,=\, \chi^{(\beta)}_-(\e^{\mp \i \frac{\pi}{4}}x,\e^{\pm \i \frac{\pi}{2}} E;\mathcal P),
 \\
\chi^{(\beta)}_{+,\pm \frac{1}{2}}(x,E;\mathcal P)\,&=\, \chi^{(\beta)}_+(\e^{\mp \i \frac{\pi}{4}}x,\e^{\pm \i \frac{\pi}{2}} E;\mathcal P).
\end{align*}
Since $\mathcal P(x)=\mathcal P(\i x)$, these functions satisfy the ODE
$\psi''=\bigl(-x^2- E - 2 \frac{\d^2}{\d x^2} \log \mathcal P(x) +\frac{\beta^2-\frac{1}4}{x^2}\bigr)\, \psi$.
By \eqref{eq:coeffQ},
\begin{align*}
    \psi_{\pm\frac12}(x,E,\beta,P)= \frac{1}{\beta} \left( Q_{-}(\pm \i E;\beta,P)  \, \chi_{+,\pm\frac12}(x,E,\beta,P) -  
       Q_+(\pm \i E;\beta,P)  \, \chi_{+,\pm\frac12}(x,E,\beta,P) \right).
\end{align*}
Using the asymptotic expansion of the solutions $\varphi_0, \chi_{\pm}$, we compute
\begin{align*}
  \Wr x(\varphi_{-\frac12},\varphi_{\frac{1}{2}})= 2 \i \e^{\frac{E \pi}{4}} \, , \quad
  \Wr x(\chi_{-,\pm\frac12},\chi_{+,\mp\frac12})=2 \beta \e^{\pm \frac{\i \pi \beta}{2}}.
\end{align*}
Using the above formulas, we deduce the thesis.
\end{proof}

Having defined the spectral gadgets associated to rational extensions of the harmonic oscillator, let us now turn our attention to those expansions that are expressed in terms of Laguerre polynomials.

\subsection{Rational extensions from Laguerre Wronskians}
Recall from Section  \ref{sec:laguerre} that to any triple $(\beta,\la,\mu) \in\C\times\Partitions\times\Partitions$, one associates the polynomial $\Phi_{\la,\mu}^{(\beta)}(y)$ (see Definition~\ref{def:PhiLaMu}). For fixed $\la,\mu$,  $\Phi_{\la,\mu}^{(\beta)}(y)$ depends polynomially on $\beta$ and so $\phi^{[0]}_{\la,\mu} (\beta)=\Phi_{\la,\mu}^{(\beta)}(0) \neq 0$ for generic $\beta$, in which case we say that the triple $(\beta,\la,\mu)$ is reduced, see Definition~\ref{def:reduced}.
We use  $\Phi_{\la,\mu}^{(\beta)}(y)$  to define a rational extension of the harmonic oscillator.
 \begin{definition}\label{def:Hblamu}
 For $(\beta,\lambda,\mu)\in\C\times\Partitions\times\Partitions$ we define $\mathscr H^{(\beta)}_{\la,\mu} := \mathscr H^{(\beta)}_{\mathcal P}$ where $\mathcal{P} = \Phi^{(\beta)}_{\lambda,\mu}\bigl(x^2\bigr) $.
 \end{definition}
If $(\beta,\la,\mu) \neq (\wt \beta, \wt \la,\wt \mu)$ the operators $\mathscr{H}_{\la,\mu}^{(\beta)}$ and $\mathscr{H}_{\wt\la,\wt\mu}^{(\wt \beta)}$ are not necessarily different.
In fact, by Corollary  \ref{cor:coalescence}, $\mathscr{H}_{\la,\mu}^{(\beta)}$ is invariant under the equivalence relation introduced in Definition~\ref{def:reduced}. 
Moreover, by construction,  $\mathscr{H}_{\la,\mu}^{(\beta)}=\mathscr{H}_{\mu,\la}^{(- \beta)}$, cf.~Remark~\ref{remark:equivalentextensions}.
 Therefore, in order to classify the operators $\mathscr H^{(\beta)}_{\la,\mu} $ we need to quotient $\C \times \Partitions \times \Partitions$ with respect to both the equivalence relation $\sim$ and the involution  $(\beta,\la,\mu) \mapsto (- \beta, \mu,  \la)$.
\begin{definition}
    We say that $(\beta,\la,\mu) \approx (\wt \beta, \wt \lambda, \wt \mu)$ if and only if $(\beta,\la,\mu) \sim (\wt \beta, \wt \lambda, \wt \mu)$ or $(-\beta,\mu,\la) \sim (\wt \beta, \wt \lambda, \wt \mu)$.
\end{definition}

The main result of the present section is the following theorem.

\begin{theorem}
\label{thm:lamuextension}~
\begin{itemize}[leftmargin=2em]
\item[I.]
If $\beta$ is not a root of $\phi^{[0]}_{\la,\mu}(\beta)$, the operator $\mathscr H^{(\beta)}_{\la,\mu} $ is a rational extension of the harmonic oscillator with momentum $\beta$ and degree $2 |\la|+ 2 |\mu|$. \newline If $\beta$ is a root of $\phi^{[0]}_{\la,\mu}(\beta)$, the operator $\mathscr H^{(\beta)}_{\la,\mu} $ is a rational extension of the harmonic oscillator with momentum $\wt \beta$, and degree $2 |\wt\la|+ 2 |\wt\mu|$, where
$(\wt\beta,\wt\lambda,\wt\mu)$ is the unique reduced triple in the equivalence class of $(\beta,\lambda,\mu)$. 

\item[II.]
Let $ \mp \beta$  not be a positive integer and let $Q^{\la,\mu}_{\pm}(E,\beta)$ be the spectral determinant of the $\pm$ spectral problem for the operator $\mathscr H^{(\beta)}_{\la,\mu} $.
If  $\phi^{[0]}_{\la,\mu}(\beta) \neq 0$, 
the following formulas hold:
 \begin{align}\nonumber
       & Q^+_{\la,\mu}(E,\beta) = (-1)^r 4^r \frac{\Gamma(1+\beta-r+s)\Gamma(-\beta+r-s)}{\Gamma(-\beta) \Gamma(1+\beta)} \\ \label{eq:Qplamu}
         & \qquad \qquad \qquad \qquad \prod_{k\in \lambda +{\bs\delta}_r}\big(E-E_+^{(\beta)}(k-r) \big)^{-1} \frac{ \Gamma(1+\beta)}{\Gamma \big(\frac{1+\beta}2-\frac{E}{4}-r \big)} ,  \\ \nonumber
        & Q^-_{\la,\mu}(E,\beta) = (-1)^r 4^s \frac{\Gamma(1-\beta+r-s)\Gamma(\beta-r+s)}{\Gamma(\beta) \Gamma(1-\beta)} \\ \label{eq:Qmlamu}
         & \qquad \qquad \qquad \qquad \prod_{k\in \mu +{\bs\delta}_s}\big(E-E^{(\beta)}_-(k-s) \big)^{-1} \frac{ \, \Gamma(1-\beta)}{\Gamma \big(\frac{1-\beta}2-\frac{E}{4}-s \big)}.
    \end{align}
 In the above formulas $E^{(\beta)}_{\pm}(k)=4 k \pm 2\beta+2$, as per \eqref{eq:eigenharm}, while $r,s$ are the length of $\la,\mu$.
The corresponding spectra are given by the formula
\begin{align} \label{eq:Eplamu}
& \mathscr{E}^{(\beta)}_{+,\la,\mu}=\left\{ 4(k-\la'_{k+1}) +2 +2\beta,\; k \in \N \right\} ,\\ \label{eq:Emlamu}
& \mathscr{E}^{(\beta)}_{-,\la,\mu}=\left\{ 4(k-\mu'_{k+1}) +2 -2\beta,\; k \in \N  \right\},
\end{align}
where $\la',\mu'$ are the conjugate partitions of $\la,\mu$.

\item[III.]
For all $(\beta,\la,\mu) \in \C \times \Partitions \times \Partitions$, the following formula hold
\begin{align}\label{eq:Tlamu}
       & T_{\la,\mu}(E,\beta) =  4^{r+s}   
           \frac{\prod_{k\in \lambda +{\bs\delta}_r}\big(E-E_+^{(\beta)}(k-r) \big)^{-1}}{\Gamma \big(\frac{1+\beta}2-\frac{E}{4}-r \big)}  \frac{\prod_{k\in \mu +{\bs\delta}_s}\big(E-E^{(\beta)}_-(k-s) \big)^{-1} }{\Gamma \big(\frac{1-\beta}2-\frac{E}{4}-s \big)} .
    \end{align}

  \item[IV.]
 $\mathscr{H}_{\la,\mu}^{(\beta)}=\mathscr{H}_{\wt\la,\wt\mu}^{(\wt \beta)}$ if and only if  $(\beta,\la,\mu) \approx (\wt \beta, \wt \lambda, \wt \mu)$.

\item[V.]
Let $(\beta,\la,\mu)$ and $(\widetilde{\beta},\widetilde{\la},\widetilde{\mu})$ be given with $\beta \notin \Z$.
Then
\begin{equation} \label{eq:injetivity}
    \Phi_{\la,\mu}^{(\beta)}= \Phi_{\widetilde{\la},\widetilde{\mu}}^{(\widetilde{\beta)}} \Longleftrightarrow \beta=\widetilde{\beta}, \; \la = \widetilde{\la} \mbox{ and } \mu=\widetilde{\mu} \mbox{ or } \beta=-\widetilde{\beta}, \;  \la = \wt \mu \mbox{ and } \mu=\wt{\la}  .
\end{equation}

Let $\beta,\la,\mu,\widetilde{\la},\widetilde{\mu}$ be given with $\beta \in \Z$ and $\phi^{[0]}(\beta)\neq 0$. 
Then \begin{equation} \label{eq:injetivityweak}
    \Phi_{\la,\mu}^{(\beta)}= \Phi_{\widetilde{\la},\widetilde{\mu}}^{(\beta)} \Longleftrightarrow  \la = \widetilde{\la} \mbox{ and } \mu=\widetilde{\mu}.
\end{equation}
\item[VI.]
Let $(\beta,\la,\mu)$ be such that $\phi^{[0]}_{\la,\mu}(\beta) \neq 0$. The polynomial $\Phi_{\la,\mu}^{(\beta)}(y)$ is even if and only if $\la=\mu'$.
\label{point:Conjecture}

\item[VII.] $\Phi_{\la,\mu}^{(\beta)}(x^2)$ is a BLZ polynomial if and only if $\mu=\la'$ and $\beta$ is not a root of  $\phi^{[0]}_{\la,\la'}(\beta)$. In this case, $\mathcal{P}^{(\beta)}_{\la}(x)=\Phi_{\la,\la'}^{(\beta)}(x^2)$ is a BLZ polynomial of level $ |\la|$ and momentum~$\beta$. 

The corresponding spectral determinants $Q^{\la}_{\pm}(E,\beta)$, defined whenever $\beta$ is not a negative integer, are given by the formula
 \begin{align}\nonumber
       & Q_+^{\la}(E,\beta) = (-1)^{r} 4^r \frac{\Gamma(1+\beta-r+r')\Gamma(-\beta+r-r')}{\Gamma(-\beta) \Gamma(1+\beta)} \\ \label{eq:Qpla}
         & \qquad \qquad \qquad \qquad \prod_{k\in \lambda +{\bs\delta}_r}\big(E-E^+(k-r,\beta) \big)^{-1} \frac{\Gamma(1+\beta)}{\Gamma \big(\frac{1+\beta}2-\frac{E}{4}-r \big)}  \\ \nonumber
        & Q_-^{\la}(E,\beta) = (-1)^r 4^{r'} \frac{\Gamma(1-\beta+r-r')\Gamma(\beta-r+r')}{\Gamma(\beta) \Gamma(1-\beta)} \\ \label{eq:Qmla}
         & \qquad \qquad \qquad \qquad \prod_{k\in \la' +{\bs\delta}_{r'}}\big(E-E^+(k-r',\beta) \big)^{-1} \frac{ \, \Gamma(1-\beta)}{\Gamma \big(\frac{1-\beta}2-\frac{E}{4}-r' \big)},
    \end{align}
    with  $r'$ the length of $\la'$. The corresponding spectra are given by
\begin{align}\label{eq:Epla}
\mathscr{E}_{+,\la}^{(\beta)}&=\left\{ 4(k-\la'_{k+1}) +2 +2\beta,\; k \in \N \right\}  , \\ \label{eq:Emla}
\mathscr{E}_{-,\la}^{(\beta)}&=\left\{4(k-\la_{k+1}) +2 -2\beta,\; k \in \N \right\}  .
\end{align}
In the above formula, the spectrum $\pm$ is defined whenever $\mp \beta$ is not a positive integer.
\end{itemize}
\end{theorem}

In Theorem \ref{thm:lamuextension}(IV) we have given a complete classification of the rational extensions of the harmonic oscillator with momentum $\beta$ obtained via Wronskians of Laguerre polynomials.
In particular, we are able to count them.

\begin{corollary}\label{cor:countingextensions}
    The number of rational extensions of the harmonic oscillator of momentum  $\beta$ and degree $2n$ obtained via Wronskians of Laguerre polynomials (i.e., via Crum--Darboux transformations) is equal to
    \begin{equation}
        \# \{ (\la,\mu) \in \Partitions \times \Partitions: |\la|+|\mu|=n,\; \phi^{[0]}_{\la,\mu}(\beta) \neq0 \}.
    \end{equation}
    Such a number is less or equal than $\sum_{k=0}^n p(k)\,p(n-k)$, namely, the number of partitions of $n$ into parts of two kinds, and equality holds for every $n$ such that $\phi^{[0]}_{\la,\mu}(\beta) \neq 0$ for all partitions $\la,\mu$ with $|\la|+|\mu|=n$.
    In particular, equality holds whenever $\beta \notin \Z$.

    The number of BLZ polynomial of level $n$ and momentum $\beta$ obtained via Wronskian of Laguerre polynomials (i.e., Crum--Darboux transformations) is
    \begin{equation}
        \# \bigl\lbrace \la \in \Partitions\,: \,|\la|=n, \;   \phi^{[0]}_{\la,\la'}(\beta) \neq0 \bigr\rbrace.
    \end{equation}
    Such a number is less or equal than $p(n)$ and equality holds, for a given $n$, if and only if $\beta \notin C_n$, where $C_n \subset \Z$ is defined in~\eqref{eq:setCn}.
\end{corollary}

A natural question is whether the rational extensions that we have constructed are all possible rational extensions, given that A.~Oblomkov~\cite{Oblomkov} has positively answered the same question when $\beta=\pm\frac12$.
The answer to this question is negative: if $\beta \neq \pm\frac12$, there are rational extensions which do not arise from Crum--Darboux transformations.
In fact, all rational extensions arising from Crum--Darboux transformations are invariant under the transformation $x \mapsto -x$ but, whenever $\beta \neq \pm\frac12$, system \eqref{eq:calogeromoser} admits four distinct solutions with degree $1$, none of which is invariant under the transformation $x \mapsto -x$.

One knows from \cite[Section 9]{coma20} that, for generic $\beta$:
\begin{itemize}[leftmargin=*]
    \item The rational extensions of the harmonic oscillator of degree $M$ are parametrized by partitions of $M$ into parts of four kinds;
    \item The rational extensions of the harmonic oscillator of degree $2 M$, invariant under the transformation $x \mapsto -x$, are parametrized by partitions of $M$ into parts of two kinds;
    \item The rational extensions of the harmonic oscillator of degree $4 M$, invariant under the transformation $x \mapsto \i x$ (namely BLZ extensions), are parametrized partitions of $M$;
\end{itemize} 
In view of this,  the following conjecture is natural.

\begin{conjecture}\label{conj:oblo}
Every rational extension of the harmonic oscillator with momentum~$\beta$, symmetric with respect to the transformation $x \mapsto -x$, is obtained via  Wronskians of Laguerre polynomials.
In particular, every BLZ polynomial is a Wronskian of Laguerre polynomials.
\end{conjecture}

\begin{remark}\label{rem:equivalence}
Due to Theorem \ref{thm:lamuextension}(IV), the condition $\phi^{[0]}_{\la,\mu}=0$ can be reformulated as follows.
Assuming $\beta \geq0$, $\phi^{[0]}_{\la,\mu}(\beta)=0$ if and only if there exists $(\wt \beta, \wt \la,\wt \mu)$ with $\wt \beta>\beta$ such that
\be
\mathscr{E}^{(\beta)}_{+,\la,\mu} + \mathscr{E}^{(\beta)}_{-,\la,\mu} =  \mathscr{E}^{(\wt\beta)}_{+,\wt\la,\wt\mu} + \mathscr{E}^{(\wt\beta)}_{-,\wt\la,\wt\mu}  ,
\ee
where $+$ is the multiset sum, and $\mathscr{E}^{(\beta)}_{\pm,\la,\mu}$ is as per (\ref{eq:Eplamu},\ref{eq:Emlamu}).
This is a a different, albeit equivalent, combinatorial problem with respect to the one studied in Section~\ref{sec:coalescence}.
\end{remark}
We can use such a combinatorial description to compute, for example, the number of BLZ polynomials of momentum $0$ and level $n$.
\begin{lemma}\label{lem:blzatbeta0}
The number of BLZ polynomials of momentum $0$ and level $n$ obtained via Crum--Darboux transformations is equal to the number of self-conjugate partitions of $n$,
\begin{equation}\label{eq:beta=0selfconjugate}
        \# \bigl\lbrace \la \in \Partitions\,: \,|\la|=n, \;   \phi^{[0]}_{\la,\la'}(0) \neq0 \bigr\rbrace =
         \# \bigl\lbrace \la \in \Partitions\,: \,|\la|=n, \;   \la=\la' \bigr\rbrace .
    \end{equation}
    \end{lemma}
    \begin{proof}
First we claim that if $\phi^{[0]}_{\la,\mu}(0)\neq 0$ then $\la=\mu$. Indeed, by construction, if a rational extension with $\beta=0$ is given, the spectral determinants $\pm$ 
coincide, $Q^{+}_{\la,\mu}=Q^{-}_{\la,\mu}$. Moreover, if $\phi^{[0]}_{\la,\mu}(0)\neq 0$ then, by Theorem \ref{thm:lamuextension}(II), $\mathscr{E}^{(0)}_{+,\la,\mu}= \mathscr{E}^{(0)}_{-,\la,\mu}$; whence $\la=\mu$. 
Now we show that if $\la=\la'$ then $\phi_{\la,\la'}^{[0]}(0) \neq 0$. If $\la=\la'$ then $\mathscr{E}^{(0)}_{+,\la}= \mathscr{E}^{(0)}_{-,\la'}$. Therefore $\mathscr{E}^{(0)}_{+,\la} + \mathscr{E}^{(0)}_{-,\la'}$ is a multiset where each element has multiplicity $2$.
Such a multiset admits a unique decomposition into $2$ multiset where each point has multiplicity $1$.
Therefore, by the criterion in Remark \ref{rem:equivalence}, $\phi_{\la,\la'}^{[0]}(0) \neq 0$.
\end{proof}

We prove Theorem \ref{thm:lamuextension} via iterated Darboux transformation on the harmonic oscillator equation, which we discuss in the next section. 

\subsection{Harmonic oscillator and Crum--Darboux transformations}
As we recalled in the definition of the Laguerre polynomials, the functions $\psi_{\pm}^{(\beta)}(n,x)$ are the eigenfunctions for the $\pm$ spectrum of the harmonic oscillator, provided that $\mp \beta$ is not a positive integer. The corresponding eigenvalues are $  E^{(\beta)}_\pm(n)= 4k+2\pm 2\beta $, hence the corresponding spectra are 
\begin{equation}\label{eq:harmonicspectra}
 \mathscr{E}_{\pm}^{(\beta)}=\left\{4k +2 \pm 2\beta, \, k \in \N \right\}.
\end{equation}
The latter is the reduction of the general formula (\ref{eq:Eplamu},\ref{eq:Emlamu}), that we aim to prove, in the case that $\la$ and $\mu$ are the empty partitions. In the next proposition, we  compute explicitly the spectral determinants $Q_{\pm}(E;\beta)$ (and the Stokes multiplier $T(E,\beta)$) of the harmonic oscillators, from which the recalled statements follow immediately.
\begin{proposition}\label{prop:harmonicspectrum}
If $\mp \beta  $ is not a positive integer, the spectral determinant $Q_{\pm}$ is given by the following formula
\begin{equation}\label{eq:Qpmhar}
    Q_{\pm}(E;\beta)= \frac{ \, \Gamma(1\pm \beta)}{\Gamma\bigl(\frac{2\pm2\beta-E}{4}\bigr)} .
\end{equation}
The Stokes multiplier $T$ is given by the following formula
\begin{equation}\label{eq:Thar}
    T(E;\beta)=  \frac{1}{\Gamma\bigl(\frac{2+2\beta-E}{4}\bigr) \, \Gamma\bigl(\frac{2-2\beta-E}{4}\bigr)} .
\end{equation}

\begin{proof} Formula \eqref{eq:Qpmhar} is a classic result. It can be found in \cite[Chapter 14]{nist} and the proofs in references therein. For completeness, we show how it can be derived. 
Assume that $\beta \notin \Z$.
 Under the transformation
\begin{align} \label{eq:harmonictoconfluent}
    u(z)\,=\,\e^{\frac12 z} \,z^{\frac14-\frac{b}{2}} \,\psi(z^{\frac12}) \, , \quad b=1+\beta\,,\quad a=\frac12+\frac{1}{2}\beta -\frac{1}{4} E
\end{align}
the equation $ \Hb \psi = E\,\psi $ becomes the confluent hypergeometric equation,
\begin{equation*}
    z\,u''(z)+(b-z)\,u'(z)-a \, u(z)\,=\,0.
\end{equation*}
Comparing the asymptotic behaviour of the solution $\psi_0$ \eqref{eq:psizero} and $\chi_{\pm}$ \eqref{eq:chipm} with the asymptotic behaviour of the solutions to the latter equation, we deduce that 
$\psi_0, \chi_{\pm}$ are mapped to the standard Tricomi and K\"ummer solutions, see \cite[Chapter 14]{nist} and \cite{tricomi47},
\begin{align*}
    \varphi_0 \mapsto U(a,b,z) \, , \; \chi_+ \mapsto M(a,b,z) \, , \; \chi_- \mapsto z^{1-b} M(a-b+1,2-b,z),
\end{align*}
which satisfy the relation 
\begin{equation*}
    U(a,b,z)= \frac{\Gamma(b-1)}{\Gamma(a)} z^{1-b} M(a-b+1,2-b,z) + \frac{\Gamma(1-b)}{\Gamma(a-b+1)} M(a,b,z).
\end{equation*}
From the latter relation, using \eqref{eq:coeffQ} and \eqref{eq:harmonictoconfluent}, we deduce the \eqref{eq:Qpmhar} restricted to the case $\beta \notin \Z$.
By Lemmas \ref{lem:frobenius} and \ref{lem:sibuya}, $Q_{\pm}(E;\beta)$ is analytic at $\pm \N$. Hence, \eqref{eq:Qpmhar} is valid also in the case $\pm \beta \in \N$.

Formula \eqref{eq:Thar} follows from \eqref{eq:T=QQ} and \eqref{eq:Qpmhar}, using the Euler reflection formula for the $\Gamma$ function.
\end{proof}
\end{proposition}
\begin{remark}
    Using the Euler reflection formula for the $\Gamma$ function, one verifies directly that the functions $Q_{\pm}(E;\beta)$
    defined in \eqref{eq:Qpmhar}, satisfy the $QQ$ relations \eqref{eq:QQrelationscorrected}.
\end{remark}

We briefly recall the notion of a Darboux transformation, see \cite{crum,Darboux}.
\begin{lemma}[Darboux~\cite{Darboux}]
 Let $a(x),b(x)$ be solutions of the Schr\"odinger equations
 \be
 a''(x)=\bigl(u(x)-E\bigr) a(x), \qquad b''(x)=u(x) b(x),
 \ee
  with $a(x)$ not identically zero.
 Then
 \be
 \tilde{b}(x)\,=\, \biggl(b'(x)-\frac{a'(x)}{a(x)} b(x)\biggr)\,=\, \frac{\Wr x\bigl(a(x),b(x)\bigr)}{a(x)}
 \ee
 solves the differential equation
 \be
 \psi''(x)\,=\, \left(u(x)-2\frac{\d^2}{\d x^2} \ln a(x)\right) \psi(x).
 \ee
 \end{lemma}
The above Lemma  describes the action of a single Darboux transformation. Assume now that we are given 3 non-trivial solutions,
$a_1,a_2,b$ of the differential equations $a_i''(x)=\big( u(x)-E_i \big) a_i(x)$ and $b''(x)= u(x) b(x)$. One can use the function $a_1$ to make a Darboux transformation both on $a_2$ and on $b$. By the above lemma, the functions obtained, $\tilde{a}_2$ and $\tilde{b}$, solve the same differential equation but for a shift in the energy. Therefore, we can use the solution $\tilde{a}_2$ to make a further Darboux transformation on $\tilde{b}$. This process can be generalized to a collection of $n+1$ solutions, $n$ of which, $a_1,\dots, a_n$, acts iteratively on the last one $b$. Next lemma, due to Crum~\cite{crum}, describes explicitly the action of $n$ iterated Darboux transformations, which we call Crum--Darboux transformations.
\begin{lemma}[Crum~\cite{crum}]\label{lemma:Crum}
 Let $a_1(x),\dots, a_n(x),b(x)$ solutions of the Schr\"odinger equation
 \be
 a_i''(x)=\left( u(x)-E_i \right) a_i(x), \quad b''(x)=u(x) b(x)
 \ee
 such that $a_1(x),\dots, a_n(x)$ are linearly independent, i.e.
 \be
W_n(x):=\Wr x\bigl(a_1(x),\dots, a_n(x)\bigr)
 \ee
 is not identically zero.
 The action of $n$ iterated Darboux transformations, of the solutions $a_1,\dots,a_n$ on the solution $b$, is given by the function
 \be
 \Phi_{n,b}(x)= \frac{W_{n,b}(x)}{W_n(x)}, \quad W_{n,b}(x):= \Wr x\bigl(a_1(x), \dots,a_n(x),b(x)\bigr),
 \ee
 which solves the differential equation
 \be
 \psi''(x)= \left(u(x)-2\frac{\d^2}{\d x^2} \ln W_n(x)\right) \psi(x).
 \ee
 \end{lemma}
 We now apply the Crum--Darboux theory to the harmonic oscillators.
 Given a triple $(\beta,\la,\mu) \in \C \times \Partitions \times \Partitions$, we
 consider the harmonic oscillator with momentum $\beta-r+s$, where $r=l(\la), s=l(\mu)$. Recall from Section \ref{sec:laguerre} that the Laguerre Wronskian associated to the triple $(\beta,\la,\mu)$ is the Wronskian of the functions
\begin{align}
   \left\{ \psi_+^{(\beta-r+s)}(m_1,x),  \dots, \psi_+^{(\beta-r+s)}(m_r,x), \psi_-^{(\beta-r+s)}(n_1,x),  \dots,
    \psi_+^{(\beta-r+s)}(n_s,x) \right\},
     \label{eq:eigenproof}
\end{align}
where ${\bs m}=\la+{\bs\delta}_r$,${\bs n}=\mu+{\bs\delta}_s$, $\bs\delta_r=(r-1,\dots1,0)$, and $\bs \delta_s=(s-1,\dots,1,0)$. Notice that the above eigenfunctions are linearly dependent if and only if two of them coincide, or, equivalently, if and only if the factor $\kappa_{\bs m,\bs n}^{(\beta-r+s)}$ vanishes, see \eqref{eq:kappamna}.

\begin{lemma}
Let $(\beta,\la,\mu) \in \C \times \Partitions \times \Partitions$ be such that the factor $\kappa_{\la+{\bs\delta}_r,\mu+{\bs\delta}_s}^{(\beta-r+s)}$, defined in \eqref{eq:kappamna}, does not vanish.

The Crum--Darboux action of the functions \eqref{eq:eigenproof} on a solution of the differential equation $\mathscr{H}^{(\beta-r+s)} \psi= E \psi$ is solution to the differential equation
\begin{equation}
  \Hb_{\la,\mu} \tilde{\psi}= \big(E-2r-2s\big)\tilde{\psi},
\end{equation}
where  $\Hb_{\la,\mu}$ is as per Definition \ref{def:Hblamu}.
\end{lemma}
\begin{proof}
According to Crum's lemma (Lemma~\ref{lemma:Crum}), starting from a solution to the equation
$\mathscr{H}^{(\beta-r+s)} \psi= E \psi$,
we obtain a solution to the differential equation
\begin{align*}
  \tilde{\psi}''= & \left( x^2+ \frac{(\beta-r+s)^2-\frac14}{x^2}- E
 -2 \frac{\d^2}{\d x^2} \ln \Psi^{(\beta-r+s)}_{\la+ {\bs{\delta}}_r,\mu+ {\bs{\delta}}_s}(x) \right)\tilde{\psi},
\end{align*}
where $ \Psi^{(\beta-r+s)}_{\la+ {\bs{\delta}}_r,\mu+ {\bs{\delta}}_s}(x)$ is as per \eqref{eq:Psimnax}.
Using the factorization \eqref{eq:PsiPhi} of $\Psi^{(\beta-r+s)}_{\la+ {\bs{\delta}}_r,\mu+ {\bs{\delta}}_s}(x)$, and the definition of $\Phi^{(\beta)}_{\la,\mu}(x^2)$ \eqref{eq:defPlambdamu}, the right-hand-side of the above equation is equal to,
\begin{align*}
  x^2+ \frac{\beta^2-\frac14}{x^2}- E+2(r+s) -2 \frac{\d^2}{\d x^2} \ln \Phi^{(\beta)}_{\la,\mu}(x^2) .
\end{align*}
According to
Definition~\ref{def:Hblamu},
 $x^2+ \frac{\beta^2-\frac14}{x^2} -2 \frac{\d^2}{\d x^2} \ln \Phi^{(\beta)}_{\la,\mu}(x^2)$ is  the potential of the operator $\mathscr{H}^{(\beta)}_{\la,\mu}$.
\end{proof}

We have thus the following definition.
\begin{definition}
Let $(\beta,\la,\mu) \in \C \times \Partitions \times \Partitions$ be such that the factor $\kappa_{\la+{\bs\delta}_r,\mu+{\bs\delta}_s}^{(\beta-r+s)}$, defined in \eqref{eq:kappamna}, does not vanish.
We denote by $V^{(\beta)}_{\la,\mu}(E)$ the space of solutions of $\Hb_{\la,\mu} \psi= E \psi$ and we define
the following linear map, called Crum--Darboux operator, 
\begin{align} \label{eq:DarbouxCrum}
   &\mathcal{D}^{(\beta)}_{\la,\mu}: V^{(\beta-r+s)}(E) \to V^{(\beta)}_{\la,\mu}(E-2r-2s)\\  \nonumber
  & \psi \mapsto Wr_x\left[ 
   \left\{ \psi_+^{(\beta-r+s)}(m_1,x),  \dots, \psi_+^{(\beta-r+s)}(m_r,x), \psi_-^{(\beta-r+s)}(n_1,x),  \dots, \right. \right.
    \\ \nonumber
 & \left. \left. \qquad   \qquad \qquad \dots,   \psi_+^{(\beta-r+s)}(n_s,x),\psi(x) \right\} \right]
    \left/ \Psi^{(\beta-r+s)}_{\la+ {\bs{\delta}}_r,\mu+ {\bs{\delta}}_s}(x) \right. 
\end{align}
 with ${\bs m}=\la+{\bs\delta}_r$ and ${\bs n}=\mu+{\bs\delta}_s$. 
\end{definition}

\begin{proposition}\label{prop:darbouxlaguerre}
Let $\la,\mu \in \Partitions$ be partitions of length $r,s$, and $\beta \in \C$ be such that the factor $\kappa_{\la+{\bs\delta}_r,\mu+{\bs\delta}_s}^{(\beta-r+s)}$, defined in \eqref{eq:kappamna}, does not vanish.
The following properties holds
\begin{enumerate}[leftmargin=*]
    \item 
    \begin{enumerate}
        \item If $E= 4k+2(\beta-r+s)+2 $ with $k \in \lambda+{\bs\delta}_r$, $\ker D_{\la,\mu}^{(\beta)} = \C \{ \psi_+^{(\beta-r+s)}(k,x)\}$.
        \item  If $E= 4k-2(\beta+r-s)+2 $ with $k \in \mu+{\bs\delta}_s$, $\ker D_{\la,\mu}= \C \{ \psi_-^{(\beta-r+s)}(k,x)\}$.
   \item In all other cases $\ker D_{\la,\mu}(\beta)= \{0\}$.
    \end{enumerate}
\item For all $\psi \in V^{(\beta-r+s)}(E) $, $\mathcal{D}_{\la,\mu}^{(\beta)} \psi$ is a meromorphic function on the universal cover of $\C^*$.
 \item $\Hb_{\la,\mu}$ is a rational extension of the harmonic oscillator.
    \item Assume that $\beta \notin \Z$. Let $\chi_{\pm}^{(\beta)}(x,E)$ be the Frobenius solutions and $\varphi_0^{(\beta)}(x,E)$ be the subdominant solution at $+\infty$ of the harmonic oscillator. Let
    $\chi_{\pm,\la,\mu}^{(\beta)}(x,E)$ and $\varphi_{0,\la,\mu}^{(\beta)}(x,E)$ the Frobenius solutions and the subdominant solution at $+\infty$ of the rational extension $\mathscr{H}_{\la,\mu}^{(\beta)}$. The following formulas hold,
    \begin{align}\nonumber
       & D_{\la,\mu}^{(\beta)} \varphi_0^{(\beta-r+s)}(x,E+2r+2s)  = 2^{-r-s} 
        \prod_{k\in \lambda +{\bs\delta}_r}\big(E-E^{(\beta)}_+(k-r) \big)  \\ \label{eq:darbouxpsi}
        & \qquad \qquad \qquad \qquad \times  \prod_{k\in \mu +{\bs\delta}_s} \big(E-E_-^{(\beta)}(k-s) \big) \; \varphi_{0,\la,\mu}^{(\beta)}(x,E)  \\ \nonumber
       &  D_{\la,\mu}^{(\beta)} \chi_+^{(\beta-r+s)}(x,E+2r+2s) = (-1)^{r} 2^{s-r} \frac{\beta-r+s}{\beta} \frac{\Gamma(\beta-r+s)}{\Gamma(\beta)}  \\ \label{eq:darbouxchip}
       & \qquad \qquad \qquad \qquad \qquad \times \prod_{k\in \lambda +{\bs\delta}_r}\big(E-
       E_+^{(\beta)}(k-s) \big) \, \chi_{+,\la,\mu}^{(\beta)}(x,E) \\ \nonumber
      &  D_{\la,\mu}^{(\beta)} \chi_-^{(\beta-r+s)}(x,E+2r+2s)  = (-1)^{r} 2^{r-s}\frac{\beta-r+s}{\beta} \frac{\Gamma(-\beta+r-s)}{\Gamma(-\beta)} \\  \label{eq:darbouxchim}
       & \qquad \qquad \qquad \qquad \qquad \times \prod_{k\in \mu +{\bs\delta}_s}\big(E-E_-^{(\beta)}(k-s)\big) \,\chi_{-,\la,\mu}^{(\beta)}(x,E)  
    \end{align}
    with $E^{(\beta)}_{\pm}(k)=4k\pm2\beta+2$ as per \eqref{eq:eigenharm}.
\end{enumerate}
\end{proposition}
\begin{proof}
\begin{enumerate}[leftmargin=*]
\item A solution $\psi$ is annihilated by the Crum--Darboux operator if and only if it coincides, up to a scalar multiple, with one of the functions of the set \eqref{eq:eigenproof}. The energy level associated to each of these eigenfunctions is, as per Proposition \ref{prop:harmonicspectrum}: $E= 4k+2(\beta-r+s)+2 $ for $\psi_+^{(\beta-r+s)}(k,x)$ with $k \in \lambda+{\bs\delta}_r$, and $E= 4k-2(\beta+r-s)+2 $ for $\psi_-^{(\beta-r+s)}(k,x)$ with $k \in \mu+{\bs\delta}_s$.
\item The thesis is proven by induction on the number of Darboux transformations. Every solution of the harmonic oscillator equation is a meromorphic (actually, holomorphic) function on the universal cover of $\C^*$ since $x=0$ is the only singular point. The induction step is proven by the following argument: if $a,b$ are meromorphic functions on the universal cover of $\C^*$, then the Darboux transformation of $a$ on $b$ is again a meromorphic function of the universal cover of $\C^*$.
\item The thesis holds if and only if for every $E \in \C$ every solution of the equation $\Hb_{\la,\mu} \psi= (E-2r-2s) \psi$ extends to a meromorphic function on the fundamental cover of $\C^*$. It is well-known that the latter requirement is equivalent to the following one \cite{duistermaat}: there exists an open $D \subset \C$ such that for every $E \in D$  every solution of the equation $\Hb_{\la,\mu} \psi= (E-2r-2s) \psi$ extends to a meromorphic function on the fundamental cover of $\C^*$. Due to 2) every solution of the equation $\Hb_{\la,\mu} \psi= (E-2r-2s) \psi$ extends to a meromorphic function on the fundamental cover of $\C^*$, if the kernel of $D_{\la,\mu}^{(\beta)}(E)$ is trivial. Due to 1), the kernel is trivial on an open set.  Hence the thesis is proven.
\item Let $\varphi_0(x,E,\beta), \varphi_0(x,F,\beta)$ be the normalized subdominant solutions at $+\infty$ for some unspecified rational extension of the harmonic oscillator with energy $E,F$.  The Darboux action of $ \varphi_0(x,F,\beta)$ on $\varphi_0(x,E,\beta)$ yields a further subdominant solution at $+\infty$, which is however in general not-normalized. In fact, using  \eqref{eq:psizero} and \eqref{eq:psizeroder}, we get
\be
 \varphi_0'(x,E,\beta)-  \left( \frac{\varphi_0'(x,F,\beta)}{\varphi_0(x,F,\beta)}\right)\varphi_0(x,E,\beta)  =   \frac{E-f}{2}  
   x^{\frac{-1+ E -2 }{2}} \e^{-\frac{x^2}{2}} \big(1+ o(1) \big), \ \  x \to +\infty.
\ee
Now let us consider the action of a normalized Frobenius solution $\chi_{\pm}(x,F,\beta)$ on another normalized Frobenius solution $ \chi_{\pm}(x,E,\beta)$. Using \eqref{eq:chipm}, we get
\be\begin{aligned} 
  \chi_+'(x,E,\beta)- \left( \frac{\chi_+'(x,F,\beta)}{\chi_+(x,F,\beta)}\right) \chi_+(x,E,\beta)  = & \frac{F-E}{2(\beta+1)}
   x^{\frac{1}{2}+(\beta+1)} \biggl(1+ \sum_{k \geq 1} a_k x^{2k} \biggr),
   \\
  \chi_-'(x,E,\beta)- \left( \frac{\chi_+'(x,F,\beta)}{\chi_+(x,F,\beta)}\right)\chi_-(x,E,\beta)  = & -2\beta \, x^{\frac{1}{2}-(\beta+1)} \biggl(1+ \sum_{k \geq 1} b_k x^{2k} \biggr),
  \\
 \chi_+'(x,E,\beta)-  \left( \frac{\chi_-'(x,F,\beta)}{\chi_-(x,F,\beta)}\right)\chi_+(x,E,\beta)  = & 2\beta\, x^{\frac{1}{2}+(\beta-1)} \biggl(1+ \sum_{k \geq 1} c_k x^{2k} \biggr),
 \\
 \chi_-'(x,E,\beta)- \left( \frac{\chi_-'(x,F,\beta)}{\chi_-(x,F,\beta)}\right)\chi_-(x,E,\beta)  = & \frac{E-F}{2(\beta-1)} 
 x^{\frac{1}{2}-(\beta-1)}  \biggl(1+ \sum_{k \geq 1} d_k x^{2k} \biggr).
\end{aligned}\ee
Since, by hypothesis $\beta \notin \Z$, all above formulas are well defined. From those, we deduce that the action of a  Frobenius function of type $\sigma, \sigma\in \{ +,-\}$ on a normalized Frobenius function of type $\tau$ is another, non-normalized, Frobenius function of type $\tau$; moreover, the momentum is shifted, $\beta \to \beta+\sigma$.

Each function in the set \eqref{eq:eigenproof} is both a subdominant and a Frobenius solution, therefore the action of the Crum--Darboux operator on a subdominant or a Frobenius solution is again a subdominant or a Frobenius solution. The latter is however not necessarily normalized. Collecting all normalizing factors form above formulae, we obtain the thesis.
\end{enumerate}
\end{proof}

Theorem \ref{thm:lamuextension} is now a corollary of Proposition \ref{prop:darbouxlaguerre}.
\begin{proof}[Proof of Theorem~\ref{thm:lamuextension}]
I. Assuming that the factor $\kappa_{\la+{\bs\delta}_r,\mu+{\bs\delta}_s}^{(\beta-r+s)}$, defined in \eqref{eq:kappamna}, does not vanish, the thesis is proven in Proposition \ref{prop:darbouxlaguerre}(3). \newline
We now deal with the case $\kappa_{\la+{\bs\delta}_r,\mu+{\bs\delta}_s}^{(\beta-r+s)} = 0$. Since $\mathscr{H}^{(\beta)}_{\la,\mu} $ only depends on the equivalence class of the triples $(\beta,\la,\mu)$ and in every equivalence class there exists a reduced triple, Corollary \ref{cor:coalescence}, we can assume, without loss of generality, that $\phi^{[0]}_{\la,\mu}(\beta) \neq 0$.
Now let $x^*$ be a zero $\Phi^{(\beta)}_{\la,\mu}$ and fix an annulus $C \subset \C^*$ centered at $x^*$. 
 We notice that the family  of polynomials $\Phi_{\la,\mu}^{(\wt \beta)}$ is a polynomial of $\wt \beta$ and that, fixed $\la,\mu$, the set of $\wt \beta$ such that $\kappa_{\la+{\bs\delta}_r,\mu+{\bs\delta}_s}^{(\beta-r+s)} = 0$ is a finite subset of $\Z$. Therefore, the annulus does not intersect any zero of $\Phi^{(\wt\beta)}_{\la,\mu}$ for any $\wt \beta$ in a small enough neighborhood $U$ of $\beta$.
Hence, one can define a basis of solutions of $\mathscr{H}_{\la,\mu}^{(\wt \beta)}\psi(x)=E \psi(x)$ on a disc $D \subset C$ , which is analytic with respect to $\wt \beta \in U$. 
For every solution $\psi(x,\wt \beta)$ of such a basis, one consider the analytic function $F(x,\wt\beta): \psi(x,\wt \beta)-\psi^*(x,\wt \beta)$ with $\psi^*$ the analytic continuation  of $\psi$ along the annulus. $F(\cdot,\wt{\beta})$ is again a solution of the differential equation and it vanishes identically on $D$ if an only if  the monodromy of $\psi$ about $x^*$ is trivial.  By construction, $F$ is an analytic function on $D \times U$. Moreover, it vanishes identically for every $\tilde{\beta}$ in an open subset of $U \setminus \beta$ since, as we have proven above, the monodromy is trivial whenever $\wt \beta \notin \Z$. Hence, by analiticity, it vanishes also at $\wt \beta=\beta$.   We deduce that, for every $E$, every solution of
$\mathscr{H}_{\la,\mu}^{( \beta)}\psi(x)=E \psi(x)$ has trivial monodromy about 
 any zero of $\Phi^{(\beta)}_{\la,\mu}$. Hence, by definition of rational extension, $\Hbp$ is a rational extension.

II. We notice that it is sufficient to prove the thesis when $\beta \notin \Z$. In fact, by Lemma \ref{lem:spectralanalityc}, we can extend the result to the case $\pm\beta$  not positive integer.
Assuming $\beta$ is not an integer, $\chi_{\pm}^{\beta-r+s}(E+2r+2s)$ is a basis of solutions of the harmonic oscillator with momentum $\beta-r+s$ and energy $E+2r+2s$. Whence, using \eqref{eq:coeffQ} and \eqref{eq:Qpmhar}, we compute that
\begin{align*}
 &  \varphi_0^{(\beta-r+s)}(x,E+2r+2s)=\frac{1}{2(\beta-r+s)} \left( \frac{2 \Gamma(1-\beta)}{\Gamma(\frac{2-2\beta-E}{4}-s)} \chi_+^{(\beta-r+s)}(x,E+2r+2s) -  \right. \\ & \left. \qquad \qquad \qquad \qquad \qquad \qquad \qquad \qquad \qquad \frac{2 \Gamma(1+\beta)}{\Gamma(\frac{2+2\beta-E}{4}-r)} \chi_-^{(\beta-r+s)}(x,E+2r+2s) \right).
\end{align*}
Applying the Crum--Darboux map $ D_{\la,\mu}^{(\beta)} $ on both side of the above equality, and using Proposition \ref{prop:darbouxlaguerre}(4), we obtain formulas (\ref{eq:Qplamu},\ref{eq:Qmlamu}). Since the spectra coincide with the set of zeros of the spectral determinants and the zeros  of $1/\Gamma(-z)$ are the natural numbers, we deduce (\ref{eq:Epla},\ref{eq:Emla}). \newline 

III. Follows from \eqref{eq:T=QQ}, \eqref{eq:Qplamu} and \eqref{eq:Qmlamu}.

IV. The fact that $(\beta,\la,\mu) \approx (\wt \beta, \wt \lambda, \wt \mu)$ implies
$\mathscr{H}_{\la,\mu}^{(\beta)}=\mathscr{H}_{\wt\la,\wt\mu}^{(\wt \beta)}$ was already established, in Corollary  \ref{cor:coalescence} and  Remark \ref{remark:equivalentextensions}. 

We prove that $\mathscr{H}_{\la,\mu}^{(\beta)}=\mathscr{H}_{\wt\la,\wt\mu}^{(\wt \beta)}$ implies $(\beta,\la,\mu) \approx (\wt \beta, \wt \lambda, \wt \mu)$. We can assume, without loss in generality, that both triples are reduced. Hence, by II, i)$Q_{+,\la,\mu}^{(\beta)}=Q_{+,\wt\la,\wt\mu}^{(\wt\beta)}$ holds or ii) i)$Q_{-,\la,\mu}^{(\beta)}=Q_{-,\wt\la,\wt\mu}$ or iii)$Q_{+,\la,\mu}^{(\beta)}=
Q_{-,\wt\la,\wt\mu}^{(\wt\beta)}$ or iv)$Q_{-,\la,\mu}^{(\beta)}=
Q_{+,\wt\la,\wt\mu}^{(\wt\beta)}$ holds. Hence, i) $\mathscr{E}^{(\beta)}_{+,\la,\mu}=\mathscr{E}^{(\wt \beta)}_{-,\wt \la,\wt \mu}$ or
ii)  $\mathscr{E}^{(\beta)}_{-,\la,\mu}=\mathscr{E}^{(\wt \beta)}_{-,\wt \la,\wt \mu}$ holds or
iii) $\mathscr{E}^{(\beta)}_{+,\la,\mu}=\mathscr{E}^{(\wt \beta)}_{-,\wt \la,\wt \mu}$ holds or iv)$\mathscr{E}^{(\beta)}_{-,\la,\mu}=\mathscr{E}^{(\wt \beta)}_{-,\wt \la,\wt \mu}$ holds.  Moreover, by III,  $T^{(\beta)}_{\la,\mu} = T^{(\wt \beta)}_{\wt\la,\wt\mu}$. Hence, 
 $ \mathscr{E}^{(\beta)}_{+,\la,\mu} + \mathscr{E}^{(\beta)}_{-,\la,\mu} =  \mathscr{E}^{(\wt\beta)}_{+,\wt\la,\wt\mu} + \mathscr{E}^{(\wt\beta)}_{-,\wt\la,\wt\mu}  $, where $+$ denotes the sum of multisets. Hence, i) or ii) implies that $\mathscr{E}^{(\beta)}_{+,\la,\mu}=\mathscr{E}^{(\wt \beta)}_{-,\wt \la,\wt \mu}$ and
 $\mathscr{E}^{(\beta)}_{-,\la,\mu}=\mathscr{E}^{(\wt \beta)}_{-,\wt \la,\wt \mu}$, while iii) or iv) implies that $\mathscr{E}^{(\beta)}_{+,\la,\mu}=\mathscr{E}^{(\wt \beta)}_{-,\wt \la,\wt \mu}$ and $\mathscr{E}^{(\beta)}_{-,\la,\mu}=\mathscr{E}^{(\wt \beta)}_{-,\wt \la,\wt \mu}$.
In the first case, $(\beta,\la,\mu)= (\wt\beta,\wt\lambda, \wt \mu)$, in the second case $(\beta,\la,\mu)= (-\wt\beta,\wt\mu, \wt \lambda) $.

V. As we have already proven, if $\beta=\tilde{\beta}, \; \la = \tilde{\la} \mbox{ and } \mu=\tilde{\mu} \mbox{ or } \beta=-\tilde{\beta}, \;  \la = \wt \mu \mbox{ and } \mu=\wt{\la}$ then $\Phi_{\la,\mu}^{(\beta)}=\Phi_{\wt \mu, \wt \la}^{(\wt \beta)} $.
Here, we prove the following, from which both statements in the thesis follow directly:
If $\Phi_{\la,\mu}^{(\beta)}=\Phi_{\wt \la, \wt \mu}^{(\wt \beta)} $ and $\phi^{[0]}_{\la,\mu}(\beta) \neq 0$ then 
$\beta=\tilde{\beta}, \; \la = \tilde{\la} \mbox{ and } \mu=\tilde{\mu} \mbox{ or } \beta=-\tilde{\beta}, \;  \la = \wt \mu \mbox{ and } \mu=\wt{\la}$.
Since $\phi^{[0]}_{\la,\mu}(\beta) \neq 0$ then $\phi^{[0]}_{\wt\la,\wt\mu}(\wt\beta) \neq 0$.
Then both polynomials provides rational extensions of the harmonic oscillator with momentum $\beta$ and with momentum $\wt \beta$. Hence, by system \eqref{eq:calogeromoser}, $\beta^2=\wt \beta^2$.
Hence, $\mathscr{H}_{\la,\mu}^{(\beta)}=\mathscr{H}_{\wt\la,\wt\mu}^{(\wt \beta)}$. Hence, by IV, $(\beta,\la,\mu) \approx (\wt\beta,\wt\lambda, \wt \mu)$. By Corollary \ref{cor:coalescence}, the only triples equivalent to  $(\beta,\la,\mu)$ under the equivalence relation $\approx$ and reduced (namely such that $\phi^{[0]}$ does not vanish) are $(\beta,\la,\mu)$ and $(-\beta,\mu,\la)$.

VI. By Theorem~\ref{thm:symmetry}, $\Phi^{(\beta)}_{\la,\mu}(-y)=\Phi^{(\beta)}_{\mu',\la'}$ and, by IV, $\Phi^{(\beta)}_{\mu',\la'}=\Phi^{(\beta)}_{\la,\mu}$ if and only if $\la'=\mu$.

VII. By definition $\mathcal{P}$ is a BLZ polynomial if and only if $\mathcal{P}(0) \neq 0$ and $\mathcal{P}(i x)=\mathcal{P}(x)$. The thesis then follows from V. Formulas (\ref{eq:Qpla},\ref{eq:Qmla},\ref{eq:Epla},\ref{eq:Emla}) are the restriction of formulas (\ref{eq:Qplamu},\ref{eq:Qmlamu},\ref{eq:Eplamu},\ref{eq:Emlamu}) in the case $\mu=\la'$.  
\end{proof}

\section{ODE/IM Correspondence at the Free-Fermion Point}\label{sec:ODEIM}

\subsection{Completeness}
In this section we solve the QQ system \eqref{eq:QQrelationscorrected} for the quantum KdV model at the free fermion point and we show that the ODE/IM correspondence is complete.
Recall that the QQ system with momentum $\beta$ ($\beta \notin \Z$) is the functional equation for two entire functions $(Q_+(E),Q_-(E))$
\begin{equation*}
 \e^{ \i  \frac{\beta \pi}{2}}\,Q_+(\i E;\beta) \,Q_-(-\i E;\beta)\,-\, \e^{- \i  \frac{\beta \pi}{2}} \,Q_+(- \i E;\beta) \,Q_-(\i E;\beta)\,=\, \i\, \e^{\frac{E \pi}4}, \quad \forall E \in \C.
\end{equation*}
We start our study by noticing that
the $QQ$ system is invariant under the transformation 
\begin{equation}\label{eq:Qgauge}
    \big(Q_+(E),Q_-(E)\bigr) \mapsto \bigl(e^{f(E)}Q_+(E),e^{-f(-E)}Q_-(E)\bigr), \;  f \; \mbox{  an entire function}.
\end{equation}
Such a transformation does not alter the roots of the solutions. At the level of BLZ potentials, it corresponds to a simple change of normalization of the
Frobenius solutions $(\chi_+,\chi_-) \to (e^{-f(E)}\chi_+,e^{f(-E)}\chi_-)$.
\begin{definition}
    We say that the ordered pair of entire functions $(Q_+,Q_-)$ is a solution of the QQ system if it satisfies \eqref{eq:QQrelationscorrected} identically. We say that two solutions the QQ system are equivalent if and only if they are related by the transformation \eqref{eq:Qgauge}.
\end{definition}
Even though the $QQ$ system is a system for two entire functions, one member of the pair completely determines the other. In fact, we have the following lemma.
 \begin{lemma}\label{lem:QQinjective}
     Assume  that $ \beta \notin  \Z $. Let four entire functions be given $Q_+,\widetilde{Q}_+,Q_-, \widetilde{Q}_-$ such that pairs $(Q_+,Q_-)$, $(\widetilde{Q}_+,Q_-)$, and $(Q_+,\widetilde{Q}_-)$  satisfy the QQ system \eqref{eq:QQrelationscorrected}. Then $Q_+=\widetilde{Q}_+$ and $Q_-=\widetilde{Q}_-$.
     \begin{proof}
    We prove that $Q_-=\widetilde{Q}_-$ as the other equality is obtained with the same argument simply exchanging $+$ and $-$.
    Let  $G=Q_--\widetilde{Q}_-$. $G$ is an entire function. Using the QQ system, we deduce that it satisfies the functional equation 
     \begin{align*}
 & \e^{ \i  \frac{\beta \pi}{2}}Q^{+}(\i E) G(-\i E)- \e^{- \i  \frac{\beta \pi}{2}} Q^+(- \i E) G(\i E)=  0, \forall E \in \C.
\end{align*}
Therefore,
   \begin{align*}
 & Q^{+}(\i E) G(-\i E)= \e^{- 2 \i  \beta \pi} Q^+(\i E) G(-\i E), \qquad \forall E \in \C.
\end{align*}
By hypothesis $ \e^{- 2 \i  \beta \pi} \neq 1$, therefore the entire function $Q^+(\i E) G(-\i E)=0$ vanishes identically. Moreover, the function $Q_+$ does not vanish identically because it satisfies the QQ system. Whence, the function $G(E)$ vanishes identically. 
     \end{proof}
 \end{lemma}
 From the QQ system one derives the Bethe equations, which are a system of equations satisfied by the solution of the $QQ$ system at their roots. The Bethe equations are in general a coupled nonlinear system in an infinite number of variables \cite{coma21} but
 they are uncoupled and essentially trivial at the free fermion point. 
\begin{lemma}
Assume $\beta \notin \Z$. Let $(Q_+,Q_-)$ be a solution of the $QQ$ system. The zeroes of the function $Q_{\pm}$ are a subset of $\{ E \in \C: \,E= 4k+2 \pm 2\beta, \; k \in \Z \}$.
\end{lemma}
\begin{proof}
Let indeed $E^*_{\pm}$ be such that $Q_{\pm}(E^*_{\pm})=0$. Evaluating \eqref{eq:QQrelationscorrected} at $\i E^*_{\pm}$ and
$-\i E^*_{\pm}$, we deduce that
\begin{equation}\label{eq:betheequations}
\e^{\i \pi E^*_{\pm}}=-\e^{\pm \i \pi \beta  }, \;  \forall E^*_{\pm}  \mbox{ such that } Q_{\pm}(E^*_{\pm})=0.
\end{equation}  
From this the thesis follows.
\end{proof}
System \eqref{eq:betheequations} are the \textit{Bethe equations} at the free-fermion point. 
By the above Lemma, we can attach to any equivalence class of solutions two sets of integers, the upper- and lower-Bethe numbers.
\begin{definition}
Let $(Q_+,Q_-)$ be a solution of the QQ system. The upper-Bethe-numbers is the set of those integers $k$ such that $E=4k+2\beta+2$ is a zero of $Q_+$. The lower-Bethe-numbers,  is the set of those integers $k$ such that $E=4k-2\beta+2$ is a zero of $Q_-$. The upper- and lower-Bethe-numbers are an invariant of the equivalence classes of solutions.
\end{definition}
The upper- and lower-Bethe-numbers of a solution of the $QQ$ system for quantum KdV are not allowed to be arbitrary set of integers \cite{BLZHigher}. In fact, for a pair $(Q_+,Q_-)$ of solutions to be admissible, at least one of the two sets of Bethe numbers must be bounded from below, hence it can be arranged into an increasing sequence  $\{n_k\}_{k \in \N}$, and, moreover, such a sequence must be stabilizing, namely, $n_k=k$ for $k$ large enough.
\begin{definition}[\cite{BLZHigher}]\label{def:bethenumbers}
 The pair of entire functions $(Q_+,Q_-)$ is a  solution of the $QQ$ system for the quantum KdV model at the free-fermion point with momentum $\beta \notin \Z$ if it satisfies the $QQ$ system, and if the upper-Bethe-numbers or the lower-Bethe-numbers forms a bounded increasing stabilizing sequence, which we denote as $\{n^+_k\}_{k \in \N}$ or $\{n^-_k\}_{k \in \N}$. 
\end{definition}
We can now state the problem of completeness of the ODE/IM correspondence, which was raised by  by Bazhanov--Lukyanov--Zamolodchikov in \cite{BLZHigher}, in the particular case $c=-2$, that is at the free fermion point
\footnote{We could have defined the completeness in terms of the Bethe equations alone, but we prefer the QQ system because it is a richer structure. We have to discard the case $\beta \in \Z$, because if $\beta \in \Z$ the $QQ$ system is not well-posed. However, the solutions of the Bethe equations for $\beta \in \Z$, i.e. either the zeroes of $Q_+$ or the zeroes of $Q_-$, can be recovered as a limiting case.}. 
\begin{definition}[\cite{BLZHigher}]
 The ODE/IM correspondence at the free-fermion point is complete if, for $\beta \notin  \Z$, every solution of the QQ system (of the quantum KdV model at the free-fermion point) can be expressed in terms of the spectral determinants of a rational extension of the harmonic oscillator associated to a BLZ polynomial.   
\end{definition}
We prove below that the ODE/IM correspondence is complete. We start by recalling a well-known fact:
Increasing stabilizing sequences are parametrized by partitions, via the map
\begin{equation}
\la \mapsto \{n^{\la}_k\}_{k\in \N}, \quad n^{\la}_k= k-\la'_{k+1}, \; \forall k \in \N,
\end{equation}
where $\la'$ is the partition conjugate to $\la$. The proposition below provides the general solution of the QQ system for quantum KdV at the free fermion point.
\begin{proposition}
Assume that $\beta \notin \Z$.
\begin{enumerate}[leftmargin=*]
\item For every $\la$, the pair $(Q^\la_+,Q^\la_-)$ as defined in equations (\ref{eq:Qpla},\ref{eq:Qmla}) defines a solution of the $QQ$-system.
\item If $\la \neq \mu$, $(Q^\la_+,Q^\la_-)$ and $(Q^\mu_+, Q^\mu_-)$ are unequivalent solutions.
    \item If $(Q_+,Q_-)$ is a solution then, up to the equivalence transformation \eqref{eq:Qgauge}, there exists a partition $\la$ such that the pair $(Q_+,Q_-)$ coincides with the pair $(Q^\la_+,Q^\la_-)$ defined by formulas 
(\ref{eq:Qpla},\ref{eq:Qmla}). Moreover, letting $\{n^+_k\}_{k\in \N}$ and $\{n^-_k\}_{k\in \N}$ denote the upper- and lower-Bethe-numbers of the solution $(Q_+,Q_-)$, we have that $n^+_k=n^\la_k$ and $n^-_k=n^{\la'}_k$ for all $k \in \N$.
\end{enumerate}
\begin{proof}
(1)  It follows directly from Lemma \ref{lem:QQsystemspectral} and Theorem \ref{thm:lamuextension}(V). \newline 
(2) By formulas (\ref{eq:Qpla}), the roots of $Q_+^{\la}$ do not coincide with the roots of $Q_+^{\mu}$ if $\la \neq \mu$. \newline
(3) We can assume without loss in generality that the positive-Bethe-numbers form an increasing stabilizing sequence (we can always reduce  to this case via the transformation $\beta \to -\beta, Q_+ \to Q_-$). In this case, necessarily  $n^+_k=n^\la_k$ for some partition $\la$. Therefore $Q_+$ and $Q_+^{\la}$ have the same roots. Therefore, $Q_+^{\la}(E)=Q_+(E) e^{f(E)}$ for an entire function $E$. Therefore, $(Q_+^{\la},e^{-f(-E)}Q_-(E))$ solves the QQ system.
By Lemma \ref{lem:QQinjective}, 
$e^{-f(-E)}Q_-(E)=Q_-^{\la}(E)$. Therefore $(Q_+,Q_-)$ and $(Q_+^{\la},Q_-^{\la})$ are equivalent solutions and the roots of $Q_-$ coincides with the roots of $Q_-^{\la}$. Therefore, $n^-_k=n^{\la'}_k$.
\end{proof}
\end{proposition}
As a corollary, we have the following theorem.
\begin{theorem}\label{thm:completeness}
    The ODE/IM correspondence at the free-fermion point is complete.
\end{theorem}

\subsection{Coalescence of zeros at the origin and reducibility of Virasoro Verma modules}

We study here some predictions of the ODE/IM correspondence for representations of the Virasoro algebra.
Let $V_{h}=\bigoplus_{n\geq 0}V^{[n]}_{h}$ be the Virasoro Verma module with central charge~$c=-2$ and highest weight $h \in \C$.
The quantum KdV model at the free-fermion point is constructed on $V_h$ with highest weight $h$ equal to $h(\beta)=\frac 18(\beta^2-1)$, where $\beta$ is the momentum~\cite{BLZ1}.
It is well-known that a Virasoro Verma module is irreducible if and only if it does not contain any singular vector.
For generic $\beta$,  the Virasoro Verma module is irreducible and $\dim V^{[n]}_{h(\beta)}=p(n) $.
If it is reducible, one defines the irreducible highest weight module (which is also graded) of highest weight $h(\beta)$ as the quotient of $V_{h(\beta)}$ by its maximal submodule $M_{h(\beta)}$~\cite{FeiginFuchs}
\begin{equation}\label{eq:irreduciblequotient}
    \wt V_{h(\beta)}:= V_{h(\beta)}/M_{h(\beta)}= \bigoplus_{n \in \N}  \wt V_{h(\beta)}^{[n]}.
\end{equation}
An important object of the theory is the dimension of the subspace of fixed degree, that is,
\begin{equation}\label{eq:dnb}
    d_n^{(\beta)}= 
    \begin{cases}
    p(n), & \mbox{if }  V_{h(\beta)} \mbox{ is irreducible } \\
    \dim  \wt V^{[n]}_{h(\beta)},  &\mbox{otherwise.} 
    \end{cases}
\end{equation}
Thus, $d_n^{(\beta)}\leq p(n)$ and equality holds if and only if $V^{[n]}_{h(\beta)}$ does not contain
any singular vector~\cite{FeiginFuchs}.
Following \textit{op. cit.} we say that $V_{h(\beta)}$ is irreducible in degree $n$ if $d_n^{(\beta)}=p(n)$, otherwise we say that it is reducible in degree~$n$.
By the Kac determinant formula~\cite{FeiginFuchs}, $V_h$ is reducible in degree~$n$ if and only if $h=h_{r,s}$ for a pair of $r,s\in\mathbb N^*$ with $rs\leq n$, where
\be
h_{r,s}=\frac 18\bigl(r-2s+1\bigr)\bigl(r-2s-1\bigr).
\ee
For all $n\in\N$ we define the set
\be
A_n\,=\,\bigl\lbrace\beta\in\C\,:\, V_{h(\beta)}^{[n]}\mbox{ contains a singular vector}\bigr\rbrace.
\ee
\begin{proposition}\label{prop:An=Cn}
For all $n\in\N$ we have $A_n=C_n$, where $C_n$ is defined in~\eqref{eq:setCn}.
\end{proposition}
\begin{proof}
Using the facts on Virasoro Verma modules just recalled, $V^{[n]}_{h(\beta)}$ contains a singular vector if and only if $h(\beta)=h_{r,s}$, that is, if and only if $\beta=\pm(2r-s)$, for some $r,s\in\mathbb N^*$ satisfying $rs\leq n$.
Considering $s=1,2$ we see that $C_n\subseteq A_n$. On the other hand, if $s\geq 3$ and $rs=n$ we have $2r-s\leq \frac 23 n-3\leq n-2$ for all $n\geq 1$. This implies that $A_n\cap \Z_{\geq 0}\subseteq C_n\cap \Z_{\geq 0}$. Since $A_n=-A_n$ and~$C_n=-C_n$, $A_n\subseteq C_n$.
\end{proof}
By Proposition \ref{prop:An=Cn}, Corollary \ref{cor:countingextensions} and Theorem \ref{thm:completeness}, we deduce the following theorem.
\begin{theorem}\label{thm:irreducible?}
The number of BLZ polynomials (obtained via Crum--Darboux transformations) of level $n$ and momentum $\beta$ is less or than or equal to $p(n)$ and equality holds if and only if the Verma module $V_{h(\beta)}$ is irreducible in degree $n$.
\end{theorem}
We discuss the expected number of BLZ polynomials when $\beta \in \Z$, in Conjecture \ref{conjecture:new} below. This in fact requires the study of the Quantum KdV Hamiltonians.

\section{Diagonalization of Quantum KdV at the Free-Fermion Point}\label{sec:qKdV}

\subsection{Free field representation}
Let $\Bosonic=\C[q_1,q_2,\dots]$ be the polynomial ring in infinitely many variables, graded by $\deg q_n=n$.
Let
\be
\label{eq:chiralbosonicfield}
\varphi(z) = \sum_{n\geq 1}\biggl(\frac {q_n}nz^{n}-\frac{\partial}{\partial q_n}z^{-n}\biggr).
\ee
Moreover, let $D_z=z\frac{\partial}{\partial z}$ and
\be
u_l(z)=D_z^{l+1}\varphi(z)=\sum_{n\geq 1}\biggl(n^{l}q_nz^n+(-1)^ln^{l+1}\frac{\partial}{\partial q_n}z^{-n}\biggr)\qquad (l\geq 0).
\ee
To simplify the notation, we set $u(z)=u_0(z)$.

The normal ordering is defined (recursively) by
\be
\label{eq:normalorder}
:f_1(z_1)\cdots f_n(z_n): \,=\, f_1^+(z_1) : f_2(z_2)\cdots f_n(z_n) :\,+\,: f_2(z_2)\cdots f_n(z_n) :f_1^-(z_1)\,,
\ee
with initial condition $:f(z):\,=\,f(z)$, where for any $f(z)=\sum_{n\in \Z}f_nz^n$ we set
\be
f^+(z)=\sum_{n\geq 0}f_{n}z^{n},\qquad
f^-(z)=\sum_{n\leq -1}f_{n}z^{n}.
\ee
It must be noted that the normal ordering does not respect the associative property for Virasoro fields (it does respect it for free bosonic or fermionic fields).
As usual, we extend the normal ordering $ :\cdots: $ linearly.

Let $\Vir_{-2}$ be the Virasoro Lie algebra (with $c=-2$), namely, $\Vir_{-2}=\C Z\oplus\biggl(\bigoplus_{n\in\Z}\C L_n\biggr)$ and
\be
[L_m,L_n]=(m-n)L_{m+n}-\delta_{m+n,0}\frac{m^3-m}6Z,\quad [Z,L_m]=0\qquad (m,n\in\Z).
\ee
It is well-known (see, e.g.,~\cite{KacRaina}) that, for any $\beta\in\C$, the assignment
\be
\begin{aligned}
\label{eq:Virrep}
\rho_\beta(L_m)&=[z^{-m}]\biggl(\frac{\beta^2-1}8+\frac \beta 2 u(z)+\frac 12: u(z)^2:+ \frac 12  u_1(z)\biggr)\quad (m\in\Z),\\
\rho_\beta(Z)&=1
\end{aligned}
\ee
defines a representation $\rho_\beta:\Vir_{-2}\to\mathrm{End}\,\Bosonic$.

\begin{remark}
The representation $\rho_\beta$ is isomorphic to the Verma module with $h=\frac{\beta^2-1}8$ and $c=-2$ if and only if the latter is irreducible (which happens if and only if $\beta\not\in\Z$).
\end{remark}

Introduce, following~\cite{BLZ1} (see Equations~(1) and (6) in \textit{op. cit.}\footnote{We omit for simplicity the additive constant $\frac c{24}=-\frac 1{12}$ in the definition of $T$, which amounts to a simple triangular linear combination of the hamiltonian operators $I_{2k+1}$.})
\be
T(z) = \sum_{m\in\Z}L_{-m}z^{m}
\ee
as well as, for $k\in\mathbb N$,
\be
I_{2k+1}\,=\,\res{z=0}\,T^{[k]}(z)\,\frac{\d z}z
\ee
where $T^{[k]}(z)$ are defined recursively by
$T^{[0]}(z)=T(z)$ and $ 
T^{[k+1]}(z)=:T^{[k]}(z)\,T(z):$~\cite{DiFrancescoMathieuSenechal,KupershmidtMathieu,SasakiYamanaka}.
Due to the non-associativity of normal ordering, $T^{[k]}(z)\,\not=\,:T(z)^{k+1}:$ for $k\geq 2$ and $I_{2k+1}\not=\res{z=0}\,:T(z)^{k+1}:\,\frac{\d z}z$ for $k\geq 3$. We do have however $I_{2k+1}=\res{z=0}\,:T(z)^{k+1}:\,\frac{\d z}z$ for $k=0,1,2$, and we will mostly use these first three hamiltonians:
\begin{align}
I_1&=L_0,\\
I_3&=L_0^2+2\sum_{n\geq 1}L_{-n}L_n,\\
I_5&=3\sum_{m,n\geq 1}(L_{-m}L_{-n}L_{m+n}+L_{-m-n}L_{m}L_{n})+3\sum_{n\geq 1}(L_0L_{-n}L_n+L_{-n}L_0L_n)+L_0^3.
\end{align}
It should also be noted that, although each coefficient of the power series $T^{[k]}(z)$ is an infinite sum of monomials in the $L_m$ (thus, in principle, ill-defined), thanks to normal ordering each coefficient has a well-defined action on $\Bosonic$ (via $\rho_\beta$).

For the upcoming computations, we need to recall Wick's lemma.
Let, for $|w|>|z|$,
\be
\mathcal C(w,z) = \langle 1, u(w)u(z) 1\rangle_\Bosonic = \sum_{n\geq 1}n\biggl(\frac{z}{w}\biggr)^n=\frac{wz}{(w-z)^2}.
\ee
Here $\langle,\rangle_{\Bosonic}$ is the standard scalar product on $\Bosonic$, namely the one for which the Schur functions are orthonormal (see below, or~\cite[Chapter~1]{Macdonald}).

\begin{lemma}[Wick's lemma]
    Let $F(w)$ and $G(z)$ be polynomial expressions in $(u(w),u_1(w),u_2(w),\dots)$ and $(u(z),u_1(z),u_2(z),\dots)$ respectively.
    Then, if $|w|>|z|$,
    \be
    :F(w):\, :G(z): = : \exp\biggl(\sum_{i,j\geq 0} \bigl(D_w^iD_z^j\mathcal C(w,z)\bigr)\frac\partial{\partial u_i(w)}\frac\partial{\partial u_j(z)}\biggr) \bigl(F(w)G(z)\bigr)  :\,.
    \ee
\end{lemma}

\begin{proposition}
\label{prop:hamiltoniansdiffpolyu}
We have, denoting $u=u(z)$ and $u_l=u_l(z)$ ($l\geq 1$) for short,
\begin{equation}
\begin{aligned}
\rho_\beta(I_1)&=\res{z=0}:\biggl(\frac{\beta^2-1}8+\frac 12u^2\biggr):\frac{\d z}z\\
\rho_\beta(I_3)&=\res{z=0}:\Biggl(\biggl(\frac{\beta^2-1}8\biggr)^2+\frac{3\beta^2}4u^2+\frac{\beta}2u^3+\frac{2u^4+2uu_2-u^2}8\Biggr):\frac{\d z}z\\
\nonumber
\rho_\beta(I_5)&=\res{z=0}:\Biggl(\biggl(\frac{\beta^2-1}8\biggr)^3+\frac{15\beta^4}{128}u^2+\frac{5\beta^3}{16}u^3+\beta^2\frac{30u^4+30uu_2-9u^2}{32}\\
&\qquad
+\beta\biggl(\frac {6u^5+15u^2u_2-3u^3}{16}\biggr)
+
\frac{16 u^6-12 u^4-240  u^2u_1^2+3 u^2+12 u_1^2+16 u_2^2}{128}
\Biggr):\frac{\d z}z.
\end{aligned}
\end{equation}
\end{proposition}
\begin{proof}
    It is convenient to denote
    \be
    \rho_\beta(T(z))=\sum_{m\in\Z}\rho_\beta(L_m) z^{-m}=\frac{\beta^2-1}8+\frac \beta 2 u(z)+\frac 12: u(z)^2:+ \frac 12  u_1(z)\biggr.
    \ee
    The formula for $\rho_\beta(I_1)=\res{z=0}\rho_\beta(T(z))\frac{\d z}z$ is immediate, whereas those for $\rho_\beta(I_3)=\res{z=0}:\rho_\beta(T(z))^2:\frac{\d z}z$ and $\rho_\beta(I_5)=\res{z=0}:\rho_\beta(T(z))^3:\frac{\d z}z$ are obtained by straightforward (albeit lengthy) computations, which we omit here. These involve the use of Wick's lemma and of the fact that $\res{z=0}\bigl(D_zf(z)\bigr)\frac{\d z}z=0$.
\end{proof}

\subsection{Infinite wedge formalism and diagonalization}
Denote $\Z'=\Z+\frac 12$.
Let $\mathcal S$ be the set of strictly decreasing sequences $S=(s_i)_{i\geq 1}$  of half-integers ($s_i\in\Z'$) such that $s_i=s_{i-1}-1$ for $i$ large enough.
For each $S\in\mathcal S$ let $v_S$ be the formal expression
\be
v_S=\underline {s_1}\wedge\underline {s_2}\wedge\underline{s_3}\wedge\cdots
\ee
and let $\Fermionic=\bigoplus_{S\in\mathcal S} \C v_S$.
We agree that an expression $\underline {s_1}\wedge\underline {s_2}\wedge\underline{s_3}\wedge\cdots$ also defines an element of $\Fermionic$ even if the sequence of the $s_i$'s is not decreasing, but still satisfies  $s_i=s_{i-1}-1$ for $i$ large enough, hence it is eventually decreasing. This is done by requiring the identity
\be
\underline {s_1}\wedge\cdots\wedge \underline{s_i}\wedge \underline{s_{i+1}}\wedge\cdots =
-\underline {s_1}\wedge\cdots\wedge \underline{s_{i+1}}\wedge \underline {s_i}\wedge\cdots
\ee
to hold true.
We introduce linear operators $\psi_a,\psi_a^*$ ($a\in\Z'$) on $\Fermionic$ by
\be
\begin{aligned}
    \psi_a (\underline {s_1}\wedge\underline{s_2}\wedge\cdots) &= \underline a\wedge\underline {s_1}\wedge\underline{s_2}\wedge\cdots,
\\
\psi_a^* (\underline a\wedge\underline {s_1}\wedge\underline{s_2}\wedge\cdots)&=\underline {s_1}\wedge\underline{s_2}\wedge\cdots.
\end{aligned}
\ee
It is easily checked that these operators satisfy
\be
\label{eq:anticommfermion}
\psi_a\psi_b+\psi_b\psi_a=0=\psi_a^*\psi_b^*+\psi_b^*\psi_a^*,\quad
\psi_a\psi_b^*+\psi_b^*\psi_a^*=\delta_{ab}\qquad(a,b\in\Z').
\ee
Let $\mathcal S_0\subset\mathcal S$ be the set of $S\in\mathcal S$ such that $s_i=-i+\frac 12$ for $i$ large enough and let $\Fermionic_0$ be the subspace of $\Fermionic$ spanned by $v_S$ with $S\in \mathcal S_0$.
The mapping $\lambda\mapsto S_\lambda=(\lambda_i-i+\frac 12)_{i\geq 1}$ defines a bijection between the set of partitions and $\mathcal S_0$.
We use the notation $v_\lambda=v_{S_\lambda}$, such that $\Fermionic_0=\bigoplus_{\lambda}\C v_\lambda$, where the direct sum runs over all partitions $\lambda$.

Consider the isomorphism $\Phi:\Fermionic_0\to\Bosonic$ defined by $\Phi v_\lambda=s_\lambda$ for all $\lambda\in\Partitions$, where $s_\lambda=s_\lambda(q_1,q_2,\dots)$ are the Schur functions.
We recall that the latter can be defined by the Jacobi--Trudi formula
\be
s_\lambda = \det\bigl(h_{\lambda_i-i+j}\bigr)_{a,b=1}^{\ell(\lambda)}
\ee
where $h_i=h_i(q_1,q_2,\dots)$ are the complete homogeneous symmetric functions when $i\geq 0$ and $0$ when $i<0$, namely,
\be
h_i=[y^i]\exp\biggl(\sum_{j\geq 0}\frac {q_j}jy^j\biggr)\ \ \mbox{when}\ \ i\geq 0,\qquad
h_i=0\ \ \mbox{when}\ \ i<0.
\ee

The \emph{boson-fermion correspondence} (see, e.g.,~\cite{KacRaina,MiwaJimboDate}) states that, introducing
\be
\psi(z)=\sum_{a\in\Z'}\psi_az^{a-\frac 12},\qquad
\psi^*(z)=\sum_{a\in\Z'}\psi_a^*z^{-a-\frac 12},
\ee
and recalling~\eqref{eq:chiralbosonicfield}, we have
\be
\label{eq:bosfer}
\Phi:\psi(w)\psi ^*(z):_F\Phi^{-1}= \frac{:\e^{\varphi(w)-\varphi(z)}:-1}{w-z},
\ee
where the fermionic normal order is defined by
\be
:\psi(w)\psi^*(z):_F=\psi^+(w)\psi^*(z)-\psi^*(z)\psi^-(w),
\ee
similarly to~\eqref{eq:normalorder}.
Setting $w=z\e^\epsilon$ in~\eqref{eq:bosfer} and expanding as $\epsilon\to 0$ we obtain
\be
\label{eq:bosferuseful}
\Phi z:\bigl(\e^{\epsilon (D_z+\frac 12)}\psi(z)\bigr)\psi^*(z):_F\Phi^{-1}
=\frac{:\exp\bigl(\sum_{l\geq 1}\frac{\epsilon^l}{l!}u_{l-1}(z)\bigr):-1}{\e^{\epsilon/2}-\e^{-\epsilon/2}}.
\ee

Let us introduce operators $A_k^f:\Fermionic_0\to\Fermionic_0$ for $k\geq 0$ by $A_0^f=1$ and
\be
A_k^f=\res{z=0}:\biggl(\bigr(D_z+\frac 12\bigr)^{k-1}\psi(z)\biggr)\psi^*(z):_F\d z =\sum_{a\in \Z'}a^{k-1}:\psi_a\psi_a^*:_F
\ee
for $k\geq 1$.
Let us also set $A_k^b=\Phi A_k^f\Phi^{-1}:\Bosonic\to\Bosonic$.
From~\eqref{eq:bosferuseful} we get
\begin{align}
A_k^b=\res{z=0}[\epsilon^k]\frac{:\exp\bigl(\sum_{l\geq 1}\frac{\epsilon^l}{l!}u_{l-1}(z)\bigr):-1}{\e^{\epsilon/2}-\e^{-\epsilon/2}}\frac{\d z}z.
\end{align}
For example, $A_1^b=0$ and, omitting the dependence on $z$ in $u_l=u_l(z)$,
\begin{align}
\label{eq:A2}
A_2^b&=\res{z=0}:\frac {u^2}2:\frac{\d z}z,\\
\label{eq:A3}
A_3^b&=\res{z=0}:\frac{u^3}3:\frac{\d z}z,\\
\label{eq:A4}
A_4^b&=\res{z=0}:\frac{2u^4 -2u_1^2-u^2}{8}:\frac{\d z}z,\\
\label{eq:A5}
A_5^b&=\res{z=0}:\frac{6u^5+15u^2 u_2-5u^3}{30}:\frac{\d z}z,\\
\label{eq:A6}
A_6^b&=\res{z=0}:\frac{16 u^6-20 u^4-240 u^2u_1^2 +7 u^2+20 u_1^2+16 u_2^2}{96}:\frac{\d z}z.
\end{align}
In these examples we used the fact that $\res{z=0} \bigl(D_z f(z)\bigr) \frac{\d z}z =0$.

\begin{remark}\label{rem:comparisonhamiltonianstructures}
It is known~\cite{BuryakRossi,Dubrovin,Eliashberg,Pogrebkov,Rossi} that the operators $A_k^b$ are (possibly up to a triangular linear combination) the hamiltonian operators of the quantization of the dispersionless periodic KdV model with respect to the first hamiltonian structure.
It follows from Proposition~\ref{prop:hamiltoniansdiffpolyu} and the above formulas that the KdV hamiltonian operators $I_1,I_3,I_5$ obtained by quantizing the second hamiltonian structure at $c=-2$ also belong to the algebra generated by the (first six) commuting KdV hamiltonian operators obtained by quantizing the first hamiltonian structure and taking the dispersionless limit.
\end{remark}

\begin{lemma}
\label{lemma:Ab}
    For all $k\in\N$, we have $ A_k^b  s_\lambda= \mathsf{p}_{k}(\lambda) s_\lambda$.
\end{lemma}
\begin{proof}
    The result is well known (cf.~\cite{Dubrovin}) but we report the proof for completeness.
    It is easy to see that
    \be
    :\psi_a\psi_a^*:_F v_\lambda = :\psi_a\psi_a^*:_F \bigl(\underline{\lambda_1-\tfrac 12}\wedge\underline{\lambda_2-\tfrac 32}\wedge\cdots\bigr)
    =\biggl(\sum_{c\in C_\lambda^+}\delta_{ac} - \sum_{c\in C_\lambda^-}\delta_{ac}\biggr)v_\lambda
    \ee
    where $C_\lambda^+=\lbrace c\in \Z':\, c>0\mbox{ and }c=\lambda_i-i+\frac 12\mbox{ for some }i\geq 1\rbrace$ and $C_\lambda^-=\lbrace c\in \Z':\, c<0\mbox{ and }c\not=\lambda_i-i+\frac 12\mbox{ for all }i\geq 1\rbrace$. 
    Since $\mathsf{p}_k(\lambda) = \sum_{c\in C^+_\lambda}c^{k-1}-\sum_{c\in C^-_\lambda}c^{k-1}$, it follows that $A_k^f v_\lambda= \mathsf{p}_k(\lambda)v_\lambda$, which is equivalent to the statement to be proven.
\end{proof}

\begin{theorem}
\label{thm:diagonalonSchur}
    For all $\lambda\in\Partitions$ and $k=1,2,3$ we have
    \be
    \rho_\beta(I_{2k-1})\,s_\lambda\,=\,\mathcal I_{2k-1}(\lambda)\,s_\lambda
    \ee
    where, denoting $\mathsf{p}_k=\mathsf{p}_k(\lambda)$ and $\mathcal I_{2k-1}=\mathcal I_{2k-1}(\lambda)$, 
    \begin{align}
    \nonumber
        \mathcal I_1&=\frac{\beta^2-1}8+\mathsf{p}_2\,,\qquad\qquad
        \mathcal I_3=\biggl(\frac{\beta^2-1}{8}\biggr)^2+\frac 34\beta^2 \mathsf{p}_2+\frac 32\beta \mathsf{p}_3+\mathsf{p}_4\,,
        \\
        \mathcal I_5&=\biggl(\frac{\beta^2-1}{8}\biggr)^3+\frac{15\beta^4}{64}\mathsf{p}_2+\frac{15\beta^3}{16}\mathsf{p}_3
        +\beta^2\biggl(\frac {15}8\mathsf{p}_4+\frac 3{16}\mathsf{p}_2\biggr)+\beta\biggl(\frac {15}8\mathsf{p}_5+\frac 38\mathsf{p}_3\biggr)+\frac34 \mathsf{p}_6+\frac {\mathsf{p}_4}4 \,.
    \end{align}
\end{theorem}
\begin{proof}
It follows from Proposition~\ref{prop:hamiltoniansdiffpolyu}, Lemma~\ref{lemma:Ab}, and equations~\eqref{eq:A2}--\eqref{eq:A6}.
\end{proof}

\begin{corollary}\label{cor:sumzi}
    For all $\lambda\in\Partitions$ we have
\begin{equation}\label{eq:sumzicor}
    \sum_{i=1}^{|\lambda|} z_i^{(\beta)}(\la) \,=\,2\,\mathcal I_3(\la)\,-\,2\,\mathcal I_1(\la)^2 
\end{equation}
where $z_i^{(\beta)}(\lambda)$ (for $i=1,\dots,|\lambda|$) are the zeros of $\mathcal P_{\lambda}^{(\beta)}(z)$, see Definition~\ref{def:Pla}.
\end{corollary}
\begin{proof}
    It follows from the explicit formulas for the eigenvalues $\mathcal I_1(\lambda),\mathcal I_3(\lambda)$ and Proposition~\ref{prop:sumzi}.
\end{proof}

\begin{remark}\label{rem:sumzi}
Conditional on the validity of the conjectural formula~\eqref{eq:sumzisquared}, we also have
\be
\label{eq:sumzisquare}
\sum_{i=1}^{|\lambda|} \bigl(z_i^{(\beta)}(\la)\bigr)^2 \,=\,\frac {56}{3}\,\mathcal I_5(\la)\,-\,40\,\mathcal I_3(\lambda)\,\mathcal I_1(\lambda)\,+\,\frac {64}{3}\,\mathcal I_1(\lambda)^3\,+\,12\,\mathcal I_3(\lambda)\,-\,12\,\mathcal I_1(\la)^2\,.
\ee
Formulas (\ref{eq:sumzicor},\ref{eq:sumzisquare}) suggest the possibility that the algebraic system~\eqref{eq:blzsystem} for the roots $z_i$ is in fact a kind of Bethe ansatz equations relative to an alternative definition of the quantum KdV hierarchy --- see \cite{Litvinov}.
\end{remark}

\subsection{Non-semisimplicity of quantum KdV hamiltonians on Verma modules and number of distinct BLZ polynomials when $\beta$ is an integer.}\label{sec:newconjecture}

As mentioned above, the representation $\rho_\beta$ is isomorphic to the Verma module $V_{h(\beta)}$ if and only if $\beta\not\in\Z$~\cite{FeiginFuchs}.
Theorem~\ref{thm:diagonalonSchur} shows that $I_3,I_5$ are always semisimple in the representation $\rho_\beta$.
On the other hand, when $\beta\in\Z$ and we consider the quantum KdV hamiltonians $I_{2k+1}$ acting on the Virasoro Verma module $V_{h(\beta)}$, the situation is more delicate.

Let us recall~\cite{FeiginFuchs} that the Shapovalov form $\langle,\rangle_h$ is a symmetric bilinear form on $V_{h}$, uniquely defined by the property
\be
\langle v,L_n w\rangle_h \,=\,\langle L_{-n}v,w\rangle_{h}\quad\text{for all }v,w\in V_h,\ n\in\mathbb Z,
\ee
and by the normalization $\langle v_h,v_h\rangle _h=1$ for the highest weight vector $v_h\in V_h$.
The Shapovalov form is degenerate if and only if $V_h$ is reducible.
Noting that the quantum KdV hamiltonians satisfy
\be
\langle I_{2k+1} v,w\rangle_{h(\beta)}\,=\,\langle v,I_{2k+1} w\rangle_{h(\beta)}\qquad (k\in\mathbb N)\,,
\ee
we infer that $I_{2k+1}$ may be non-semisimple only if $V_{h(\beta)}$ is reducible (we recall that this happens only if $\beta\in\mathbb Z$).
And, indeed, this happens (for $k\in\mathbb N^*$), as the following examples show.

\begin{example}~
    \begin{itemize}[leftmargin=*]
    \item When $\beta=0$ (and so $h(0)=-\frac 18$), the matrices representing $I_3,I_5,I_7$ with respect to the basis $\bigl(e_1=(L_{-1}^2-\frac 12L_{-2})v_{h(0)},\ e_2=\frac 16 L_{-2}v_{h(0)}\bigr)$ of $V_{h(0)}^{[2]}$ ($v_h$ being the highest weight vector in $V_h$) are, respectively,
    \be
    \renewcommand{\arraystretch}{1.2}
    \begin{pmatrix}
 \frac{225}{64} & 1 \\
 0 & \frac{225}{64}
        \end{pmatrix},\quad
        \begin{pmatrix}
 \frac{3375}{512} & \frac{27}{8} \\
 0 & \frac{3375}{512}
        \end{pmatrix},\quad
        \begin{pmatrix}
 \frac{50625}{4096} & \frac{279}{32} \\
 0 & \frac{50625}{4096}
        \end{pmatrix}.
        \ee
        The only common eigenvector is $e_1$, which is a null vector, i.e., $\langle e_1,v\rangle_{-\frac 18}=0$ for all $v\in V_{-\frac 18}$.
        We conclude that there are no common eigenvectors of the hamiltonians $I_{2k+1}$ which are nonzero after projection onto the irreducible quotient of the Virasoro Verma module.
        On the other hand, since
        \be
        \mathcal P_{(2)}^{(\beta=0)}(z)\,=\,
        \mathcal P_{(1,1)}^{(\beta=0)}(z)\,=\,z^2
        \ee
        we also see that there are no BLZ polynomials (obtained via Crum–Darboux transformations, which however, by Conjecture~\ref{conj:oblo}, should be all of them) of degree $n=2$ and momentum $\beta=0$ .
\item This situation is confirmed in the same case $\beta=0$ on $V_{h(0)}^{[3]}$.
The matrices representing $I_3,I_5,I_7$ with respect to given by
        \be
        e_1=\biggl(L_{-1}^3-2L_{-2}L_{-1}+\frac 14 L_{-3}\biggr)v_h,\
        e_2=\biggl(L_{-1}^3-\frac 12 L_{-2}L_{-1}-\frac 12 L_{-3}\biggr)v_h,\
        e_3=\frac 16 L_{-2}L_{-1}v_h,
        \ee
are, respectively,
        \be
        \renewcommand{\arraystretch}{1.2}
        \begin{pmatrix}
 \frac{433}{64} & 0 & 0 \\
 0 & \frac{1009}{64} & 1 \\
 0 & 0 & \frac{1009}{64}
\end{pmatrix},\quad
\begin{pmatrix}
 \frac{6695}{512} & 0 & 0 \\
 0 & \frac{39527}{512} & \frac{67}{8} \\
 0 & 0 & \frac{39527}{512}
\end{pmatrix},\quad
\begin{pmatrix}
 \frac{101089}{4096} & 0 & 0 \\
 0 & \frac{1450081}{4096} & \frac{5077}{96} \\
 0 & 0 & \frac{1450081}{4096}
\end{pmatrix}.
        \ee
Thus, $e_1,e_2$ are the only common eigenvectors. However, $e_2$ is a null vector and so there is only one eigenvector which is nonzero after projection onto the irreducible Virasoro Verma module.
On the other hand, since
\be
\mathcal P_{(3)}^{(\beta=0)}(z)\,=\,\mathcal P_{(1,1,1)}^{(\beta=0)}(z)\,=\,z^2(z-15),\quad 
\mathcal P^{(\beta=0)}_{(2,1)}(z)\,=\,z^3+3z^2+27z+9\,,
\ee
we see that there exists exactly one BLZ polynomial  (obtained via Crum–Darboux transformations) of degree $n=3$ and momentum $\beta=0$.
\item When $\beta=2$, such that $h(\beta)=\frac 38$, the matrices representing $I_3$ and $I_5$ on $V_{\frac 38}^{[4]}$ with respect to the basis $(e_1,e_2,e_3,e_4,e_5)$ given by
\be
\renewcommand{\arraystretch}{1.2}
\begin{pmatrix}
    e_1 \\ e_2 \\ e_3 \\ e_4 \\ e_5
\end{pmatrix}
\,=\,
\begin{pmatrix}
 -\frac{3}{2} & -2 & 0 & 4 & -\frac{4}{5} \\
 -360 & 240 & 540 & -1200 & 240 \\
 45 & -100 & 0 & 0 & 8 \\
 -\frac{1}{2} & -3 & -\frac{7}{4} & -1 & 1 \\
 \frac{15}{2} & 13 & \frac{9}{4} & 7 & 1
\end{pmatrix}
\begin{pmatrix}
    L_{-4}v_h \\ L_{-3}L_{-1}v_h \\ L_{-2}^2v_h \\ L_{-2}L_{-1}^2v_h \\ L_{-1}^4v_h
\end{pmatrix}
\ee
are, respectively,
\be
\renewcommand{\arraystretch}{1.2}
\begin{pmatrix}
 \frac{1225}{64} & 0 & 0 & 0 & 0 \\
 0 & \frac{1225}{64} & 1 & 0 & 0 \\
 0 & 0 & \frac{1225}{64} & 0 & 0 \\
 0 & 0 & 0 & \frac{2761}{64} & 0 \\
 0 & 0 & 0 & 0 & \frac{5833}{64}
\end{pmatrix}\,,\quad
\begin{pmatrix}
 \frac{42875}{512} & 0 & 0 & 0 & 0 \\
 -\frac{3}{40} & \frac{42875}{512} & \frac{73}{8} & 0 & 0 \\
 0 & 0 & \frac{42875}{512} & 0 & 0 \\
 0 & 0 & 0 & \frac{207227}{512} & 0 \\
 0 & 0 & 0 & 0 & \frac{720251}{512} 
\end{pmatrix}\,.
\ee
This shows that $I_3$ and $I_5$ do not have the same flag of generalized eigenspaces.
In particular, the only common eigenvectors to $I_3$ and $I_5$ are $e_2,e_4,e_5$, and $e_2$ is a null vector, while $e_4,e_5$ remain linearly independent after projection onto the irreducible quotient of $V_{\frac 38}^{[4]}$.
On the other hand, since
\be
\begin{aligned}
    &\mathcal P_{(4)}^{(\beta=2)}(z)=z^4-144 z^3+3888 z^2-17280 z+5184\,,\quad
    \mathcal P_{(1,1,1,1)}^{(\beta=2)}(z)=z^4\,,\\
    &\mathcal P_{(2,2)}^{(\beta=2)}(z)=z^4+16 z^3-24 z^2+576 z+144\,,\quad
    \mathcal P_{(3,1)}^{(\beta=2)}(z)=\mathcal P_{(2,1,1)}^{(\beta=2)}(z)=z^2(z^2-8 z+80)\,,
\end{aligned}
\ee
we also see that there are exactly two distinct BLZ polynomials (obtained via Crum–Darboux transformations) of degree $n=4$ and momentum $\beta=2$.
\end{itemize}
\end{example}

Motivated by these examples, we formulate the following conjecture.
Let $\pi:V_h\to \wt V_h$ be the projection onto the irreducible quotient, as in~\eqref{eq:irreduciblequotient}, and let $b_n^{(\beta)}$ be the number of distinct BLZ polynomials of degree $n$ and momentum $\beta$.

\begin{conjecture}
\label{conjecture:new}
For $k,n\in\mathbb N$ and $\beta\in\mathbb C$, let $E_{2k+1,n}^{(\beta)}$ be the span of all eigenvectors of $I_{2k+1}$ on $V_{h(\beta)}^{[n]}$.
For all $n\in\mathbb N$ and $\beta\in\C$, we have
\be
    b_n^{(\beta)}\,=\,\dim \pi \left(\,\bigcap_{k\in\mathbb N}E_{2k+1,n}^{(\beta)}\,\right)\,.
\ee
\end{conjecture}

In particular, if true, this conjecture would imply $b_n^{(\beta)}\leq d^{(\beta)}_n$, the latter being the dimension of $\wt V_{h(\beta)}$ as per~\eqref{eq:dnb}.
Hence Theorem~\ref{thm:irreducible?} for integer $\beta$ is a weaker form of this conjecture, as it states that $b_n^{(\beta)}<p(n)$ if and only if $V_{h(\beta)}$ is reducible in degree $n$.

\begin{remark}
Our paper supports a natural interpretation of the operators $\mathscr{H}^{(\beta)}_{\la,\la'}$ as states in the Virasoro Verma module $V_{h(\beta)}$: the harmonic oscillator $\mathscr{H}^{(\beta)}$ corresponds to the highest weight vector with weight $h(\beta)$ and the operator $\mathscr{H}^{(\beta)}_{\la,\la'}$ corresponds to a descendant of degree $|\la|$. 
In particular, when $\beta \notin \Z$, the quantum KdV hamiltonians are diagonalizable and the operators $\mathscr{H}^{(\beta)}_{\la,\la'}$ correspond to a basis of common eigenvectors (in other words, a basis of Bethe states).

When $\beta \in \Z$, some polynomials $\mathcal{P}^{(\beta)}_{\la}(x)$ vanish at $x=0$  and the module $V_{h(\beta)}$ is reducible. Recall that the polynomial $\mathcal{P}^{(\beta)}_{\la}$ vanishes at $0$ if and only if the triple $(\beta,\la,\la')$ is not reduced; in this case, the operator $\mathscr{H}^{(\beta)}_{\la,\la'}$ is not a rational extension of the harmonic oscillator with momentum $\beta$, but it is a rational extension of the harmonic oscillator with momentum $\wt \beta$ where $\wt \beta$ is the momentum of the reduced triple equivalent to $(\beta,\la,\la')$, see Theorem \ref{thm:lamuextension}. Therefore, if the triple is not reduced, $\mathscr{H}^{(\beta)}_{\la,\la'}$ corresponds to a vector in $V_{h(\beta)}$ which is either a singular vector of highest weight $h(\wt \beta)$ or a descendant of such a vector. On the contrary, if the triple is reduced, $\mathscr{H}^{(\beta)}_{\la,\la'}$ corresponds to a vector in $V_{h(\beta)}$ which is neither singular nor a descendent of a singular vector. 
\end{remark}

\appendix 

\section{Properties of Laguerre Wronskians}\label{app:propertiesLaguerre}

\subsection{Proof of Theorem~\ref{thm:propertiesWronskianLaguerre}}\label{sec:proofstructure}
\begin{proof}[Proof of Theorem~\ref{thm:propertiesWronskianLaguerre}]
By standard properties of the Wronskian,
\be
\Wr x\bigl(g(x)f_1(x^2),\dots,g(x)f_k(x^2)\bigr)\,=\,(2x)^{\binom k2}\,g(x)^k\,\Wr y (f_1(y),\dots,f_k(y))\Bigr|_{y=x^2}.
\ee
Hence,
\be
\label{fact0}
\Psi^{(\beta)}_{\bs m,\bs n}(x)\,=\,\e^{-\frac12 (r+s)x^2}\,x^{\frac 12(r+s)}\, (2x)^{\binom{r+s}2}\,\wt\Psi(x^2)
\ee
where, omitting the explicit dependence on $\beta,\bs m,\bs n$,
\be
\label{eq:wtPsi}
\wt\Psi(y)\,=\,\Wr y \bigl(y^{\frac \beta 2}\wh{L}_{m_1}^{(\beta)}(y),\dots ,y^{\frac \beta 2} \wh{L}_{m_r}^{(\beta)}(y),y^{-\frac \beta 2}\wh{L}_{n_1}^{(-\beta)}(y),\dots,y^{-\frac \beta 2}\wh{L}_{n_s}^{(-\beta)}(y)\bigr).
\ee
Introduce the $(r+s)\times (m_1+n_1+2)$ matrix $Y=(Y_1|Y_2)$, where $Y_1$ is an $(r+s)\times (m_1+1)$ matrix and $Y_2$ is an $(r+s)\times(n_1+1)$ matrix given by
\be
\begin{aligned}
(Y_1)_{ij}&=\partial_y^{i-1}y^{\frac \beta 2+j-1},& & 1\leq i\leq r+s,\ 1\leq j\leq m_1+1,
\\
(Y_2)_{ij}&=\partial_y^{i-1}y^{-\frac\beta2+j-1},& & 1\leq i\leq r+s,\ 1\leq j\leq n_1+1.
\end{aligned}
\ee
Let us also introduce the $(m_1+n_1+2)\times (r+s)$ matrix $C=A_{\bs m}\oplus B_{\bs n}$ where $A_{\bs m}$ is an $(m_1+1)\times r$ matrix and $B_{\bs n}$ is an $(n_1+1)\times s$ matrix given by
\be
\begin{aligned}
(A_{\bs m})_{ij}&\,=\,c_{m_j}^{[i-1]}(\beta),&& 1\leq i\leq m_1+1 ,\ 1\leq j\leq r,\\
(B_{\bs n})_{ij}&\,=\,c_{n_j}^{[i-1]}(-\beta),&& 1\leq i\leq n_1+1 ,\ 1\leq j\leq s,\\
\end{aligned}
\ee
where $c_{n}^{[i]}(\beta)$ is given in~\eqref{eq:LaguerreExplicit}.
More explicitly,
\be
\begin{aligned}
A_{\bs m}&\,=\,\left(\begin{smallmatrix}
(-1)^{m_1}(1+\beta)_{m_1} & (-1)^{m_2}(1+\beta)_{m_2}& \cdots & (-1)^{m_r}(1+\beta)_{m_r}\\
\vdots &\vdots &  &\vdots \\
\vdots &\vdots & \iddots &-m_r(m_r+\beta) \\
\vdots &-m_2(m_2+\beta) & \iddots &1 \\
\vdots &1 & \iddots &0 \\
\vdots &0 & \iddots &\vdots \\
-m_1(m_1+\beta) &\vdots & & \\
1 &0 & \cdots &0
\end{smallmatrix}\right),
\\
\label{eq:AB}
B_{\bs n}&\,=\,\left(\begin{smallmatrix}
(-1)^{n_1}(1-\beta)_{n_1} & (-1)^{n_2}(1-\beta)_{n_2}& \cdots & (-1)^{n_s}(1-\beta)_{n_s}\\
\vdots &\vdots &  &\vdots \\
\vdots &\vdots & \iddots &-n_s(n_s-\beta) \\
\vdots &-n_2(n_2-\beta) & \iddots &1 \\
\vdots &1 & \iddots &0 \\
\vdots &0 & \iddots &\vdots \\
-n_1(n_1-\beta) &\vdots & & \\
1 &0 & \cdots &0
\end{smallmatrix}\right).
\end{aligned}
\ee
By the Binet--Cauchy identity, we have $\wt\Psi(y)=\det\bigl(YC\bigr)=\sum_{I} |Y_I|\,|C_I|$, where the sum runs over subsets $I$ of $[m_1+n_1+2]$ of cardinality $r+s$ and where $|Y_I|$ is the minor of $Y$ corresponding to the columns indexed by $I$ and $|C_I|$ is the minor of $C$ corresponding to the rows indexed by $I$.
Thanks to the structure of the matrices $Y$ and $C$, we get
\be
\wt \Psi(y)=\sum_{\substack{\bs t\in\mathscr V_r,\,\bs t\leq\bs m \\ \bs u\in\mathscr V_s,\,\bs u\leq\bs n}}
\Wr y(y^{\frac \beta 2+t_r},\dots,y^{\frac \beta 2+t_1},y^{-\frac \beta 2+u_s},\dots,y^{-\frac \beta 2+u_1})\,|A_{\bs m,\bs t}|\,|B_{\bs n,\bs u}|
\ee
where we define the relation $\bs t\leq \bs m$ by $t_i\leq m_i$ for all $i=1,\dots,r$ and $|A_{\bs m,\bs t}|$ is the minor of $A_{\bs m}$ corresponding to the rows indexed by $\lbrace t_r+1,\dots,t_1+1\rbrace$, and similarly for $\bs u\leq \bs n$ and $|B_{\bs n,\bs u}|$.
The reason for the restrictions $\bs t\leq\bs m$ and $\bs u\leq \bs n$ is that otherwise the minors of $A$ and $B$ vanish, see~\eqref{eq:AB}.

Since for any vector $\bs v=(v_1,\dots,v_k)$ we have
\be
\Wr y (y^{v_1},\dots,y^{v_k})\,=\,\Delta(\bs v)\,y^{\bigl(\sum_{i=1}^kv_i\bigr)-\binom k2},
\ee
we can rewrite the previous relation as
\be
\label{eq:Psiexpression}
\wt \Psi(y)\,=\,(-1)^{\binom r2+\binom s2}\,y^{\frac \beta 2(r-s)-\binom{r+s}2}\,\sum_{\substack{\bs t\in\mathscr V_r,\,\bs t\leq\bs m \\ \bs u\in\mathscr V_s,\,\bs u\leq\bs n}}
y^{|\bs t|+|\bs u|}\,\Delta\biggl(\bs t+\frac \beta 2,\bs u-\frac\beta 2\biggr)\,|A_{\bs m,\bs t}|\,|B_{\bs n,\bs u}|.
\ee
It follows from~\eqref{eq:Psiexpression} that $\wt \Psi(y) y^{-\frac\beta 2(r-s)+\binom{r+s}2}$ is a polynomial in $y$ whose lowest degree term comes from $\bs t=(r-1,\dots,1,0)$ and $\bs u=(s-1,\dots,1,0)$ only and whose highest degree term comes from $\bs t=\bs m$ and $\bs u=\bs n$ only.
Therefore, $\widetilde{\Phi}_{\bs m,\bs n}^{(\beta)}(y)$, defined by~\eqref{eq:PsiPhi}, is indeed a polynomial in $y$ of degree $d$ as in~\eqref{eq:degree}.
In particular, with reference to~\eqref{eq:phik} and recalling the factor $2^{\binom{r+s}2}$ in~\eqref{fact0},
\be
\label{eq:phiexpression}
\Delta\biggl(\bs m+\frac\beta 2,\bs n-\frac\beta 2\biggr)\,\widetilde{\phi}^{[k]}_{\bs m,\bs n}(\beta)
\,=\,
(-1)^{\binom r2+\binom s2}\,\sum_{\bs t,\,\bs u}
\Delta\biggl(\bs t+\frac \beta 2,\bs u-\frac \beta 2\biggr)\,|A_{\bs m,\bs t}|\,|B_{\bs n,\bs u}|,
\ee
where the sum on the right-hand side ranges over $\bs t\in\mathscr V_r$ and $\bs u\in\mathscr V_s$ such that $\bs t\leq\bs m$, $\bs u\leq \bs n$, and $|\bs t|+|\bs u|=k+\binom r2+\binom s2$.
Based on this expression, let us show that each coefficient $\widetilde{\phi}^{[k]}_{\bs m,\bs n}(\beta)$ is in $\Z[\beta]$.
(Once this is proved, it also follows from~\eqref{eq:phiexpression} that the degree of $\widetilde{\phi}^{[k]}_{\bs m,\bs n}(\beta)$ in $\beta$ is $\leq d-k$.)
To prove this claim, we need to show that the right-hand side in~\eqref{eq:phiexpression} is divisibile by $\Delta(\bs m)\Delta(\bs n)\prod_{i=1}^r\prod_{j=1}^s(n_j-m_i-\beta)$.
We first prove divisibility by the last factor, namely, we show that any singularity in $\beta$ is removable.
Indeed, consider an $\beta_0\in\mathbb{Z}$ which is a zero of $\prod_{i=1}^r\prod_{j=1}^s(n_j-m_i-\beta)$ of order
\be
\mu\, =\, \#\lbrace (m_i,n_j):\,i\in[r],\, j\in[s],\,n_j-m_i=\beta_0\rbrace\,.
\ee
On the other hand, $\wt\Psi(y)$ is a determinant of a matrix depending polynomially on $\beta$ which has corank at least $\mu$ when $\beta=\beta_0$ by point~(3) in Remark~\ref{remark:LaguerrePolynomials}.
Hence, $\wt\Psi(y)$ has a zero of order at least $\mu$ at $\beta_0$, and $\beta=\beta_0$ is a removable singularity.
Divisibility by $\Delta(\bs m)$ and $\Delta(\bs n)$ follows from the fact that $(-1)^mm!L_m^{(\beta)}(y)$ is a meromorphic function of~$m$.

Moreover, by~\eqref{eq:AB}, for all $\bs m,\bs t\in\mathscr V_r$ with $\bs t\leq \bs m$, we have
\be
\label{eq:Aminor}
\begin{aligned}
|A_{\bs m,\bs t}|&=\det_{1\leq i,j\leq r} c^{[t_{r-i+1}]}_{m_j} (\beta)
\\
&=\det_{1\leq i,j\leq r} (-1)^{t_{r-i+1}+m_j}\binom{m_j}{t_{r-i+1}}(\beta+t_{r-i+1}+1)_{m_j-t_{r-i+1}}
\\
&=\det_{1\leq i,j\leq r} (-1)^{m_j+t_{r-i+1}}\frac{(\beta+1)_{m_j}}{(\beta+1)_{t_{r-i+1}}}\binom{m_j}{t_{r-i+1}}
\\
&=(-1)^{|\bs m|+|\bs t|}\prod_{i=1}^r\frac{(\beta+1)_{m_i}}{(\beta+1)_{t_i}}\det_{1\leq i,j\leq r}\binom{m_j}{t_{r-i+1}}
\\
&=(-1)^{|\bs m|+|\bs t|}\prod_{i=1}^r(\beta+t_i+1)_{m_i-t_i}\det_{1\leq i,j\leq r}\binom{m_j}{t_{r-i+1}}.
\end{aligned}
\ee
Similarly, for all $\bs n,\bs u\in\mathscr V_s$ with $\bs u\leq \bs n$, we have
\be
\label{eq:Bminor}
|B_{\bs n,\bs u}|=(-1)^{|\bs n|+|\bs u|}\prod_{i=1}^s(-\beta+u_i+1)_{n_i-u_i}\,\det_{1\leq i,j\leq s}\binom{n_j}{u_{s-j+1}}.
\ee
In particular, $|A_{\bs m,\bs m}|=(-1)^{\binom r2}$ and $|B_{\bs n,\bs n}|=(-1)^{\binom s2}$, hence $\widetilde{\phi}^{[d]}_{\bs m,\bs n}(\beta)=1$.
On the other hand, as we already explained, the constant term in $\widetilde{\Phi}^{(\beta)}_{\bs m,\bs n}(y)$ comes from~$\bs t=(r-1,\dots,1,0)$ and $\bs u=(s-1,\dots,1,0)$ only, hence the claimed formula for $\widetilde{\phi}^{[0]}_{\bs m,\bs n}(\beta)$ follows from~\eqref{eq:phiexpression}, \eqref{eq:Aminor}, and~\eqref{eq:Bminor}, thanks to the Gessel--Viennot \textit{binomial determinant identity}~\cite{GV}, which asserts that for any $\bs v=(v_1,\dots,v_k)\in\mathscr{V}_r$ we have
\be
\det_{1\leq i,j\leq k}\binom{v_j}{i-1}\,=\,(-1)^{\binom k2}\,\frac{\Delta(\bs v)}{\Delta\bigl((k-1,k-2,\dots,1,0)\bigr)}\,.
\ee
It remains to prove the formulas for $\widetilde{\phi}^{[d-1]}_{\bs m,\bs n}(\beta)$ and $\widetilde{\phi}^{[d-2]}_{\bs m,\bs n}(\beta)$.
For the former, by the general expression~\eqref{eq:phiexpression} we see that the contribution to this coefficient comes exclusively from vectors $\bs t,\bs u$ of the form $\bs t=\bs m'_i=(m_1,\dots,m_i-1,\dots,m_r),\bs u=\bs n$ or of the form $\bs t=\bs m$, $\bs u=\bs n'_j=(n_1,\ldots,n_j-1,\dots,n_s)$.
This implies that $\widetilde{\phi}^{[d-1]}_{\bs m,\bs n}(\beta)$ is equal to
\be
\begin{aligned}
&(-1)^{\binom{r}{2}}\sum_{i=1}^r \frac{\Delta(\bs m'_i+\frac \beta 2,\bs n-\frac \beta 2)}{\Delta(\bs m+\frac \beta 2,\bs n-\frac \beta 2)}|A_{\bs m,\bs m'_i}|
+(-1)^{\binom{s}{2}}\sum_{j=1}^s\frac{\Delta(\bs m+\frac \beta 2,\bs n'_j-\frac \beta 2)}{\Delta(\bs m+\frac \beta 2,\bs n-\frac \beta 2)}|B_{\bs n,\bs n'_j}|
\\
&=
-\sum_{i=1}^r \frac{\Delta(\bs m'_i+\frac \beta 2,\bs n-\frac \beta 2)}{\Delta(\bs m+\frac \beta 2,\bs n-\frac \beta 2)}m_i(m_i+\beta)-\sum_{j=1}^s\frac{\Delta(\bs m+\frac \beta 2,\bs n'_j-\frac \beta 2)}{\Delta(\bs m+\frac \beta 2,\bs n-\frac \beta 2)}n_j(n_j-\beta).
\end{aligned}
\ee
Some $\bs m'_i$ might not be in $\mathscr{V}_r$ and that they do not contribute in this sum due to the Vandermonde factors; similarly for the $\bs n'_j$.
We have used the identities
\be
\label{eq:subleadingminors}
(-1)^{\binom r2}|A_{\bs m,\bs m'_i}|=-m_i(m_i+\beta),\quad
(-1)^{\binom s2}|B_{\bs n,\bs n'_j}|=-n_j(n_j-\beta),
\ee
whenever $\bs m'_i\in\mathscr V_r$ and $\bs n'_j\in\mathscr V_s$ (respectively), which are obtained directly by~\eqref{eq:AB}.
Rewriting the previous expression, we obtain that $\widetilde{\phi}^{[d-1]}_{\bs m,\bs n}(\beta)$ equals
\be
\begin{aligned}
&-\sum_{i=1}^r \frac{\Delta(\bs m'_i+\frac \beta 2,\bs n-\frac \beta 2)}{\Delta(\bs m+\frac \beta 2,\bs n-\frac \beta 2)}\bigl(m_i+\frac\beta2\bigr)^2-\sum_{j=1}^s\frac{\Delta(\bs m+\frac \beta 2,\bs n'_j-\frac \beta 2)}{\Delta(\bs m+\frac \beta 2,\bs n-\frac \beta 2)}\bigl(n_j-\frac\beta2\bigr)^2
\\
&\quad +\frac{\beta^2}{4}\biggl(\sum_{i=1}^r\frac{\Delta(\bs m'_i+\frac \beta 2,\bs n-\frac \beta 2)}{\Delta(\bs m+\frac \beta 2,\bs n-\frac \beta 2)}+\sum_{j=1}^s\frac{\Delta(\bs m+\frac \beta 2,\bs n'_j-\frac \beta 2)}{\Delta(\bs m+\frac \beta 2,\bs n-\frac \beta 2)}\biggr).
\end{aligned}
\ee
It suffices to apply the two identities in~\eqref{eq:lemmaVander0} with $\bs x=\bigl(\bs m+\frac \beta 2,\bs n-\frac\beta 2\bigr)$ to obtain (after straightforward algebra) the claimed formula for $\widetilde{\phi}_{\bs m,\bs n}^{[d-1]}(\beta)$ in the statement.

The proof of the formula for $\widetilde{\phi}_{\bs m,\bs n}^{[d-2]}(\beta)$ is similar, although a little bit more involved.
First, in the general expression~\eqref{eq:phiexpression} we have contributions from vectors $\bs t$ and $\bs u$ of the form
\be
\bs t=\bs m'_{i_1,i_2}=(m_1,\ldots,m_{i_1}-1,\dots,m_{i_2}-1,\dots,m_r),\quad \bs u=\bs n,
\ee
for some $1\leq i_1<i_2\leq r$, or
\be
\bs t=\bs m,\quad\bs u=\bs n'_{j_1,j_2}=(n_1,\ldots,n_{j_1}-1,\dots,n_{j_2}-1,\dots,n_s),
\ee
for some $1\leq j_1<j_2\leq s$, or
\be
\bs t=\bs m'_i,\quad \bs u=\bs n'_j,
\ee
for some $i\in[r]$ and $j\in[s]$, or 
\be
\bs t=\bs m'_{ii}=(m_1,\ldots,m_{i}-2,\dots,m_r),\quad \bs u=\bs n
\ee
for some $i\in[r]$, or, finally,
\be
\bs t=\bs m,\quad \bs u=\bs n'_{jj}=(n_1,\ldots,n_{j}-2,\dots,n_s)
\ee
for some $j\in[s]$.
We note that $\bs m_{i_1,i_2}'$ and $\bs n_{j_1,j_2}'$ might not be in $\mathscr V_r$ and $\mathscr V_s$ (respectively).
When they are in $\mathscr V_r$ and $\mathscr V_s$, the corresponding minors are given by either
\be
\label{eq:case1}
\begin{aligned}
(-1)^{\binom r2}|A_{\bs m,\bs m'_{i_1,i_2}}| &\,=\,m_{i_1}m_{i_2}(m_{i_1}+\beta)(m_{i_2}+\beta),\\
(-1)^{\binom s2}|B_{\bs n,\bs n'_{j_1,j_2}}| &\,=\,n_{j_1}n_{j_2}(n_{j_1}-\beta)(n_{j_2}-\beta),
\end{aligned}
\ee
in case $m_{i_1}\not=m_{i_2}+1$ or $n_{j_1}\not=n_{j_2}+1$, or by
\be
\label{eq:case2}
\begin{aligned}
(-1)^{\binom r2}|A_{\bs m,\bs m'_{i_1,i_2}}| &\,=\,\binom{m_{i_1}}{2}(m_{i_1}+\beta)(m_{i_1}-1+\beta),\\
(-1)^{\binom s2}|B_{\bs n,\bs n'_{j_1,j_2}}| &\,=\,\binom{n_{j_1}}{2}(n_{j_1}-\beta)(n_{j_1}-1-\beta),
\end{aligned}
\ee
in case $m_{i_1}=m_{i_2}+1$ or $n_{j_1}=n_{j_2}+1$.
When they are not in $\mathscr V_r$ and $\mathscr V_s$, they do not contribute to the general formula~\eqref{eq:phiexpression} (due to the Vandermonde factor).
Similarly, $\bs m''_i$ and $\bs n''_j$ might not be in $\mathscr V_r$ and $\mathscr V_s$ (respectively).
When they are in $\mathscr V_r$ and $\mathscr V_s$, the corresponding minors are
\be
\begin{aligned}
(-1)^{\binom r2}|A_{\bs m,\bs m''_{i}}| &\,=\,\binom{m_i}{2}(m_i+\beta)(m_i+\beta-1),\\
(-1)^{\binom s2}|B_{\bs n,\bs n''_{j}}| &\,=\,\binom{n_j}{2}(n_j-\beta)(n_j-\beta-1).
\end{aligned}
\ee
When they are not in $\mathscr V_r$ and $\mathscr V_s$, two cases can happen: either two entries of $\bs m''_i$ coincide (and this case does not contribute to the general formula due to the Vandermonde factor) or $m_{i+1}=m_i-1$ and $m_{i+2}\not=m_i-2$.

Combining all these cases, we see from the general formula~\eqref{eq:phiexpression} that
\be
\begin{aligned}
    \widetilde{\phi}_{\bs m,\bs n}^{[d-2]}(\beta) &= \sum_{1\leq i_1<i_2\leq r}\frac{\Delta(\bs m'_{i_1,i_2}+\frac \beta 2,\bs n-\frac \beta 2)}{\Delta(\bs m+\frac \beta 2,\bs n-\frac \beta 2)}m_{i_1}m_{i_2}(m_{i_1}+\beta)(m_{i_2}+\beta)
    \\
    &\qquad+\sum_{1\leq j_1<j_2\leq s}\frac{\Delta(\bs m+\frac \beta 2,\bs n'_{j_1,j_2}-\frac \beta 2)}{\Delta(\bs m+\frac \beta 2,\bs n-\frac \beta 2)}n_{j_1}n_{j_2}(n_{j_1}-\beta)(n_{j_2}-\beta)
    \\
    &\qquad+\sum_{i=1}^r\sum_{j=1}^s\frac{\Delta(\bs m'_i+\frac \beta 2,\bs n'_{j}-\frac \beta 2)}{\Delta(\bs m+\frac \beta 2,\bs n-\frac \beta 2)}m_i(m_i+\beta)n_j(n_j-\beta)
    \\
    &\qquad+\sum_{i=1}^r\frac{\Delta(\bs m'_{ii}+\frac \beta 2,\bs n-\frac \beta 2)}{\Delta(\bs m+\frac \beta 2,\bs n-\frac \beta 2)}\binom{m_i}{2}(m_i+\beta)(m_i-1+\beta)
    \\
    &\qquad+\sum_{j=1}^s\frac{\Delta(\bs m+\frac \beta 2,\bs n'_{jj}-\frac \beta 2)}{\Delta(\bs m+\frac \beta 2,\bs n-\frac \beta 2)}\binom{n_j}{2}(n_j-\beta)(n_j-1-\beta).
    \end{aligned}
\ee
Here we combine the cases corresponding to~\eqref{eq:case1} and~\eqref{eq:case2} in the first line by noting that when $m_{i_1}=m_{i_2}+1$ we also get a contribution in the penultimate line from $\bs m'_{i_1i_1}$ which is not in $\mathscr V_r$ but for which the formula still makes sense and produces actually the correct counter-term needed in view of the discrepancy between~\eqref{eq:case1} and~\eqref{eq:case2}; similarly, for the $\bs n'_{j_1,j_2}$ with $j_1=j_2+1$.
Hence, denoting $N=r+s$ and $\bs x=\bigl(\bs m+\frac \beta2,\bs n-\frac \beta2\bigr)$ we have (with the notations of Lemma~\ref{lemma:Vander})
\be
\begin{aligned}
    \widetilde{\phi}_{\bs m,\bs n}^{[d-2]}(\beta) &= \sum_{1\leq i<j\leq N}\frac{\Delta(\bs x'_{ij})}{\Delta(\bs x)}\Bigl(x_i^2x_j^2-\frac{\beta^2}{4}(x_i^2+x_j^2)+\frac{\beta^4}{16}\Bigr)
    \\
    &\qquad+\frac 12\sum_{i=1}^N\frac{\Delta(\bs x'_{ii})}{\Delta(\bs x)}\biggl(\Bigl(x_i^2-\frac{\beta^2}{4}\Bigr)\Bigl((x_i-1)^2-\frac{\beta^2}{4}\Bigr)\biggr)
    \end{aligned}
\ee
and the proof is completed by using~\eqref{eq:lemmaVander1}--\eqref{eq:lemmaVander6} and straightforward (though tedious) algebraic manipulations.
\end{proof}

\begin{lemma}
\label{lemma:Vander}
    For any $N\in\mathbb N^*$ and $i\in [N]$, let $\bs e_i=(0,\dots,1,\dots,0)$ be the $N$-dimensional vector with $1$ in the $i$th coordinate and $0$ in the other coordinates.
    For any vector $\bs x=(x_1,\ldots,x_N)$ we denote, for any $i,j\in [N]$, $\bs x'_i=\bs x-\bs e_i$, $\bs x_{ij}'=\bs x-\bs e_i-\bs e_j$, and, for any $k\in\N$, $p_k=\sum_{i=1}^Nx_i^k$.
    Then,
    \begin{align}
        \label{eq:lemmaVander0}
        &\sum_{i=1}^N\frac{\Delta(\bs x_i')}{\Delta(\bs x)} = N,
        \quad\sum_{i=1}^Nx_i^2\frac{\Delta(\bs x_i')}{\Delta(\bs x)} = p_2-(N-1)p_1+\binom{N}{3},
        \\
        \label{eq:lemmaVander1}
        &\sum_{i=1}^N\frac{\Delta(\bs x'_{ii})}{\Delta(\bs x)} = N,
        \quad\sum_{i=1}^Nx_i^2\frac{\Delta(\bs x_{ii}')}{\Delta(\bs x)} = p_2-2(N-1)p_1+4\binom{N}{3},
        \\
        \label{eq:lemmaVander2}
        &\sum_{i=1}^N(x_i-1)^2\frac{\Delta(\bs x_{ii}')}{\Delta(\bs x)} = p_2-2Np_1+\frac{N(2N^2+1)}{3},
        \\
        \nonumber
        &\sum_{i=1}^Nx_i^2(x_i-1)^2\frac{\Delta(\bs x_{ii}')}{\Delta(\bs x)} = p_4-2(N-1)p_3-2p_2p_1+(2N^2-4N+3)p_2
        \\
        \label{eq:lemmaVander3}
        &\qquad+2(N-1)p_1^2-\frac {2}{3}(N-1)(2N^2-4N+3)p_1+\frac{4}{5}(N^2-2N+2)\binom{N}{3},
        \\
        \label{eq:lemmaVander4}
        &\sum_{1\leq i<j\leq N}\frac{\Delta(\bs x_{ij}')}{\Delta(\bs x)} =\binom{N}{2},
        \\
        \label{eq:lemmaVander5}
        &\sum_{1\leq i<j\leq N}(x_i^2+x_j^2)\frac{\Delta(\bs x_{ij}')}{\Delta(\bs x)} =(N-1)p_2-2\binom{N-1}{2}p_1+(N-4)\binom{N}{3},
        \\
        \nonumber
        &\sum_{1\leq i<j\leq N}(x_ix_j)^2\frac{\Delta(\bs x_{ij}')}{\Delta(\bs x)} = -\frac{p_4}{2}+\frac{p_2^2}{2}+(N-2)\bigl(p_3-p_2p_1\bigr)
        \\
        \nonumber
        &\qquad\qquad\qquad +\frac{(N-2)(N^2-7N+9)}{6}p_2+\frac{(N-2)^2}{2}p_1^2
        \\
        \label{eq:lemmaVander6}
        &\qquad\qquad\qquad-(N-2)\binom{N-1}{3}p_1+\frac{5N^2-24N+31}{15}\binom{N}{4}.
    \end{align}
\end{lemma}
\begin{proof}
    Let $g(z)=\prod_{i=1}^N(z-x_i)$.
    It is easy to check that we have
    \be
    \sum_{i=1}^Nx_i^k\frac{\Delta(\bs x'_i)}{\Delta(\bs x)} = -\sum_{i=1}^Nx_i^k\frac{g(x_i-1)}{g'(x_i)}=\res{z=\infty}z^k\frac{g(z-1)}{g(z)}\d z
    \ee
    hence~\eqref{eq:lemmaVander0} follows from the large-$z$ expansion
    \be
    \begin{aligned}
    \frac{g(z-1)}{g(z)} &=\prod_{i=1}^N \biggl(1-\frac 1z-\frac{x_i}{z^2}-\frac{x_i^2}{z^3}-\frac{x_i^3}{z^4}-\frac{x_i^4}{z^5}+O\bigl(z^{-6}\bigr)\biggr)  
    \\
    &=1-\frac Nz -\frac{p_1-\binom N2}{z^2}-\frac{p_2-(N-1)p_1+\binom{N}{3}}{z^3}
    \\
    &\quad-\frac{p_3-\bigl(N-\frac{3}{2}\bigr)p_2-\frac 12p_1^2+\binom{N-1}{2}p_1-\binom{N}{4}}{z^4}
    \\
    &\quad -\frac {p_4-(N-2)p_3-p_2p_1+\frac{(N-2)^2}{2}p_2+\frac {N-2}{2}p_1^2-\binom{N-1}3p_1+\binom{N}{5}}{z^5}
    \\
    &\quad+O(z^{-6}).
    \end{aligned}
    \ee
    Similarly, 
    \be
    \sum_{i=1}^Nx_i^k\frac{\Delta(\bs x'_{ii})}{\Delta(\bs x)} = -\frac 12\sum_{i=1}^Nx_i^k\frac{g(x_i-2)}{g'(x_i)}=\frac 12\res{z=\infty}z^k\frac{g(z-2)}{g(z)}\d z
    \ee
    and so~\eqref{eq:lemmaVander1}--\eqref{eq:lemmaVander3} follow from the large-$z$ expansion of $g(z-2)/g(z)$.
    The latter is obtained immediately from that of $g(z-1)/g(z)$ computed above by the transformations $z\mapsto z/2$ and $x_i\mapsto x_i/2$.
    Finally, the double sums~\eqref{eq:lemmaVander4}--\eqref{eq:lemmaVander6} are obtained by reducing them to single sums as those examined above.
    As an example, we give the proof of~\eqref{eq:lemmaVander5}, the other two being similar.
    First, we have
    \be
    \sum_{1\leq i<j\leq N}\bigl(x_i^2+x_j^2\bigr)\frac{\Delta(\bs x'_{ij})}{\Delta(\bs x)}=
    \sum_{i,j=1}^Nx_i^2\frac{\Delta(\bs x'_{ij})}{\Delta(\bs x)}-\sum_{i=1}^Nx_i^2\frac{\Delta(\bs x'_{ii})}{\Delta(\bs x)}.
    \ee
    In the right-hand side, the last sum has been dealt with above, hence we focus on the first sum: we have
    \be
    \begin{aligned}
    \sum_{i,j=1}^Nx_i^2\frac{\Delta(\bs x'_{ij})}{\Delta(\bs x)}
    =\sum_{i=1}^Nx_i^2\left(\sum_{j=1}^N\frac{\Delta(\bs x'_{ij})}{\Delta(\bs x'_i)}\right)\frac{\Delta(\bs x'_i)}{\Delta(\bs x)}
    =N\sum_{i=1}^Nx_i^2\frac{\Delta(\bs x'_i)}{\Delta(\bs x)}
    \end{aligned}
    \ee
    and we conclude using the second identity in~\eqref{eq:lemmaVander0}.
\end{proof}
\subsection{Proof of Theorem~\ref{thm:equivalence}}\label{sec:proofequivalence}
\begin{proof}[Proof of Theorem~\ref{thm:equivalence}]
It is enough to show that if $\bs m=(m_1,\dots,m_r)\in\mathscr V_r$, $\bs n=(n_1,\dots,n_s)\in\mathscr V_s$, $\wt{\bs m}=(m_1+1,\dots, m_r+1,0)\in\mathscr V_{r+1}$, $\wt{\bs n}=(n_1+1,\dots, n_s+1,0)\in\mathscr V_{s+1}$, we have
\be\label{eq:tobeprovedpartitions}
\widetilde{\Phi}_{\wt{\bs m},\bs n}^{(\beta)}(y)=\widetilde{\Phi}_{\bs m,\bs n}^{(\beta+1)}(y),\quad
\widetilde{\Phi}_{\bs m,\wt {\bs n}}^{(\beta)}(y)=\widetilde{\Phi}_{\bs m,\bs n}^{(\beta-1)}(y).
\ee
Throughout this proof we write 
\be\psi_\pm^{(\beta)}(n)\,=\,\psi_\pm^{(\beta)}(n,x).\ee
By the identity 
\be\label{eq:dLaguerre}
\partial_y L_{m+1}^{(\beta)}(y)=-L_{m}^{(\beta+1)}(y)
\ee
we deduce that
\be
\partial_x\psi_+^{(\beta)}(m+1)=\psi_+^{(\beta)}(m+1)\frac{\partial_x\psi_+^{(\beta)}(0)}{\psi_+^{(\beta)}(0)} \,+\,2(m+1)\psi_+^{(\beta+1)}(m)\,.
\ee
By induction we obtain that for all $k\in\N$ we have
\be
\label{eq:id1}
\partial_x^k\psi_+^{(\beta)}(m+1)=\psi_+^{(\beta)}(m+1)\frac{\partial_x^k\psi_+^{(\beta)}(0)}{\psi_+^{(\beta)}(0)}+2(m+1)\sum_{i=1}^k\partial_x^{i-1}\psi_+^{(\beta+1)}(m)t(k,i)
\ee
where $t(k,i)$ are defined by the recursive relations
\be
t(k+1,i) = \begin{cases}
\partial_xt(k,i)+t(k,i-1) & \mbox{if }i\geq 2,\\
\partial_xt(k,1)+\frac{\partial_x^k\psi_+^{(\beta)}(0)}{\psi_+^{(\beta)}(0)} & \mbox{if }i=1,
\end{cases}
\ee
and the boundary conditions
\be
t(1,1)=1\qquad \mbox{and}\qquad t(k,i)=0\ \mbox{ if }i>k.
\ee
In particular, we note that $t(k,k)=1$ for all $k\geq 1$.

Similarly, by~\eqref{eq:dLaguerre} and the explicit form of $\psi_+(0,\beta)$ we deduce that
\be
\partial_x\psi_-^{(\beta)}(n)=\psi_-^{(\beta)}(n)\frac{\partial_x\psi_+^{(\beta)}(0)}{\psi_+^{(\beta)}(0)} \,+\,2(n-\beta)\psi_-^{(\beta+1)}(n)
\ee
such that, by induction, for all $k\in\N$ we have
\be
\label{eq:id2}
\partial_x^k\psi_-^{(\beta)}(n)=\psi_-^{(\beta)}(n)\frac{\partial_x^k\psi^{(\beta)}_+(0)}{\psi_+^{(\beta)}(0)}+2(n-\beta)\sum_{i=1}^k\partial_x^{i-1}\psi_-^{(\beta+1)}(n)t(k,i),
\ee
in terms of \emph{the same} coefficients $t(k,i)$.

For any $\bs m=(m_1,\dots, m_r)\in\mathscr V_r$ and $\bs n=(n_1,\dots, n_s)\in\mathscr V_s$ we denote by $M_{\bs m,\bs n}^{(\beta)}$ the $(r+s)\times(r+s)$ matrix
\be
   M_{\bs m,\bs n}^{(\beta)}\,=\,
   \left(\begin{smallmatrix}
        \psi_-^{(\beta)}(n_1) & \dots & \psi_-^{(\beta)}(n_s) & \psi_+^{(\beta)}(m_1) & \dots & \psi_+^{(\beta)}(m_r) \\
        \partial_x\psi_-^{(\beta)}(n_1) & \dots & \partial_x\psi_-^{(\beta)}(n_s) & \partial_x\psi_+^{(\beta)}(m_1) & \dots & \partial_x\psi_+^{(\beta)}(m_r)\\
        \vdots & & \vdots & \vdots & & \vdots \\
        \partial_x^{r+s-1}\psi_-^{(\beta)}(n_1) & \dots & \partial_x^{r+s-1}\psi_-^{(\beta)}(n_s) & \partial_x^{r+s-1}\psi_+^{(\beta)}(m_1) & \dots & \partial_x^{r+s-1}\psi_+^{(\beta)}(m_r) \\
    \end{smallmatrix}
    \right)
\ee
and by $\wt M_{\bs m,\bs n}^{(\beta)}$ the $(r+s)\times(r+s)$ matrix
\be
  \wt M_{\bs m,\bs n}^{(\beta)}=M_{\bs m,\bs n}^{(\beta)}\,\cdot\, \mathrm{diag}(n_1-\beta+1,\dots,n_s-\beta+1,m_1+1,\dots,m_r+1)
\ee
such that, by definition, 
\begin{align}
\label{eq:MPsi}
\det M_{\bs m,\bs n}^{(\beta)}&=(-1)^{rs}\Psi_{\bs m,\bs n}^{(\beta)}(x),
\\
\label{eq:wtMPsi}
\det\wt M_{\bs m,\bs n}^{(\beta)}&=(-1)^{rs}\prod_{i=1}^r(m_i+1)\prod_{j=1}^s(n_j-\beta+1)\Psi_{\bs m,\bs n}^{(\beta)}(x).
\end{align}
Identities~\eqref{eq:id1} and~\eqref{eq:id2} imply the matrix factorization
\be
M_{\wt{\bs m},\bs n}^{(\beta)}\,=\,T\,\left(
\begin{array}{c|c}
\psi_-^{(\beta)}(n_1) , \dots,  \psi_-^{(\beta)}(n_s), \psi_+^{(\beta)}(m_1+1), \dots,\psi_+^{(\beta)}(m_r+1) & \psi_+^{(\beta)}(0) 
\\ \hline
2\,\wt M_{\bs m,\bs n}^{(\beta+1)} & \begin{array}{c} 0 \\ \vdots \\ 0\end{array}
    \end{array}\right)
\ee
where
\be
T=\left(\begin{smallmatrix}
        1 & 0 & 0 & \cdots & 0 \\
        \frac{\partial_x\psi_+^{(\beta)}(0)}{\psi_+(0)} & 1 & 0 & \cdots & 0 \\
        \frac{\partial_x^2\psi_+^{(\beta)}(0)}{\psi_+^{(\beta)}(0)} & t(2,1) & 1 & \cdots & 0 \\
        \vdots &  \vdots & \vdots & \ddots & \vdots \\
        \frac{\partial_x^{r+s}\psi_+^{(\beta)}(0)}{\psi_+^{(\beta)}(0)} & t(r+s-1,1)& t(r+s-2,2)&\dots & 1
    \end{smallmatrix}\right).
\ee
Taking determinants we get
\be
\det M_{\wt{\bs m},\bs n}^{(\beta)}\,=\,(-2)^{r+s}\,\psi_+^{(\beta)}(0)\,\det \wt M_{\bs m,\bs n}^{(\beta+1)}
\ee
and so, using $\psi_+^{(\beta)}(0)=x^{\frac 12+\beta}\e^{-\frac 12 x^2}$ and~\eqref{eq:MPsi}--\eqref{eq:wtMPsi},
\be
(-1)^{s}\Psi_{\wt{\bs m},\bs n}^{(\beta)}(x)\,=\,(-2)^{r+s}x^{\frac 12+\beta}\e^{-\frac 12 x^2}\prod_{i=1}^r(m_i+1)\prod_{j=1}^s(n_j-\beta)\Psi_{\bs m,\bs n}^{(\beta+1)}(x).
\ee
It is readily verified that
\be
\kappa_{\wt{\bs m},\bs n}^{(\beta)}\,=\,\kappa_{\bs m,\bs n}^{(\beta+1)}\,2^{r+s}\,\prod_{i=1}^r(-m_i-1)\,\prod_{j=1}^s(n_j-\beta),
\ee
hence, by~\eqref{eq:PsiPhi} we get the first identity in~\eqref{eq:tobeprovedpartitions}.
Since $\widetilde{\Phi}_{\bs m,\bs n}^{(\beta)}(y)=\widetilde{\Phi}_{\bs n,\bs m}^{(-\beta)}(y)$, the second identity in~\eqref{eq:tobeprovedpartitions} is implied by the first one.
\end{proof}

\subsection{Proof of Theorem~\ref{thm:symmetry}}\label{sec:proofsymmetry}

For any $K\in\mathbb N$, we introduce the $(2K)\times (2K)$ matrix $\mathbb L^{(\beta)}_K(y)$ whose entries are, for $i,j\in [2K]$,
\be
\bigl(\mathbb L^{(\beta)}_K(y)\bigr)_{i,j}\,=\,\begin{dcases}
y^{-\frac\beta 2}\partial_y^{i-1}\bigl(y^{\frac\beta 2}L_{j-1}^{(\beta)}(y)\bigr)&\mbox{if }j\leq K\,,
\\
y^{\frac\beta 2}\partial_y^{i-1}\bigl(y^{-\frac\beta 2}L_{j-K-1}^{(-\beta)}(y)\bigr)&\mbox{if }j\geq K+1\,.
\end{dcases}
\ee
Let us also introduce the anti-diagonal square matrix $\mathbb S_{K}$ of size $(2K)\times (2K)$ with entries $(\mathbb S_{K})_{i,j}=\delta_{i+j,2K+1}$ and the diagonal matrix $\mathbb X_K^{(\beta)}$ whose entries are, for $i,j\in [2K]$,
\be
\bigl(\mathbb X_K^{(\beta)}\bigr)_{i,j}\,=\,\begin{cases}
\delta_{ij}/(\beta-K+i)_K  &\mbox{if }i\leq K\,,
\\
\delta_{ij}/(-\beta-2K+i)_K &\mbox{if }i\geq K+1\,.
\end{cases}
\ee

\begin{lemma}
    For any $K\in\mathbb N$ the matrix
    \be
    \mathbb U_K^{(\beta)}(y)\,=\,(-1)^{K-1}y^{K}\,\mathbb L_K^{(\beta)}(y)\,\mathbb X_K^{(\beta)}\,\mathbb S_{K}\,\mathbb L^{(\beta)}_K(-y)^\top\,\mathbb S_{K}
    \ee
    is a lower triangular unipotent matrix, i.e., we have $\bigl(\mathbb U_N^{(\beta)}(y)\bigr)_{i,j}\,=\,\delta_{ij}$ if $i\leq j$.
\end{lemma}
\begin{proof}
Let us introduce the operator $\mathcal D^{(\beta)}_y=\partial_y+\frac{\beta}{2y}=y^{-\frac\beta 2}\partial_yy^{\frac\beta 2}$ and
\be
\mathscr R^{(\beta)}(y,z)\,=\,\sum_{l=1}^K\frac{L_{l-1}^{(\beta)}(y)\,L_{K-l}^{(-\beta)}(z)}{(\beta-K+l)_K}\,.
\ee
Using the block-structure of the matrices involved, it is straightforward to see that $\mathbb U_{ij}^{(\beta)}(y)=(-1)^{K-1}y^K\bigl(\wt U_{ij}(\beta)+\wt U_{ij}(-\beta)\bigr)$ where
\be
\wt U_{i,j}(\beta)\,=\,\biggl(\bigl(\mathcal D_z^{(-\beta)}\bigl)^{2K-j}\,\bigl(\mathcal D_y^{(\beta)}\bigl)^{i-1} \mathscr R^{(\beta)}(y,z)\biggr)\bigg|_{z=-y}\,.
\ee
Since $\frac 1{(\beta-K+l)_K}=\frac{\Gamma(\beta-K+l)}{\Gamma(\beta+l)}$, we have
\be
\mathscr R^{(\beta)}(y,z)\,=\,[t^{K-1}]\biggl( \mathscr F^{(\beta)}(y,t)\,\mathscr G^{(\beta)}(z,t)\biggr)
\ee
where
\be
\mathscr F^{(\beta)}(y,t) \,=\,\sum_{m=0}^\infty t^m \,\frac{L^{(\beta)}_m(y)}{\Gamma(\beta+m+1)}\,,\quad
\mathscr G^{(\beta)}(z,t)\,=\,\sum_{m=0}^\infty t^m \,\Gamma(\beta-m)\,L^{(-\beta)}_m(z)\,.
\ee
A direct computation using the explicit formula~\eqref{eq:LaguerreExplicit} for Laguerre polynomials and the hypergeometric series $_0\mathrm F_1$ shows that
\be
\mathscr F^{(\beta)}(y,t)\,=\,\frac{\e^t}{\Gamma(\beta+1)}\,_{0}\mathrm{F}_{1}\bigl(;1+\beta;-yt\bigr)\,,
\quad
\mathscr G^{(\beta)}(z,t)\,=\,\Gamma(\beta)\,\e^{-t}\,_{0}\mathrm{F}_{1}\bigl(;1-\beta;zt\bigr)\,.
\ee
Therefore, denoting $F_\pm(\zeta)=\,_0\mathrm{F}_1(;1\pm\beta;\zeta)$,
\be
\begin{aligned}
\wt U_{ij}(\beta)\,&=\,\frac 1\beta \biggl(\bigl(\mathcal D_z^{(-\beta)}\bigl)^{2K-j}\,\bigl(\mathcal D_y^{(\beta)}\bigl)^{i-1} [t^{K-1}]\,\bigl(F_+(-yt)\,F_-(zt)\bigr)\biggr)\bigg|_{z=-y}
\\
&=\,\frac {(-1)^{K-j-1}y^{-K+j-i}}\beta [t^{-K-i+j}]\bigl((\mathcal D_t^{(\beta)})^{i-1}F_+(t)\bigl)\,\bigl((\mathcal D_t^{(-\beta)})^{2K-j}F_-(t)\bigl).
\end{aligned}
\ee
By repeated applications of the hypergeometric differential equation
\be
\bigl(\mathcal D^{(\pm\beta)}_t\bigr)^2\,F_\pm(t)\,=\,\biggl(-\frac 1t\mathcal D^{(\pm\beta)}_t+\Bigl(\frac 1t+\frac{\beta^2}{4t^2}\Bigr)\biggr)\,F_\pm(t)
\ee
we obtain
\be
\bigl(\mathcal D^{(\pm\beta)}_t\bigr)^n\,F_\pm(t)\,=\,\bigl(p_n(t,\beta)+q_n(t,\beta)\mathcal D_t^{(\pm\beta)}\bigr)\,F_\pm(t)
\ee
where $p_n$ and $q_n$ are Laurent polynomials in $t$ with coefficients which are \textit{even} polynomials of $\beta$, determined explicitly by the recursion
\be
p_{n+1}\,=\,\partial_t p_n+\Bigl(\frac 1t+\frac{\beta^2}{4t^2}\Bigr)q_n,\quad
q_{n+1}\,=\,\partial_t q_n+p_n-\frac 1tq_n,
\ee
with initial conditions $p_0=1$, $p_1=0$, $q_0=0$, $q_1=1$. In particular, it can be shown by induction on $n$ that, defining $\wt p_n=t^np_n$ and $\wt q_n=t^{n-1}q_n$,
\begin{itemize}[leftmargin=*]
    \item $\wt p_n$ is a polynomial in $t$ of degree $\lfloor \frac n2\rfloor$ for all $n$, with leading coefficient $1$ if $n$ is even, and
    \item $\wt q_n$ is a polynomial in $t$ of degree $\lfloor \frac{n-1}2\rfloor$ for all $n$, with leading coefficient $1$ if $n$ is odd.
\end{itemize}
Finally, making use of the Wronskian identity
\be
\bigl(\mathcal D^{(\beta)}_tF_+(t)\bigr)\,F_-(t) - F_+(t)\,\bigl(\mathcal D^{(-\beta)}_tF_-(t)\bigr)=\,\frac \beta t
\ee
and of the fact that $p_n,q_n$ are even in $\beta$, we obtain
\be
\label{eq:finallemmafinal}
\begin{aligned}
\wt U_{ij}(\beta)+\wt U_{ij}(-\beta)&\,=\,
(-1)^{K-j-1}y^{-K+j-i} [t^{-K-i+j+1}]\bigl(q_{i-1}p_{2K-j}-p_{i-1}q_{2K-j}\bigr)
\\
&\,=\,(-1)^{K-j-1}y^{-K+j-i} [t^{K-1}]\bigl(\wt q_{i-1}\wt p_{2K-j}-\wt p_{i-1}\wt q_{2K-j}\bigr)\,.
\end{aligned}
\ee
On the other hand, $\wt q_{i-1}\wt p_{2K-j}-\wt p_{i-1}\wt q_{2K-j}$ is a polynomial in $t$ of degree $\leq k-1+\frac{i-j}2$ hence if $i<j$ we have 
$\wt U_{ij}(\beta)+\wt U_{ij}(-\beta)=0$.
When $i=j$, the statement follows from~\eqref{eq:finallemmafinal} by distinguishing two cases according to the parity of $i=j$ and by the above mentioned properties of the two families of polynomials $\wt q_n$ and $\wt p_n$, which imply the following further properties:
\begin{itemize}[leftmargin=*]
    \item if $i=j$ is odd, then $i-1$ is even and $2K-j$ is odd, hence $\wt q_{i-1}\wt p_{2K-j}$ has degree $K-2$ and $-\wt p_{i-1}\wt q_{2K-j}$ has degreee $K-1$ and leading coefficient $-1$, and
    \item if $i=j$ is even, then $i-1$ is odd and $2K-j$ is even, hence $\wt q_{i-1}\wt p_{2K-j}$ has degree $K-1$ and leading coefficient $1$ and $-\wt p_{i-1}\wt q_{2K-j}$ has degreee $K-2$.
\end{itemize}
\end{proof}

We will also need the following well-known fact from linear algebra.

\begin{lemma}
\label{lemma:linearalgebra}
    Let $1\leq p\leq q$ be integers and let $I,J$ be subsets of~$[q]$ of cardinality~$p$.
    For any square matrix $M$ of size $q$ we have
    \be
    \det ((M^{-1})_{I,J}) = \frac 1{\det M} (-1)^{\sum_{i\in I}i+\sum_{j\in J}j}\det (M_{[q]\setminus J,[q]\setminus I})
    \ee
    where $A_{I,J}$ is the submatrix of $A$ with rows indexed by $I$ and columns indexed by $J$.
\end{lemma}

Finally, we can prove Theorem~\ref{thm:symmetry}.

\begin{proof}[Proof of Theorem~\ref{thm:symmetry}]
We start by choosing any $r\in\mathbb N$ such that $r\geq \ell(\lambda),\ell(\mu),\lambda_1,\mu_1$ and forming the vectors $\bs m=\lambda+\bs\delta_r$, $\bs n=\mu+\bs\delta_r$, $\bs m'=\lambda'+\bs\delta_r$, $\bs n'=\mu'+\bs\delta_r$ such that, by Definition~\ref{def:PhiLaMu}, $\Phi_{\lambda,\mu}^{(\beta)}(y)=\wt\Phi_{\bs m,\bs n}^{(\beta)}(y)$ and $\Phi_{\mu',\lambda'}^{(\beta)}(y)=\wt\Phi_{\bs n',\bs m'}^{(\beta)}(y)$.
Let us also introduce the following subsets of $[2r]$: $M=\lbrace m_1,\dots,m_r\rbrace+1$, $N=\lbrace n_1,\dots,n_r\rbrace+1$, $M'=\lbrace m_1',\dots,m_r'\rbrace+1$, and $N=\lbrace n_1',\dots,n_r'\rbrace+1$.
It is well-known (see, for example, \cite[Chapter~1]{Macdonald}) that $M'=2r+1-M''$ and $N'=2r+1-N''$ where $M''=[2r]\setminus M$ and $N''=[2r]\setminus N$.
The Wronskian
\be
\label{eq:wronskianref}
\Wr y \bigl(y^{\frac\beta 2}L_{m_1}^{(\beta)}(y),\dots,y^{\frac\beta 2}L_{m_r}^{(\beta)}(y),y^{-\frac\beta 2}L_{n_1}^{(-\beta)}(y),\dots,y^{-\frac\beta 2}L_{n_r}^{(-\beta)}(y)\bigr)
\ee
is equal, by definition, to the $(2r)\times (2r)$ determinant $\det \bigl(\mathbb L_{2r}^{(\beta)}(y)\bigr)_{[2r],M\cup (N+2r)}$, where $A_{I,J}$ is the submatrix of $A$ corresponding to the rows indexed by $I$ and columns indexed by $J$.
Hence, by using Lemma~\ref{lemma:linearalgebra} and the fact that $\mathbb U_{2r}^{(\beta)}$ is a lower triangular unipotent matrix, the Wronskian~\eqref{eq:wronskianref} can be expressed as
\be\begin{aligned}
\det \bigl(\mathbb L_{2r}^{(\beta)}(y)\bigr)_{[2r],M\cup (N+2r)}\,&=\,\det \Bigl((\mathbb U_{2r}^{(\beta)})^{-1}\,\mathbb L_{2r}^{(\beta)}(y)\Bigr)_{[2r],M\cup (N+2r)}
\\
&=\,\pm \det\mathbb L_{2r}^{(\beta)}(y)\,\det \bigl(\mathbb L_{2r}^{(\beta)}(y)^{-1}\,\mathbb U_{2r}^{(\beta)}\bigr)_{M''\cup (N''+2r),[4r]\setminus[2r]}
\\
&=\,\pm\det\mathbb L_K^{(\beta)}(y)\,\det\bigl( y^{2r}\,\mathbb X_{2r}^{(\beta)}\,\mathbb S_{2r}\,\mathbb L^{(\beta)}_{2r}(-y)^\top\,\mathbb S_{2r}\bigr)_{M''\cup (N''+2r),[4r]\setminus[2r]}
\end{aligned}
\ee
where we do not specify the sign $\pm 1$ (which might differ across the equality sign).
Now we observe that $\det\mathbb L_{2r}^{(\beta)}(y)$ is equal to a constant independent of $y$ times a power of $y$, by Theorems~\ref{thm:propertiesWronskianLaguerre} and~\ref{thm:equivalence}.
Therefore, also exploiting the fact that $\mathbb X^{(\beta)}_{2r}$ is diagonal and independent of $y$, the Wronskian~\eqref{eq:wronskianref} is, up to a constant independent of $y$ and a power of $y$,
\be
\det\bigl(\mathbb S_{2r}\,\mathbb L^{(\beta)}_{2r}(-y)^\top\,\mathbb S_{2r}\bigr)_{M''\cup (N''+2r),[4r]\setminus[2r]}\,.
\ee
Now, for all $K\in\mathbb N$ and $i,j\in[2K]$ we have
\be
\bigl(\mathbb S_{K}\,\mathbb L^{(\beta)}_K(z)^\top\,\mathbb S_{K}\bigr)_{i,j}\,=\,\begin{cases}
z^{\frac\beta 2}\partial_z^{2K-j}\bigl(z^{-\frac\beta 2}L_{K-i}^{(-\beta)}(z)\bigr)&\mbox{if }i\leq K\,,
\\
z^{-\frac\beta 2}\partial_z^{2K-j}\bigl(z^{\frac\beta 2}L_{2K-i}^{(\beta)}(z)\bigr)&\mbox{if }i\geq K+1\,.
\end{cases}
\ee
and so $\det\bigl( \mathbb S_{2r}\,\mathbb L^{(\beta)}_{2r}(z)^\top\,\mathbb S_{2r}\bigr)_{M''\cup (N''+2r),[4r]\setminus[2r]}$ 
coincides, up to a sign, with the Wronskian
\be
\label{eq:wronskianref2}
\Wr z \bigl(z^{\frac\beta 2}L_{n'_1}^{(\beta)}(z),\dots,z^{\frac\beta 2}L_{n'_r}^{(\beta)}(z),z^{-\frac\beta 2}L_{m'_1}^{(-\beta)}(z),\dots,z^{-\frac\beta 2}L_{m'_r}^{(-\beta)}(z)\bigr).
\ee
Hence we have shown that the two Wronskians~\eqref{eq:wronskianref} and~\eqref{eq:wronskianref2} coincide up to a constant independent of $y$ when $z=-y$. 
The structure of these Wronskians implied by Theorem~\ref{thm:propertiesWronskianLaguerre} forces the two monic polynomials $\wt\Phi_{\bs m,\bs n}^{(\beta)}(y)$ and $(-1)^{|\lambda|+|\mu|}\wt\Phi_{\bs n',\bs m'}^{(\beta)}(-y)$ to coincide.
\end{proof}

\end{document}